\documentclass[aps,prc, groupedaddress,tightenlines,notitlepage,nofootinbib,onecolumn,superscriptaddress, longbibliography, 11pt]{revtex4-1}

\usepackage{amsthm,amsmath,amssymb,amsfonts,mathtools,physics,bm,bbm, enumitem, graphicx, epsfig, multirow}

\renewcommand{\thetable}{\arabic{table}}

\usepackage{tabularray}
\DefTblrTemplate{caption-tag}{default}{Table\hspace{0.25em}\thetable}

\usepackage[caption = false]{subfig}
\usepackage[colorlinks=true, citecolor=blue, urlcolor=blue, linkcolor = blue]{hyperref}
\usepackage{tcolorbox}
\usepackage{xcolor}
\definecolor{mycolour1}{HTML}{56C1FF}

\newcommand{\overbar}[1]{\mkern 2mu\overline{\mkern-2mu#1\mkern-2mu}\mkern 2mu}
\newcommand{\fitB}[2]{\mathcal{A}^{(0)}_{#1}(#2)}
\newcommand{\fitBold}[2]{A^{(0)}_{#1}(#2)}
\newcommand{\fitC}[2]{\mathcal{B}_s^{(#1)}(#2)}
\newcommand{\fitCold}[2]{B^{(#1)}(#2)}
\newcommand{\fitCs}[2]{\mathcal B_s^{(#1)}(#2)}
\newcommand{\bone}{\beta^{(1)}_{\text{crit}}}
\newcommand{\btwo}{\beta^{(2)}_{\text{crit}}}
\newcommand{\bcrit}{\beta_{\text{crit}}}
\newcommand{\xcrit}{x_{\text{crit}}}
\newcommand{\Binom}{\mathsf{Binomial}}
\newcommand{\unif}{\mathsf{Uniform}}
\newcommand{\efvar}{V_{\text{eff}}}

\DeclarePairedDelimiterX\innerp[1]{\langle}{\rangle}{#1}

\DeclareOldFontCommand{\rm}{\normalfont\rmfamily}{\mathrm}
\newtheorem{theorem}{Theorem}

\newtheorem{lemma}{Lemma}
\newtheorem*{example*}{Example}%
\newtheorem*{remark*}{Remark}%

\begin{document}

\title{Tipping points in fitness landscape of heterogeneous populations}


\author{Sumana Bhattacharyya}
\email{bsumana2021@gmail.com}
\affiliation{Physics of Living Matter Group, Department of Physics and Materials Science, University of Luxembourg, 162 A, Avenue de la Faïencerie, L-1511, Luxembourg City, Luxembourg}
\author{Uttam Singh}
\affiliation{Centre for Quantum Science and Technology, International Institute of Information Technology Hyderabad, Hyderabad-500032, Telangana, India}
\author{Anupam Sengupta}
\email{anupam.sengupta@uni.lu}
\affiliation{Physics of Living Matter Group, Department of Physics and Materials Science, University of Luxembourg, 162 A, Avenue de la Faïencerie, L-1511, Luxembourg City, Luxembourg}
\affiliation{Institute for Advanced Studies, University of Luxembourg, 2, Avenue de l’Université, L-4365, Esch-sur-Alzette, Luxembourg}
\begin{abstract}

Predicting fitness of biologically-active populations, communities or systems in fluctuating environments is a long-standing challenge. Phenotypic plasticity and bet-hedging strategy–two key evolutionary traits living systems harness to optimize fitness in dynamic environments–have been widely reported yet how interplays therein could mediate fitness landscapes of heterogeneous populations remain unknown. Leveraging the financial \textit{asset pricing} model, here we provide a dynamical framework for fitness of heterogeneous populations, underpinned by the interrelations between sub-populations exhibiting phenotypic plasticity and bet-hedgeding. Our framework, independent of the definition of fitness, employs a nonlinear difference equation to present fitness dynamics, and capture the emergence of tipping points, marking the onset of critical state transitions which lead to catastrophic shifts. This study identifies limits on the selective advantage conferred by bet-hedging through reduction in the temporal variance of fitness, with far-reaching ramifications on our current understanding of hedging-mediated fitness enhancement of a population. The lower bound of the effective fitness variance is set by a maximum number of bet-hedgers, beyond which the fitness landscape approaches critical transition, as confirmed by critical slowing down in the vicinity of tipping points. We estimate the scaling law for the critical slowing down numerically and derive the characteristic recovery time for heterogeneous populations. Taken together, our work provides a generic theoretical framework to quantify fitness dynamics and predict critical transitions in heterogeneous populations. The results can be extended further to model fitness landscapes of natural and synthetic multi-species consortia exposed to environmental fluctuations mimicking climatic shifts and immunopathological settings.

\end{abstract}

\maketitle

\section{Introduction}
\label{sec:introduction}


Living systems inhabit dynamic environments, wherein fluctuations across multiple orders of length and time scales, shape organisms' response, adaptation and resilience~\cite{SCHEFFERbook, botero_evolutionary_2015, dakos_ecosystem_2019}. A vast majority of organisms remain perpetually exposed to changes which are unpredictable, driving the evolution of diverse responsive and adaptive traits which mediate optimal survival and succession, and ultimately, maximize population fitness \cite{Kussell2005b, ackermann2015, sengupta2017, Morawska2022}. In light of the current climatic shifts, the emergence of adaptive traits has received a fresh perspective, specifically by highlighting the diversity and rapidity with which some of these evolutionary responses engage to mitigate the unforeseen impacts of environmental perturbations~\cite{Hoffmann2011, thomas2012, Moritz2013, carrara2021, sengupta2022}. 

The evolutionary traits manifested in fluctuating environments fall into one of the three broad categories of adaptive tracking, phenotypic plasticity and bet-hedging strategy~\cite{botero_evolutionary_2015}. While adaptive tracking enhances population fitness in highly unpredictable but slowly fluctuating environmental shifts via gradual selection of an altered phenotype over other sub-optimal phenotypes~\cite{barrett_adaptation_2008, simons_modes_2011}, phenotypic plasticity mediates optimal fitness in rapidly evolving yet predictable environments, with diverse phenotypes being expressed over the course of the environmental states~\cite{de1996evolutionary}. Depending on the perturbation timescale relative to that of the population (for instance, the doubling time), phenotypic adjustments may be reversible (in case of short-lived fluctuations) or irreversible when fluctuations persist throughout the lifespan~\cite{Eberhard2003, Ratikainen2019}. However, certain adaptive strategies could lag in following the environmental cues, for which stochastic responses such as bet-hedging may play a dominant role in ensuring optimal fitness.  

Evolutionary bet-hedging, a risk spreading strategy, signifies a trade-off between the mean fitness and the temporal variance of fitness, that allows phenotypes with reduced arithmetic mean fitness a selective advantage under fluctuating and unpredictable environmental conditions \cite{philippi_hedging_1989, ripa2010}. Bet-hedgers are isogenic populations, which minimize fitness variance across all possible environmental conditions by either: (a) producing a generalist phenotype which is relatively but not critically sub-optimal, referred to as a conservative bet-hedging strategy (CBH)~\cite{simons_suboptimal_2003} or (b) diversifying phenotypes to potentially suit multiple environmental scenarios in a probabilistic manner, known as diversification bet-hedging (DBH)~\cite{philippi_hedging_1989, Morawska2022}. Despite the fitness cost due to maladapted individuals emerging in a bet-hedging population, the strategy allows populations to thrive in dynamic settings \cite{Acar2008, Veening2008, Beaumont2009,Grimbergen2015, Carey2019, villa2019}. DBH has been recently reported during bacterial diauxic shift~\cite{solopova_bet-hedging_2014}, metabolically mediated heterogeneity~\cite{kotte_phenotypic_2014}, stochastic sporulation by integrating noise in regulatory components~\cite{setlow2007will}, as well response to hydrodynamic cues \cite{sengupta2017,jin2024,hubert2024} and nutrient landscapes \cite{Gasperotti2020,sengupta2022,dinezio2023}. Although the role of bet-hedging has gained traction in diverse systems and environmental contexts, distinguishing bet-hedging from alternate evolutionary responses still remains a challenge, often due to inconclusive datasets~\cite{Jong2011, simons_modes_2011, Morawska2022}. 

In realistic scenarios where both predictable and unpredictable environmental variations occur, a combination of adaptive and stochastic response pathways is harnessed by living systems exposed to such fluctuations~\cite{Xue2019}. Depending on the fluctuation attributes~\cite{Wong2005}, available nutrient resources, and associated metabolic costs~\cite{Relyea2002}, sub-populations exhibiting phenotypic plasticity and bet-hedging emerge, with DBH occurring around norms of reaction~\cite{simons_modes_2011, Wong2005}. Hybrid response strategies have been studied both empirically~\cite{Gilbert2013, Furness2015, Grantham2016} as well theoretically~\cite{Wong2005, Donaldson2013, Reed2010, Draghi2023, Joschinski2020} to understand if organisms employ bet-hedging strategies as a mechanism of natural selection, and how the norm of reaction leads to bet-hedging strategies mediated by complex interactions between environmental variability, population and individual cues, and the emergent traits triggered by relevant cues~\cite{Donaldson2013, Reed2010, Draghi2023, Joschinski2020}. 
Yet, the number of bet-hedgers in an isogenic population--a key determinant of a successful hybrid strategy--has received little attention, thus leaving a major gap in our understanding of such strategies. Specifically, how the number of bet-hedgers shapes the fitness dynamics of a population eliciting hybrid response so far remains unknown. 

One might expect that the aid of bet-hedgers in a heterogeneous population can potentially postpone or entirely prevent the onset of tipping points; however, it comes at the significant cost of maladaptation~\cite{starrfelt_bet-hedgingtriple_nodate, haaland2019short, Libby2019}, as natural selection leads to suboptimal phenotypes by minimizing the temporal variance of fitness to maximize the long-run geometric mean fitness across generations~\cite{Seger1987, Gillespie1974, Kussell2005}. This emerging stability-cost trade-off leads to detrimental consequences if the diversified bet-hedgers cross a certain number, positing a theoretical lower bound on fitness variance. Beyond this threshold, the likelihood of the fitness dynamics triggering catastrophic shifts in the system significantly increases.

Here, we model the fitness dynamics of a hybrid responsive system as a discrete time series, mediated by the number of bet-hedgers in the system, and provide a dynamical model of fitness that is independent of the chosen metric of fitness (growth rate, carrying capacity, risk avoidance etc.). Our model, inspired by the \textit{asset pricing} model with heterogeneous beliefs used in the world of finance~\cite{Brock2009}, analyses isogenic populations subject to environmental variations comprising sub-populations with phenotypic plasticity and diversified bet-hedging strategies. Starting from an initial fitness, we predict the population fitness at a later time, which is determined by the environmental states. We analyze sub-populations whose fitness dynamics depend stochastically on the environmental states and, in turn, may express heterogeneity in response to the environmental cues by exhibiting discrete phenotypic states. Another sub-population employing the DBH strategy will contribute to fitness only if it is specialized for the present environmental state; individuals specialized in other environmental states are likely to remain dormant. Now, when this system evolves dynamically, it may encounter a critical point due to the inherent non-linearities in the system, driving the population to potentially undesirable fitness states. By simulating various scenarios, we uncover how the presence of bet-hedgers affect the criticality of the system, marked the emergence of tipping points. Using this generic model, we demonstrate that the presence of additional bet-hedgers, beyond a maximum number, may anticipate tipping points in such dynamical systems. Furthermore, we show that, for a fixed set of parameters, an upper bound on the number of bet-hedgers imposes the lower threshold on the effective variance of fitness; with additional bet-hedgers pushing system's dynamics toward tipping points. Finally, by studying the dynamics in the vicinity of the tipping points, we extract the power law scaling, and characteristic recovery timescales for such adaptive dynamical systems. Our results highlight a complex interplay between environmental cues, responsive traits and stochastic switching, which are crucial for understanding ecosystem dynamics and resilience; while this framework serves as a foundational step towards theoretically predicting the bifurcation parameters in phenotypically heterogeneous populations exposed to environmental fluctuations. 



The paper is organized as follows: In Sec.~\ref{sec:model}, we model our system of interest and present a steady-state fitness dynamics as a nonlinear third order difference equation. We linearize the fitness dynamics in Sec. \ref{sec:analytic} where the population comprises bet-hedgers and two phenotypically distinct sub-populations. In Sec.~\ref{sec:result} we present the numerical results on the fitness dynamics and analyze the effect of bet-hedger numbers on the dynamical landscape of fitness. In particular, we show that the tipping points emerge earlier when additional bet-hedgers are introduced in the system (Secs.~\ref{result1}-\ref{result2}). In Sec.~\ref{result3}, we derive the exponent for the critical slowing down near the tipping points, and show that this is unaffected by the change in the number of bet-hedgers. Furthermore, in Sec.~\ref{result4} we present a cut-off value on the number of bet-hedgers and the effective variance of fitness beyond which the system approaches the tipping points. Finally, we present our conclusions and offer perspectives in Sec.~\ref{sec:discussion}. Additional Appendices offer a glossary of symbols, supplementary plots, and proofs supporting some of our results presented in the main part.

\section{\label{sec:model} Fitness dynamics in varying environments}

\subsection{Modelling fitness of a population}
\label{sec:fitness-model}

An isogenic population experiencing a dynamic environment with predictable and unpredictable variations can transition between $S\geq 2$ states with a time-independent probability $\{q_{s}\}_{s=1}^{S}$, where $q_s\geq 0$ for all $s\in \{1,\cdots,S\}$ and $\sum_{s=1}^Sq_s=1$. Individuals within the population may also exhibit phenotypic heterogeneity in response to environmental variations. Thus we consider, at a given time $t$, the heterogeneous population comprises sub-populations of fitness types $\mathsf{A}, \mathsf{B}$ and $\mathsf{C}$ (Fig.~\ref{fig:enter-label}):
\begin{enumerate}[label={(\roman*)}]
    \item Type $\mathsf{A}$ exhibits phenotypic plasticity as an adaptive response to environmental variations, the fitness of which may stochastically depend on the environmental states. 
    \item Type $\mathsf{B}$ follows a diversified bet-hedging strategy, where small fractions of the population are preadapted (stochastic response to environmental variations) to dormant metabolic states, resulting in slower growth but higher resilience.
    \item Type $\mathsf{C}$ remains unaffected by environmental states, and in the absence of competition, this sub-population maintains a constant growth rate, $R$.
\end{enumerate}

\begin{figure}
    \centering
    \includegraphics[width=0.9\linewidth]{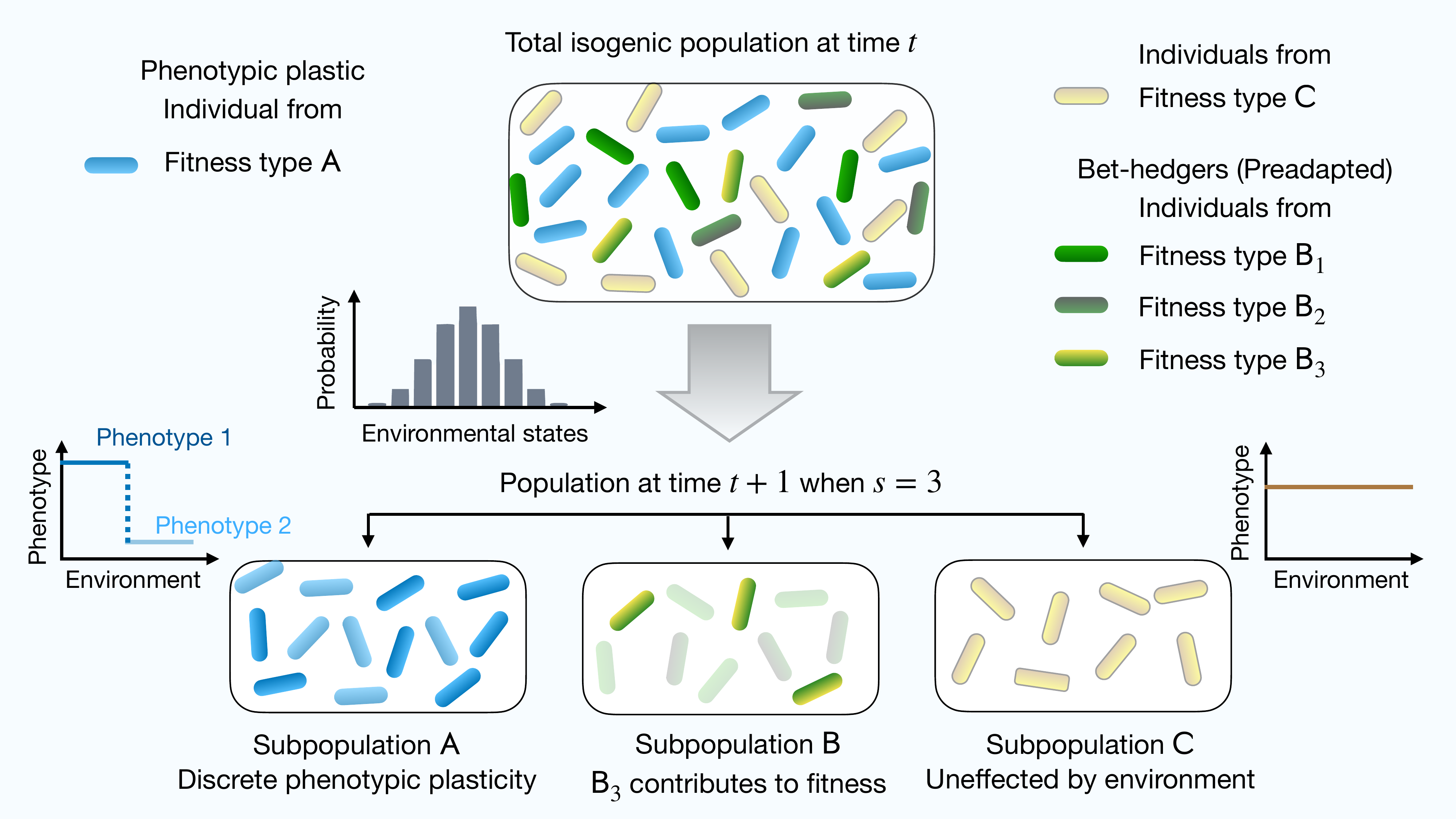}
    \caption{\textbf{Dynamical framework for fitness of heterogeneous population comprising bet-hedgers and phenotypically distinct sub-populations.} Schematic representation of our isogenic population system at time $t$, comprising three fitness types: Type $\mathsf{A}$ (blue) for which fitness probabilistically depends on environmental variation. Consequently, $\mathsf{A}$ may exhibit phenotypic plasticity when exposed to certain environmental cues. Type $\mathsf{B}$ (green), represented by three distinct bet-hedger subtypes $\mathsf{B}_1$, $\mathsf{B}_2$, and $\mathsf{B}_3$, denoted by different shades of green, employs a diversified bet-hedging strategy. Since the bet-hedging strategy is a stochastic response at time $t$, the presence of all subtypes of bet-hedgers does not depend on the environmental states. Type $\mathsf{C}$ (beige) does not respond to environmental variations; its contribution to total fitness is determined solely by its growth rate. Upon exposure to environmental changes, fitness type $\mathsf{A}$ bifurcates into two phenotypic states, each associated with distinct fitness values. In contrast, only one subtype, $\mathsf{B}_i$, where $i=s$, contributes to fitness. In the diagram, $s=3$, so $\mathsf{B}_3$ is the contributing bet-hedger subtype.}
    \label{fig:enter-label}
\end{figure}

If the fitness of one individual from the type $\mathsf{A}$ subpopulation at time $t$ is given by $\fitBold{}{t}\geq 0$, and the environment is assumed to be in a state $s$ at time $t+1$, then the fitness of that individual at time $t+1$ is defined as
\begin{align}
    \fitB{s}{t+1}=\fitBold{}{t+1} + \alpha_s,
\end{align}
where $\alpha_s$ is the contribution to the fitness because of environmental variations which respects the condition $\fitB{s}{t+1} \geq 0$. 

Now, let the type $\mathsf{B}$ further consist of $n\leq S-2$ subtypes with each following a diversified bet-hedging strategy. The fitness of an individual from $i^{\text{th}}$ subtype at time $t$ is defined as $\fitCold{i}{t}$. At time $t+1$, if the environment is in state $s$, the fitness of that individual from $i^{\text{th}}$ subtype of $\mathsf{B}$ is given by
\begin{align}
    \fitCs{i}{t+1}=\delta_{is},
\end{align}
where $i\in\{1,\cdots, n\}$, $s\in \{1,\cdots,S\}$ and $\delta_{is}$ is the kronecker delta which is $1$ if $i=s$ and 0 otherwise. Throughout the paper, we refer $n$ as the number of bet-hedger subtypes present in the system.

The dynamical fitness scenario presented by this framework is comparable to the gain/loss dynamics in the context of financial markets ~\cite{Brock1998, Brock2009}. Building up on this, we derive the steady-state fitness dynamics for the sub-population $\mathsf{A}$, and in particular, present a dynamical theory of fitness of $\mathsf{A}$ in the presence of bet-hedgers, following the three assumptions listed below:

\begin{tcolorbox}[colback=mycolour1!5,colframe=mycolour1!5,sharp corners]
\begin{enumerate}[label=(A\arabic*)]
    \item \label{assump-1} The number of bet-hedgers, $n$, present in the system is strictly less than the total number of possible environmental states, $S$. Specifically, we consider $n\leq S-2$. This assumption accounts for the possibility that the environmental states may have both predictable and unpredictable variations. In the case of a predictable environmental state, the system will opt for the optimal adaptive strategy, namely, phenotype plasticity. However, when the environmental variations are unpredictable, the bet-hedging strategy will be triggered. Thus, there may not be bet-hedgers corresponding to each environmental state considered.
    
    \item \label{assump-2} The deterministic nature of the fitness dynamics is enforced by aggregating the stochastic nature of environmental states via mean-variance fitness (Eq.~\eqref{eq:geometric-mean-fitness-a}),  an approximation of geometric mean fitness,  instead of environmental state dependent fitness (Eq.~\eqref{eq:random-fitness-a}).
    
    \item \label{assump-3} Type $\mathsf{A}$ sub-population is assumed to exhibit distinct phenotypic subtypes in response to certain environmental cues (see Sec.~\ref{sec:hetrogeneity-in-model}). The fitness of bet-hedgers, however, does not depend on such plasticity. Further, the correlations between bet-hedgers and the type $\mathsf{A}$ sub-population, quantified by the matrix $V_n$ (Eq.~\eqref{eq:correlation-matrix}), are assumed to be independent of the phenotypic plasticity.
\end{enumerate}
\end{tcolorbox}

Let $z^{(0)}(t)$ be the fraction of sub-population that follows the fitness type $\mathsf{A}$ at time $t$ and $z^{(i)}(t)$ be the fraction of $i^{\text{th}}$ subtype of $\mathsf{B}$, the bet-hedging sub-population. Since we are interested in studying the fitness dynamics of sub-populations that are affected by environmental variations, the total fitness does not consider the fitness of type $\mathsf{C}$. Assuming the environment is in state $s$, at time $(t+1)$, the total fitness is given by
\begin{align}
\label{eq:fit-minus-c}
    G_s(t+1)&= \left(-R\fitBold{}{t} z^{(0)}(t)+ \fitB{s}{t+1} z^{(0)}(t)\right)\nonumber\\
    &~~~+\sum_{i=1}^n\left(-R\fitCold{i}{t} +\fitCs{i}{t+1}\right)z^{(i)}(t).
\end{align}
Note that, to find the total fitness for a particular fitness type, we simply multiplied the fraction of the sub-population with their individual fitness. More sophisticated demographic models, such as matrix models integrating an individual's average life history attributes could also be used to arrive at the cumulative fitness of the population~\cite{Leslie1945, Caswell2001}. 

Let us now define
$\bm{z}(t):=\left(z^{(1)}(t),\cdots, z^{(n)}(t)\right)$ and $\vec{Z}(t)=\left(z^{(0)}(t), \bm{z}(t)\right)^T$, where the superscript $T$ denotes the transpose; and 
\begin{align}\nonumber
    \bm{B}(t):=\begin{pmatrix}
        \fitCold{1}{t}\\
        \vdots\\
        \fitCold{n}{t}
    \end{pmatrix}
~~\text{and}~~
    \bm{\mathcal{B}}_s(t+1):=\begin{pmatrix}
        \fitCs{1}{t+1}\\
        \vdots\\
        \fitCs{n}{t+1}
    \end{pmatrix}.
\end{align}
Then Eq.~\eqref{eq:fit-minus-c} in terms of $\vec{Z}(t)$, $\bm{B}(t)$ and $\bm{\mathcal{B}}_s(t+1)$ becomes
\begin{align}
\label{eq:random-fitness-a}
    G_s(t+1)=\innerp*{\begin{pmatrix}
        -R\fitBold{}{t} + \fitB{s}{t+1}\\
        -R \bm{B}(t) + \bm{\mathcal{B}}_s(t+1)
    \end{pmatrix},~\vec{Z}(t)
    },
\end{align}
where $\innerp*{\vec{x},\vec{y}}=\vec{x}^T\vec{y}$ for $\vec{x}, \vec{y}\in \mathbb{R}^{n+1}$. 

The total fitness, denoted by $G(t+1)$, is therefore a random variable that takes values $\{G_s(t+1)\}_{s=1}^{S}$ with probability $\{q_s\}_{s=1}^{S}$. It is widely believed that the natural selection favors a strategy that maximizes the geometric mean of the fitness of a population~\cite{Dempster1955, Haldane1963, Gillespie1973, Orr2007}. Note further that the geometric mean $g$ of random variable $X$ can be approximated as $g\approx \mu-\text{var}/(2\mu)$, where $\mu$ is the arithmetic mean and $\text{var}$ is the variance of $X$ (see Appendix~\ref{append:geo-ari}). Keeping this relation in mind, in the following, we consider the conservative mean-variance fitness $\mathsf{H}(t+1)$ of the population as our quantity of interest which is defined as follows (following the Assumption~\ref{assump-2}):
 \begin{align}
 \label{eq:geometric-mean-fitness-a}
     \mathsf{H}(t+1) = \mathbb{E}\left[G(t+1)\right] - \frac{1}{2}\mathbb{E}\left[\left(G(t+1)-\mathbb{E}\left[G(t+1)\right]\right)^2\right],
 \end{align}
 where $\mathbb{E}\left[G(t+1)\right]=\sum_{s=1}^S q_s G_s(t+1)$. The above quantity has to be maximized as a consequence of natural selection and in quantitative finance, relates to the risk averse agents ~\cite{Orr2007, Brock1998}. Computing the expectation values above (see Appendix~\ref{append:var-matrix}), we get
 \begin{align}
 \label{eq:fitness-beg}
     \mathsf{H}(t+1) = \innerp*{\begin{pmatrix}
        -R\fitBold{}{t} + \mathbb{E}\left[\fitB{}{t+1}\right]\\
        -R \bm{B}(t) + \bm{q}
    \end{pmatrix},~\vec{Z}(t)
    }-\frac{1}{2}\innerp*{\vec{Z}(t), V_{n}\vec{Z}(t)},
 \end{align}
where $\bm{q}=\left(q_1,\cdots,q_n\right)^T$ and $V_n$ is the covariance matrix for the individual fitnesses for the subpopulations. The $(n+1)\times (n+1)$ real symmetric matrix $V_{n}$ is given by
 \begin{align}
 \label{eq:correlation-matrix}
     V_{n}=\begin{pmatrix}
       v & \bm{u}^T\\
        \bm{u} & \Gamma
    \end{pmatrix}.
 \end{align}
 The scalar $v$ is the variance of the individual fitness for type $\mathsf{A}$ subpopulation. The $n\times 1$ vector $\bm{u}$ consists of the covariances between individual fitness of type $\mathsf{A}$ subpopulation and individual fitnesses of $n$ bet-hedgers. $\Gamma$, an $n\times n$ matrix, consists of covariances of individual fitnesses of $n$ bet-hedgers. Let $\alpha$ be a random variable that takes the values $\{\alpha_s\}$ with probability $\{q_s\}$ and let $\mathbb{E}[\alpha]:=\overbar{\alpha}$, then  $v=\sum_{s=1}^S q_s \left(\alpha_s-\overbar{\alpha}\right)^2$, the $i^{\text{th}}$ component of $\bm{u}$ is given by $u_i =q_i\left(\alpha_i-\overbar{\alpha}\right)$, and  $(i,j)^{\text{th}}$ entry of $\Gamma$ is given by $\Gamma_{ij}=q_i \left(\delta_{ij}- q_j\right)$.  See Appendix~\ref{append:var-matrix} for more details. Now let us define the population vector $\vec{Z}(t)$ that maximizes the mean-variance fitness $\mathsf{H}(t+1)$ (Eq.~\eqref{eq:fitness-beg}) as $\vec{Z}_{\text{opt}}$. After little algebra
we obtain
\begin{align}
\label{eq:optimal-pop}
    \vec{Z}_{\text{opt}}(t)=\left(V_{n}\right)^{-1} \begin{pmatrix}
       -R\fitBold{}{t} +\fitBold{}{t+1} + \overbar{\alpha}\\
        -R \bm{B}(t) + \bm{q}
    \end{pmatrix} := \left(V_{n}\right)^{-1}\vec{Q}(t),
 \end{align}
 where $\vec{Q}(t)$ is given by
 \begin{align}
     \vec{Q}(t):= \begin{pmatrix}
       -R\fitBold{}{t} +\fitBold{}{t+1} + \overbar{\alpha}\\
        -R \bm{B}(t) + \bm{q}
    \end{pmatrix}.
 \end{align}
To simplify our exposition of the mean-variance fitness, we introduce and define the notion of the {\it intrinsic fitness} of the population in the following.


\subsection{Fitness in terms of intrinsic fitness}
\label{sec:fitness-fundamental}
The optimal population vector $\vec{Z}_{\text{opt}}$, given by Eq.~\eqref{eq:optimal-pop}, comprises the fractions of population following the fitness type $\mathsf{A}$, and type $\mathsf{B}$ that yields the maximum mean-variance fitness $\mathsf{H}(t+1)$. Now, we define the intrinsic fitness of the population as the fitness $\begin{pmatrix}
        A_*\\
        \bm{B}_{*}
    \end{pmatrix}$ that satisfies the following two conditions: $(1) $ $\fitBold{}{t}=A_{*}=A^{(0)}(t+1)$ and $\bm{B}(t) =\bm{B}_{*}=\bm{B}(t+1)$; $(2)$ 
    $\vec{Z}_{\text{opt}}(t) = \begin{pmatrix}
        P\\
        \bm{0}
    \end{pmatrix}$.
Here $P$ is the total population of sub-populations of type $\mathsf{A}$ and $\mathsf{B}$. The first condition imposes that the intrinsic fitness is independent of time. The second condition imposes that the optimal population vector only have type $\mathsf{A}$ sub-population and no type $\mathsf{B}$ sub-population. Fitness can be considered to have reached its intrinsic value when individuals within the population follow the same strategy as of type $\mathsf{A}'$s and at the same time maximize the mean-variance fitness. The conditions $(1)$ and $(2)$ for the intrinsic fitness impose the constraint
\begin{align}
\label{eq:fundamental-fitness}
    \begin{pmatrix}
        P\\
        \bm{0}
    \end{pmatrix}=V_{n}^{-1}\begin{pmatrix}
        (1-R)A_{*}+\overbar{\alpha}\vspace{0.2cm}\\
        -R\bm{B}_* + \bm{q}
    \end{pmatrix}\Leftrightarrow
    V_{n}\begin{pmatrix}
        P\\
        \bm{0}
    \end{pmatrix}
    = \begin{pmatrix}
        (1-R)A_{*}+\overbar{\alpha}\vspace{0.2cm}\\
        -R\bm{B}_* + \bm{q}
    \end{pmatrix}, 
 \end{align}
 which can be solved to obtain $\begin{pmatrix} A_*\\ \bm{B}_{*} \end{pmatrix}$. Let us define the deviations from intrinsic fitness of the type $\mathsf{A}$ and type $\mathsf{B}$ sub-populations at time $t$ as $X(t)=A^{(0)}(t)-A_{*}$ and $\bm{Y}(t)=\bm{B}(t)-\bm{B}_{*}$, respectively. Now using Eqs.~\eqref{eq:optimal-pop} and \eqref{eq:fundamental-fitness}, we have
 \begin{align}
 \label{eq:opt-dev-fund}
\vec{Z}_{\text{opt}}(t)= \begin{pmatrix}
        P\\
        \bm{0}
    \end{pmatrix} + V_{n}^{-1}\begin{pmatrix}
        -R X(t)+X(t+1)\vspace{0.2cm}\\
        -R\bm{Y}(t)
    \end{pmatrix},
 \end{align}
 and 
 \begin{align}
 \label{eq:q-def}
     \vec{Q}(t)= V_{n} \begin{pmatrix}
        P\\
        \bm{0}
    \end{pmatrix} + \begin{pmatrix}
        -R X(t)+X(t+1)\vspace{0.2cm}\\
        -R\bm{Y}(t)
    \end{pmatrix}=V_{n}\vec{Z}_{\text{opt}}(t).
    \end{align}   
So far, we have considered the sub-population of type $\mathsf{A}$ as having fitness that depends on the environmental states. However, to better reflect realistic scenarios, it is crucial to account for phenotypic heterogeneity within type $\mathsf{A}$, as different individuals within the same species may respond differently to the same environmental cues. In the following section, we will introduce and discuss this phenotypic heterogeneity in the type $\mathsf{A}$ sub-population.


\subsection{Phenotypic heterogeneity in the response to environmental variations}
\label{sec:hetrogeneity-in-model}

As we know the fitness of type $\mathsf{A}$ sub-population depends on environmental condition, it may respond to environmental cues by exhibiting distinct phenotypes $\mu=1,\cdots, k$. Given that the environmental cues are appropriate and consistently predicted, the fitness at a future time $(t+1)$ can be written as a sum of intrinsic fitness and a function influenced by the history of the fitnesses over the past few cycles. Thus adaptive phenotypic plastic response to environmental perturbations can be included in the model simply by replacing the deviation $X(t+1)$ from the intrinsic fitness at time $(t+1)$ with a function of the deviations of the fitness from the intrinsic fitness values at prior times, say, at times $\{(t-1),\cdots,(t-m)\}$, where $m$ is some integer. In particular, for the $\mu^{\text{th}}$ phenotype in sub-population $\mathsf{A}$, we model $X_{\mu}(t+1)$ as
\begin{align}
\label{eq:heterogeneity}
    X_{\mu}(t+1):=h_{\mu}^{(t)}\equiv h_{\mu}\left(X{(t-1)},\cdots,X{(t-m)}\right),
\end{align}
where the function $h_{\mu}$ will be specified based on the practical scenario at hand (see e.g.  Section~\ref{sec:analytic2}). We further assume that the presence of phenotypic plasticity in sub-population $\mathsf{A}$ will not affect the fitness of bet-hedgers and the matrix $V_{n}$ of covariances also remains the same (see Assumption~\ref{assump-3}). Using these considerations into Eq.~\eqref{eq:opt-dev-fund}, we get the optimal population vector $\vec{Z}_{\mu}$ for the $\mu^{\text{th}}$ subtype inside type $\mathsf{A}$ sub-population as follows.
\begin{align}
 \label{eq:opt-dev-fund-mu-type}
\vec{Z}_{\mu}(t)= \begin{pmatrix}
        P\\
        \bm{0}
    \end{pmatrix} + V_{n}^{-1}\begin{pmatrix}
        -R X(t)+h_{\mu}^{(t)}\vspace{0.2cm}\\
        -R\bm{Y}(t)
    \end{pmatrix},
 \end{align}
 Also, 
 \begin{align}
 \label{eq:q-z-type-mu}
V_{n}\vec{Z}_{\mu}(t):= V_{n}\begin{pmatrix}
        P\\
        \bm{0}
    \end{pmatrix} + \begin{pmatrix}
        -R X(t)+h_{\mu}^{(t)}\vspace{0.2cm}\\
        -R\bm{Y}(t)
    \end{pmatrix}=\vec{Q}_{\mu}(t).
 \end{align}
Further, let $n_{\mu}(t)$ be the fraction of subtype of $\mathsf{A}$ sub-population with phenotype $\mu$ such that $\sum_{\mu=1}^k n_{\mu}(t)=1$ and
\begin{align}
\label{eq:rel-nmu-p}
\sum_{\mu=1}^{k}n_{\mu}(t)\vec{Z}_{\mu}(t)= \begin{pmatrix}
        \sum_{\mu=1}^{k}n_{\mu}(t)z^{(0)}_{\mu}(t)\vspace{0.2cm}\\
        \bm{z}(t)
    \end{pmatrix} =  \begin{pmatrix}
        P \vspace{0.2cm}\\
        \bm{0}
    \end{pmatrix}.
 \end{align}
Using Eq.~\eqref{eq:rel-nmu-p} in Eq.~\eqref{eq:opt-dev-fund-mu-type}, we get
\begin{align}
\begin{pmatrix}
        P\\
        \bm{0}
    \end{pmatrix}= \begin{pmatrix}
        P\\
        \bm{0}
    \end{pmatrix} + V_{n}^{-1}\begin{pmatrix}
        -R X(t)+\sum_{\mu=1}^{k}n_{\mu}(t)h_{\mu}^{(t)}\vspace{0.2cm}\\
        -R\bm{Y}(t)
    \end{pmatrix}.
 \end{align}
Noting that $V_{n}$ is invertible and $R\neq 0$, we conclude that
 \begin{align}
 \label{eq:fundamental-condition-main}
        R X(t)=\sum_{\mu=1}^{k}n_{\mu}(t)h_{\mu}^{(t)},~~\text{and}~~\bm{Y}(t)=0.
 \end{align}
Eq.~\eqref{eq:fundamental-condition-main} represents the steady-state fitness dynamics, and is our main dynamical equation. To analyze the steady-state fitness dynamics, we model $n_{\mu}(t)$ and $h_{\mu}^{(t)}$ with appropriate functions. In the next section, we will study a system that has two phenotypic heterogeneous subtypes, as well as bet-hedgers along with specific choice of functions for $n_{\mu}(t)$ and $h_{\mu}^{(t)}$.


\section{\label{sec:analytic} Two fold heterogeneity}
In this section, we explore analytically the case when phenotypic plasticity leads to the formation of two distinct subtypes, i.e., $k=2$ in Eq.~\eqref{eq:fundamental-condition-main}, in presence of $n$ bet-hedgers. In this case, the steady-state fitness dynamics of type $\mathsf{A}$ sub-population becomes
\begin{align}
 \label{eq:fundamental-condition}
        R X(t)=n_{1}(t)h_{1}^{(t)} + n_{2}(t)h_{2}^{(t)}.
 \end{align}
We choose the functions $n_{1}(t)$, $n_{2}(t)=1-n_{1}(t)$, $h_{1}^{(t)}$ and $h_{2}^{(t)}$ in such a way that the steady-state fitness dynamics, given by Eq.~\eqref{eq:fundamental-condition}, can accommodate a wide variety of isogenic populations.


\subsection{\label{sec:analytic1} Choosing the function \texorpdfstring{$n_{\mu}(t)$}{}}
The function $n_{\mu}(t)$, describing the fraction of $\mu^{\text{th}}$ phenotypic subtype, should be an increasing function of its fitness. This means that $n_{1}(t)\geq n_{2}(t)$ holds if the individual fitness of phenotype $1$ is larger than that of the individual fitness of phenotype $2$ at time $t-1$. But what should be the strength of this dependence of $n_{\mu}(t)$ on the fitness at time $(t-1)$? Let us denote the fitness of phenotype $\mu$ individual at time $t-1$ by $f_{\mu}(t-1)$ and the strength of dependence of $n_{\mu}(t)$ on fitness $f_{\mu}(t-1)$ by $\beta$. One then expects that if $\beta$ is very large, the individuals from larger fitness phenotype will dominate the total population. In other words $\lim_{\beta\rightarrow\infty} n_{2}(t)\rightarrow 1$ if $f_2(t-1)> f_1(t-1)$. Further, if $\beta$ is very small, then irrespective of the fitness value one expects that $n_1(t)\approx n_2(t)$. The actual dependence of $n_{\mu}(t)$ on $f_{\mu}(t-1)$ may be complex and can be obtained from the available data. Here, we use an exponential dependence of $n_{\mu}(t)$ as a generic model that satisfies the above expectations,
\begin{align}
\label{eq:n-mu-f-beta}
        n_{\mu}(t):=\exp\left[\beta f_{\mu}(t-1)\right]/\left(\sum_{\mu=1}^k\exp\left[\beta f_{\mu}(t-1)\right]\right).
 \end{align}
 With the above relation, it is straightforward to note that $\lim_{\beta\rightarrow\infty} n_{2}(t)= 1$ when $f_2(t-1)> f_1(t-1)$ and $\lim_{\beta\rightarrow 0} n_{1}(t)= 1/2=n_{2}(t)$ irrespective of their corresponding fitness values. The above functional form appears in the multinomial logit model~\cite{Anderson1988, Brock1998} and can be argued to be an outcome of Jayne's principle of maximum entropy~\cite{Jaynes1957a, Jaynes1957b} applied in an appropriate way ~\cite{Golan1996, Shipley2006, Harte2011}. The function $\beta$ is closely related also to the inverse of fitness variance of the sub-population~\cite{Brock1998, Brock2009}, and can be used as a parameter for controlling the steady-state fitness dynamics. By taking $f_{\mu}(t-1)$ as the optimal fitness given by Eq.~\eqref{eq:fitness-beg} where $\vec{Z}(t-2)$ is replaced by optimal $\vec{Z}_{\mu}(t-2)=V_{n}^{-1}\vec{Q}_{\mu}(t-2)$, we have
\begin{align}
f_{\mu}(t-1)&:=\innerp*{\vec{Q}(t-2),~\vec{Z}_{\mu}(t-2)
    }-\frac{1}{2}\innerp*{\vec{Z}_{\mu}(t-2), V_{n}\vec{Z}_{\mu}(t-2)}\nonumber\\
    &=\frac{1}{2}\innerp*{\vec{Q}(t-2), V_{n}^{-1}\vec{Q}(t-2)}-\frac{1}{2}\left(V_{n}^{-1}\right)_{00}\left(X(t-1)-h_{\mu}^{(t-2)}\right)^2,
 \end{align}
where $\left(V_{n}^{-1}\right)_{00}$ is the first element of the inverse covariance matrix $V_{n}^{-1}$. Using above expression in Eq.~\eqref{eq:n-mu-f-beta}, we get
\begin{align}
\label{eq:n-mu-g-beta}
        n_{\mu}(t)=\exp\left[\beta g_{\mu}(t-1)\right]/\left(\sum_{\mu=1}^ke^{\beta g_{\mu}(t-1)}\right);~g_{\mu}(t-1)=-\frac{\left(V_{n}^{-1}\right)_{00}}{2}\left(X(t-1)-h_{\mu}^{(t-2)}\right)^2.
 \end{align}


\subsection{\label{sec:analytic2} Choosing the function \texorpdfstring{$h_{\mu}^{(t)}$}{}}
Note that the functions $h_{\mu}^{(t)}$ $(\mu=1,~2)$ determine both the heterogeneity in fitness and the fraction of subtypes following that fitness type, see Eqs.~\eqref{eq:heterogeneity} and \eqref{eq:n-mu-g-beta}. The choice of function $h_{\mu}^{(t)}$ is sensitive to the particular system at hand and there may not be a simple single expression representing $h_{\mu}^{(t)}$.  So to elaborate the predictive power of the fitness dynamics presented here, we model the functions $h_{\mu}^{(t)}$ for $\mu=1$ and $\mu=2$, which are inspired from models of reinforcement learning. Let us consider phenotype $1$ ($\mu=1$) such that its fitness at time $t+1$, given by Eq.~\eqref{eq:heterogeneity}, doesn't depend on its fitness at earlier times. Such a phenotype may arise when a population behaves conservatively and does not take clue from past \footnote{Not taking a clue is indeed also a response to the environmental changes.}. For such phenotypes we have $h_{1}^{(t)}=k_1$, where $k_1$ is a real non-negative constant. Let us consider the other phenotype ($\mu=2$) to be corresponding to the population that does take clues from past. In particular, let us take 
\begin{align}
\label{eq:h-mu-t}
        &h_{2}^{(t)}=k_2\tanh X(t-1) + k_3 X(t-1),
 \end{align}
where $k_2$ and $k_3$ are real non-negative constants. $h_{2}^{(t)}$ contains a linear term depending on just the previous fitness and another transcendental term that is well behaved and bounded between $0$ and $1$ for $X(t-1)\geq 0$. The boundedness of $\tanh X(t-1)$ term shows a saturation in the amount of the learning from the past. $\tanh$ function is used heavily as an activation function in recurrent neural networks (RNN) and deep learning models (see e.g. Refs.~\cite{Yu2019, Alzubaidi2021}). 


\subsection{\label{sec:analytic3} Fitness dynamics}
\label{subsec:fix-points}
In the following, we take $n_{\mu}(t)$ as given by Eq.~\eqref{eq:n-mu-g-beta} and $h_{\mu}(t)$ as given by Eq.~\eqref{eq:h-mu-t}, and define the following functions of three variables to facilitate the analysis of fixed points of the steady-state fitness dynamics given by Eq.~\eqref{eq:fundamental-condition}:
\begin{align}
    &n\left(u_1, u_2, u_3\right)
    :=\left(1+\exp\left[\widetilde{\beta}\left(\left(u_1-k_1\right)^2-\left(u_1-(k_2\tanh u_3 + k_3 u_3)\right)^2\right)\right]\right)^{-1},\nonumber\\
    &F(u_1, u_2, u_3)=\frac{1}{R}\left[n(u_1, u_2, u_3)k_1 + (1-n(u_1, u_2, u_3)) \left(k_2\tanh u_1+ k_3 u_1\right)\right],
\end{align}
where $\widetilde{\beta}= \frac{1}{2}\left(V_{n}^{-1}\right)_{00}\beta$. Now the steady-state fitness dynamics, Eq.~\eqref{eq:fundamental-condition}, becomes
\begin{align}
\label{eq:tw-het-eq}
    X(t) = F\left(X(t-1), X(t-2), X(t-3)\right).
\end{align}
The steady-state fitness dynamics is a third order nonlinear difference equation whose fixed points are the solution of the equation $F_{\text{fix}}(x^*):= x^* - F\left(x^*, x^*, x^*\right)=0$. Note that  $F_{\text{fix}}(0)<0$ and $F_{\text{fix}}(x_0)\geq 0$ for $$x_0:=\min\{x^*:~x^* \geq k_1/R ~\text{and} ~x^*/\tanh x^* \geq k_2/(R-k_3)\}.$$ Thus, all the fixed points of the fitness dynamics lie in the range $[0,x_0]$. The exact values of the fixed are found numerically (see Sec.~\ref{sec:result}).


\subsection{\label{sec:analytic4} Linear stability analysis of the fixed points}
The linearization of the steady-state  fitness dynamics around the fixed point $x^*$ can be obtained by putting $y_t=X(t)-x^*$ in Eq.~\eqref{eq:tw-het-eq} and ignoring the quadratic and higher order terms. This converts Eq.~\eqref{eq:tw-het-eq} into
 \begin{align}
     y_t=\sum_{i=1}^3 c_i y_{t-i},
 \end{align}
where $c_i\equiv c_i(n;\beta,x^*)$ for $i=1,2,3$ are real constants (independent of time) given by

\begin{align*}
c_1&:=\frac{\partial F(u_1, u_2, u_3)}{\partial u_1}\Bigg|_{(x^*,x^*,x^*)}\\
&=\frac{1}{R}\left[n_{1*}^22\widetilde{\beta} \left(k_1-(k_2\tanh x^*+ k_3 x^*)\right)^2 \exp\left[\widetilde{\beta}\left(\left(x^*-k_1\right)^2-\left((1-k_3)x^*-k_2\tanh x^*\right)^2\right)\right]\right.\\
&~~~~~~~~~~\left.+ (1-n_{1*})\left(k_2\sech^2x^* + k_3\right)\right]\\
c_2&:=\frac{\partial F(u_1, u_2, u_3)}{\partial u_2}\Bigg|_{(x^*,x^*,x^*)}=0\\
c_3&:=\frac{\partial F(u_1, u_2, u_3)}{\partial u_3}\Bigg|_{(x^*,x^*,x^*)}\\
&=\frac{1}{R}\left[-n_{1*}^22\widetilde{\beta} \left( (1- k_3)x^* - k_2\tanh x^* \right)\left( k_2\sech^2 x^* + k_3\right) \left(k_1-k_2\tanh x^*- k_3 x^*\right)\right.\\
&~~~~~~~~\left. \times \exp\left[\widetilde{\beta}\left(\left(x^*-k_1\right)^2-\left((1-k_3)x^*-k_2\tanh x^*\right)^2\right)\right]\right].
 \end{align*}
In the above equations, we have defined 
\begin{align}    n_{1*}:=\left[1+\exp\left[\widetilde{\beta}\left(\left(x^*-k_1\right)^2-\left((1-k_3)x^*-k_2\tanh x^* \right)^2\right)\right]\right]^{-1}.
 \end{align}
The quantity $n_{1*}$ represents the fraction of phenotype 1 ($\mu$ = 1) of $\mathsf{A}$ sub-population at the fixed point $x^*$. The linearized steady-state fitness dynamics around the fixed point $x^*$, at fixed parameter value $\beta$, is given by
 \begin{align}
 \label{eq:yt}
     y_t=c_1y_{t-1}+c_3y_{t-3}.
 \end{align}
 We convert this third order linear difference equation into three first order linear difference equations by defining the following new variables: $z_{i}(t)$ for $i=1,~2,~3$ as $z_1(t)=y_t$, $z_2(t)=y_{t-1}$ and $z_3(t)=y_{t-2}$. In terms of these new variables we can rewrite Eq.~\eqref{eq:yt} in the matrix form as
\begin{align}
    \begin{pmatrix}
        z_1(t)\\
        z_2(t)\\
        z_3(t)
    \end{pmatrix} &=\begin{pmatrix}
        c_1 & 0 & c_3\\
        1  & 0 & 0\\
        0  & 1 & 0
    \end{pmatrix} \begin{pmatrix}
        z_1(t-1)\\
        z_2(t-1)\\
        z_3(t-1)
    \end{pmatrix}.
\end{align}
Now the stability of the fixed points of the Eq.~\eqref{eq:fundamental-condition} can be determined by the eigenvalues of the matrix  $J=\begin{pmatrix}
        c_1 & 0 & c_3\\
        1  & 0 & 0\\
        0  & 1 & 0 \end{pmatrix}$ 
called Jacobin matrix as follows. The eigenvalue equation/characteristic equation for the matrix $J$ is given by
 \begin{align}
     \lambda^3-c_1 \lambda^2-c_3=0.
 \end{align}
According to the Schur-Cohn criterion~\cite{Elaydi2005}, the fixed point $x^*$ is asymptotically stable  if and only if
\begin{align}
\label{eq:conds-for-stability}
    \left| c_1+c_3 \right| <1 ~~\& ~~|c_1c_3|< 1-c_3^2.
\end{align}
Otherwise, we will call $x^*$ unstable (see Fig.~(\ref{fg.n=0})). Further, note that $\lambda_i=1$ for any $i=1,2,3$, where $\lambda_i$'s are the eigenvalues of the Jacobian matrix, is a necessary condition for the fold bifurcation. This necessary condition can be translated easily into demanding $c_1+c_3=1$. Thus, for fixed $n$, the pairs $\left(\bcrit(n), \xcrit(n)\right)$ for which $c_1( n; \bcrit(n), \xcrit(n)) + c_3(n; \bcrit(n), \xcrit(n))=1$ comprise the fold bifurcations.

\section{\label{sec:result} Results and discussion}

Throughout this work, for our numerical analysis, we fix following parameter values in Eq.~\eqref{eq:tw-het-eq}: 
\begin{align}
\label{eq:fix-params}
R=1.2,~k_1=0.5,~ k_2=1.1,~ k_3=0.01,~\text{and}~\alpha_s=s-1. 
\end{align}
In this section, we further fix the total number of environmental states as $10$, i.e.,  $S=10$. For the distribution over environmental states, we consider two choices: (1) Modified binomial distribution denoted as $\Binom$ (2) Uniform distribution denoted as $\unif$. More specifically, for the distribution $\Binom$ we have
\begin{align}
\label{eq:binom-dist}
q_s=\binom{S}{s} p^s (1-p)^{S-s} + (1-p)^S/S;~~p=1/2,
\end{align}
and for the distribution $\unif$ we have
\begin{align}
\label{eq:unif-dist}
q_s=1/S
\end{align}
for $s=1,\cdots,S$. While choosing the parameters, it is essential to note that $R$ represents the growth rate of the population in the absence of environmental variations, thus $R$ must always be greater than $1$ and $n\leq S-2$ such that there is scope of adding more bet-hedgers, where $S$ is the total number of environmental states.

\subsection{\label{result1} Generating the bifurcation diagram and finding tipping points}
Here we identify the tipping points of fitness dynamics of our system by generating the bifurcation diagrams which quantify the changes in the fixed points as a function of the bifurcation parameter $\beta$. We start with a scenario with no bet-hedgers ($n=0$), and the system exhibits two fold heterogeneity in response to environmental variations. For our numerical analysis we fix the parameters as in Eq.~\eqref{eq:fix-params}, and take $S=10$. We further take environmental states to follow $\Binom$ distribution (see Eq.~\eqref{eq:binom-dist}). Note that for these parameter values, $k_1/R\approx 0.41$ and $k_2/(R-k_3)\approx 0.92$. Following section~\ref{subsec:fix-points}, this signifies that for the above choices of the parameters, all fixed points of the fitness dynamics (Eq.~\eqref{eq:tw-het-eq}) lie in the range $[0,0.41]$.

\begin{figure}[!htb]
\centering
\includegraphics[width=0.5\textwidth]{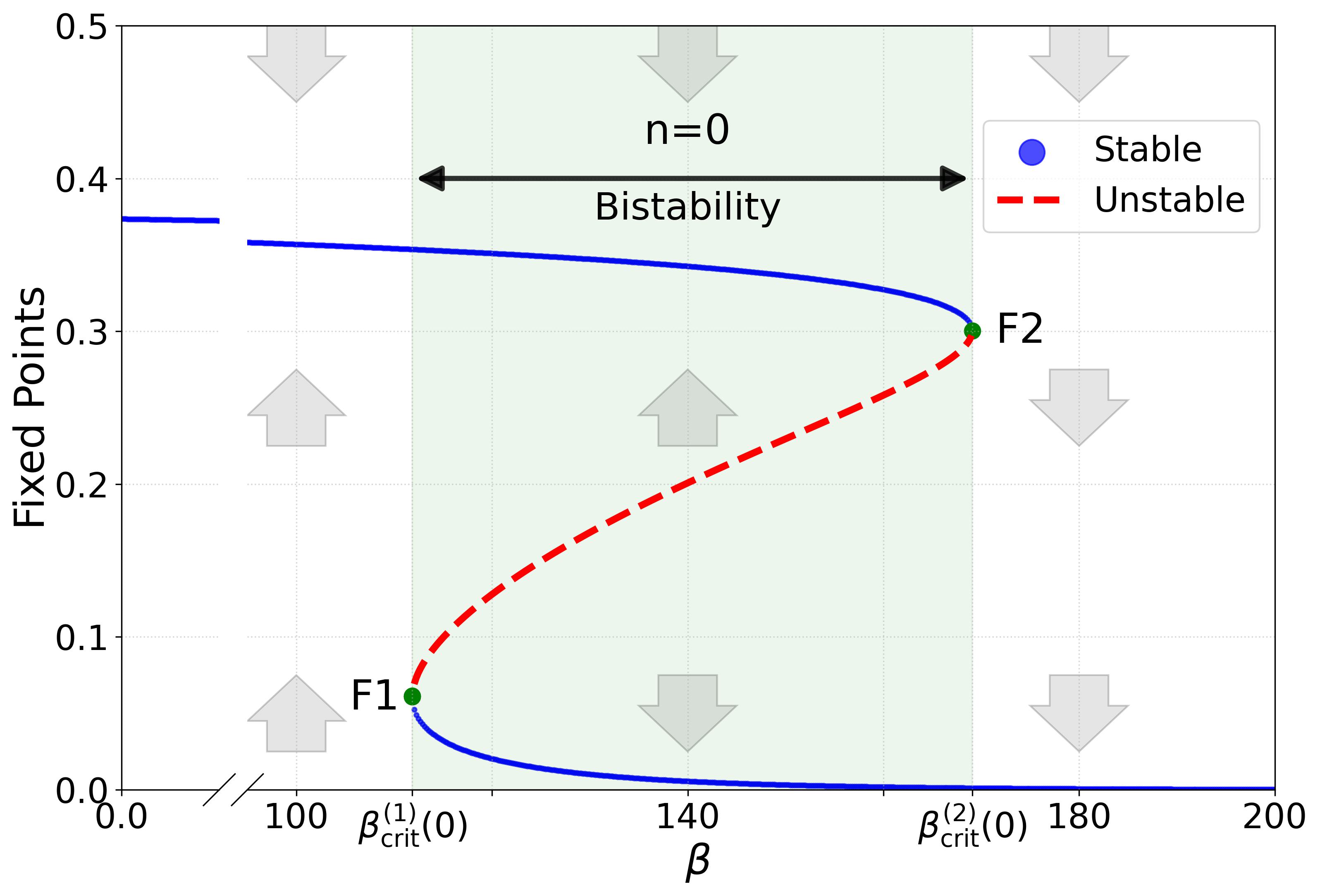}
\caption{\textbf{Fixed points in the fitness landscape as a function of the system parameter.} Stable fixed points (solid blue line) and unstable fixed points (dotted red line) of the dynamical system as a function of $\beta$ for values considered in Eq.~\eqref{eq:fix-params}, with $n=0$, $S=10$ and environmental states following $\Binom$ distribution (Eq.~\eqref{eq:binom-dist}). F1 and F2 mark the tipping points of the dynamical system. Starting from a higher fitness state (top branch), the system experiences a slow, gradual decrease in fitness as $\beta$ increases. At F2, a fold bifurcation emerges as the stable node is annihilated upon colliding with the unstable node. As $\beta$ crosses $\btwo(0) \approx 169.08$, the system rapidly transitions into the second stable state, with lower fitness value. To restore the system to the higher fitness state, $\beta$ needs to be decreased all the way to the value at F1, $\bone (0) \approx 111.85$, where another fold bifurcation occurs, creating a rapid backward shift to the higher fitness state (simply decreasing $\beta$ to the corresponding F2 value, i.e. $\btwo(0)$, is not sufficient). Within the bistable region (light green hue), bounded by $\bone(0)$ and $\btwo(0)$ values corresponding to the fold bifurcation points F1 and F2, respectively, the system can exist in either of the two stable states.}
\label{fg.n=0} 
\end{figure}

For each value of $\beta$, the fitness dynamics described by Eq.~\eqref{eq:tw-het-eq} gives rise to various fixed points and Fig.~(\ref{fg.n=0}) describes all these fixed points as a function of $\beta$. The stability analysis of these fixed points (see Sec.~\ref{sec:analytic4}) gives us the bifurcation scenario. From Fig.~(\ref{fg.n=0}) we note that for $\beta<\bone(0)$ and $\beta>\btwo(0)$ all the fixed points of the fitness dynamics are stable and the values of the fixed points for $\beta<\bone(0)$ are higher compared to the fixed points for $\beta >\btwo(0)$, which approaches zero. For each $\beta$ value in the range $\left(\bone(0),\btwo(0)\right)$ there are three fixed points, out of which one is an unstable fixed point and the other two are stable. If we start at $\beta$ close to zero then fitness of subpopulation $\mathsf{A}$ starts at a higher value (in the range $[0.35,0.4]$) and it traverses the upper blue line of stable fixed points as we increase $\beta$ value. As $\beta$ reaches $\btwo(0)$ the stable fixed point of fitness dynamics is annihilated by an unstable fixed point and gives rise to a fold bifurcation denoted as F2. A further increase in $\beta$ beyond $\btwo(0)$ results in an abrupt change in the fitness to a lower value and follows the blue line of lower stable fixed points. Now if we decrease the value of $\beta$ then the fitness traverses the lower blue line until $\beta$ reaches $\bone(0)$. At this point, the low stable fixed point merges with an unstable fixed point giving rise to another fold bifurcation denoted as F1. If $\beta$ is decreased further, then the fitness jumps to the higher fixed point and starts traversing the upper blue line of fixed points.

After exploring the stability of the fixed points of the fitness dynamics in Fig.~(\ref{fg.n=0}) as a function of $\beta$, we investigate the effect of varying $\beta$ on the fractions $n_1(t)$ and $n_2(t)$ of phenotypic $1$ and $2$ subtypes, respectively. At $\beta=0$, the fractions $n_1(t)=1/2=n_2(t)$, i.e., there is an equal number of individuals of both phenotypes. As we increase $\beta$ beyond $\btwo(0)$, most of the phenotype $1$ individuals switch into phenotype $2$ individuals, reflected by the limit $n_2(t) \xrightarrow[]{\beta>\btwo(0)}1$. This can be seen as follows: From Fig.~(\ref{fg.n=0}) when $\beta>\btwo(0)$, the fixed point of fitness dynamics approaches $0$, which implies $h_2(t)\approx 0$, then the dynamics given by Eq.~\eqref{eq:tw-het-eq} reads as $0\approx n_1(t) h_1(t)$. But $h_1(t)=0.5$ for our parameter range, which implies $n_1(t)$ must be zero and $n_2(t)\approx 1$. Now in the range $\left(0,\btwo(0)\right)$ if we start near $\beta=0$ and increase $\beta$ then $n_1(t)$ slightly decreases from $1/2$ while $n_2(t)$ slightly increases from $1/2$. At $\beta=\btwo(0)$, $n_1(t)$ drops close to zero, and each individual switches into phenotype $2$. Upon starting from a $\beta$ value higher than $\btwo(0)$ and gradually decreasing $\beta$, we observe that $n_1(t)$ starts near $0$ and remains close to zero until we reach $\bone(0)$. At $\beta=\bone(0)$ there is an abrupt change in $n_1(t)$ as it changes from near $0$ to near $1/2$ resulting in the situation where roughly half of the individuals of phenotype $2$ switch into phenotype $1$. The region $\left(\bone(0),\btwo(0)\right)$ is called the region of bistability (region of green hue in Fig.~(\ref{fg.n=0})) where either $n_1(t)\approx 0$ and $n_2(t)\approx 1$ or $n_1(t)\leq 1/2$ and $n_2(t)\geq 1/2$.

Our results show that in varying environments, phenotypic heterogeneity in population increases the fitness. However, with progressive increase of $\beta$, more individuals may switch their phenotype towards the optimal fitness strategy. When $\beta$ becomes sufficiently large, the loss of heterogeneity in the presence of adverse environmental conditions might push the system to a less desirable fitness state. The bifurcation points F1 and F2 in Fig.~(\ref{fg.n=0}) mark the catastrophic shifts between two possible fitness states and hence correspond to the \textit{tipping points} of the dynamical system. However, in this analysis, we have fixed the number of bet-hedgers to be zero. Building on this, we present the emergent fitness landscape in the following section when one or more bet-hedgers are added, specifically examining how this affects the position of the tipping points $\left( \beta_{\text{crit}}(n),x_{\text{crit}}(n)\right)$.


\subsection{\label{result2}Effect of increasing number of bet-hedgers}

\begin{figure}
\centering
\includegraphics[width=0.5\textwidth]{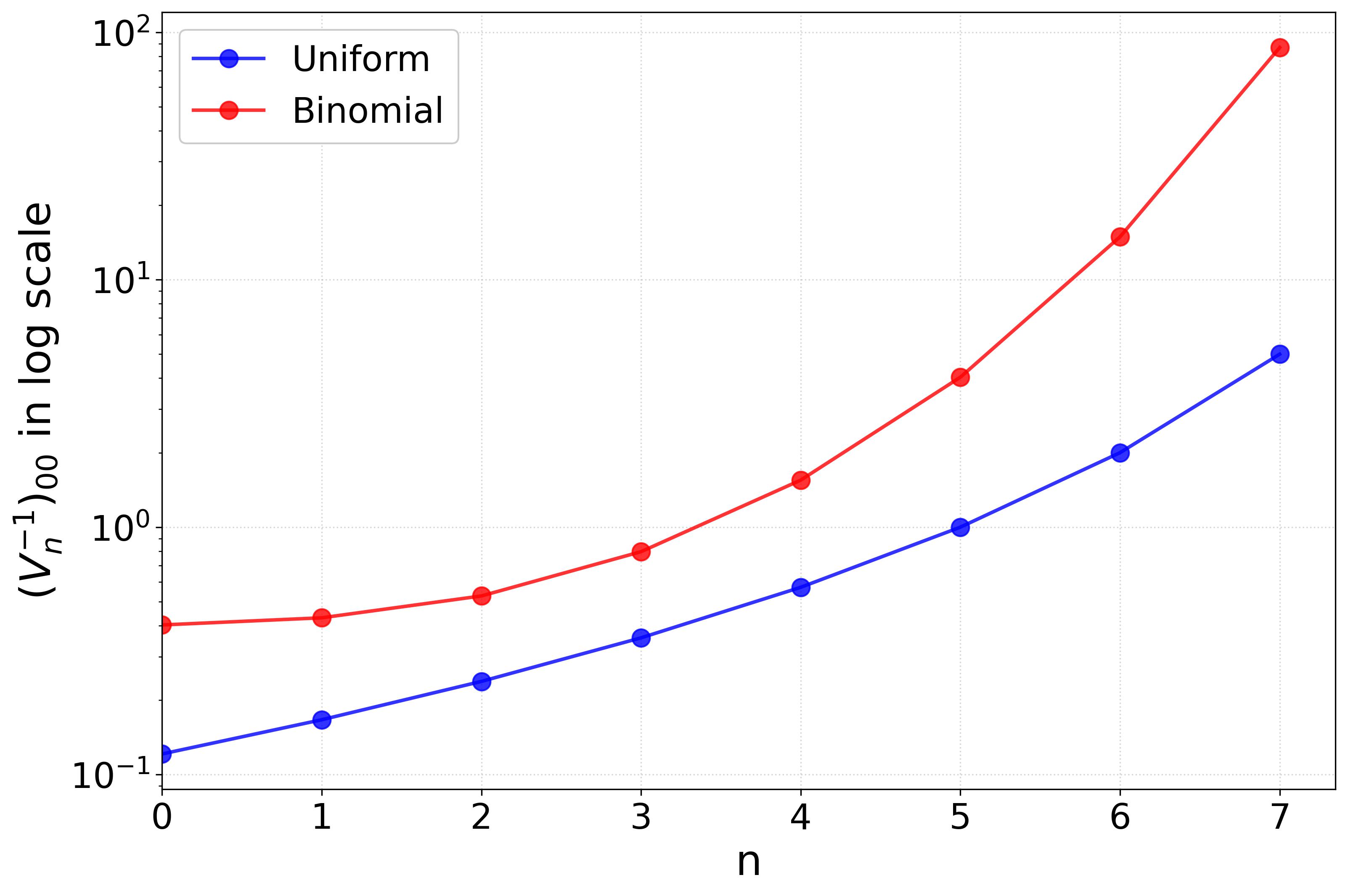}
\caption {\textbf{Inverse of fitness correlations for different distributions of environmental states.} Linear-log plot shows variation of $\left(V_{n}^{-1}\right)_{00}$ (Eq.~\eqref{eq:correlation-matrix}) as a function of $n$ for values considered in Eq.~\eqref{eq:fix-params} with $S=10$. $\left(V_{n}^{-1}\right)_{00}$ is a monotonically increasing function of $n$; shown here for $\Binom$ distribution (red curve) of the environmental states given by Eq.~\eqref{eq:binom-dist}, and for $\unif$ distribution (blue curve) according to Eq.~\eqref{eq:unif-dist}.}
\label{fg.v_inv_COMPARISON} 
\end{figure}

\begin{figure}[ht!]
\centering
\subfloat[]{
{\includegraphics[width=0.45\textwidth]{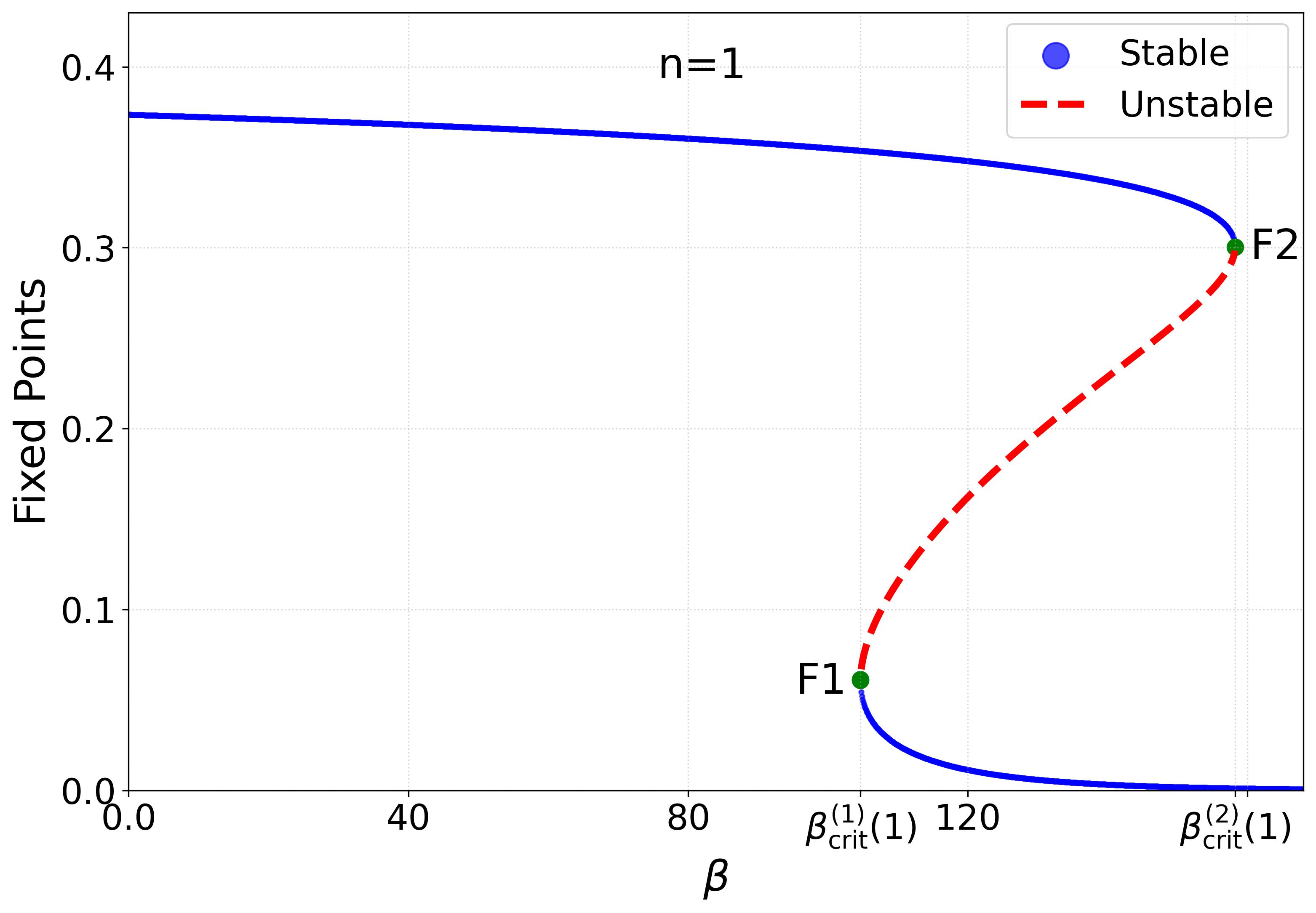}\label{fg.n=1}}
}
\subfloat[]{
{\includegraphics[width=0.45\textwidth]{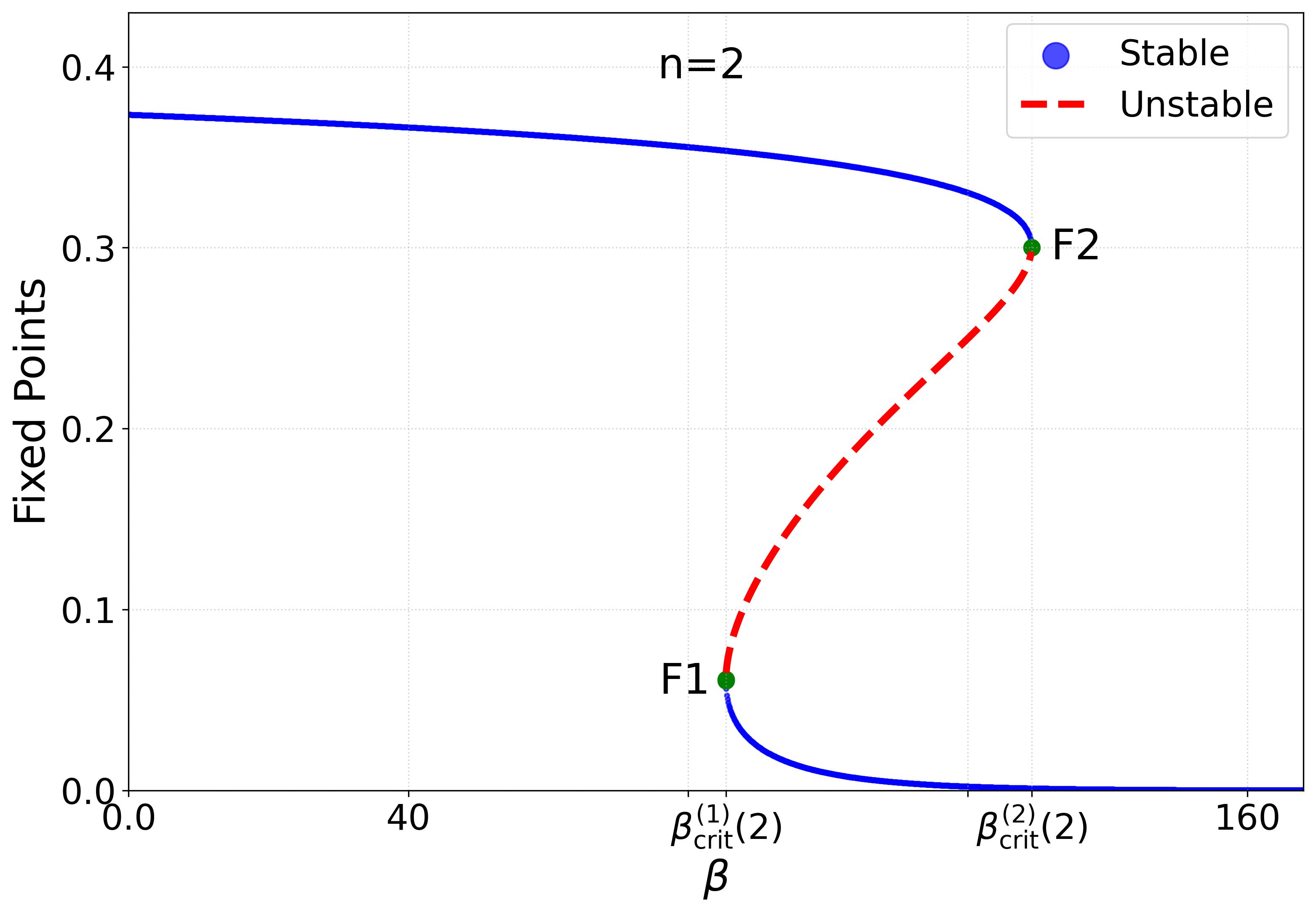}\label{fg.n=2}}
}\\
\subfloat[]{
{\includegraphics[width=0.45\textwidth]{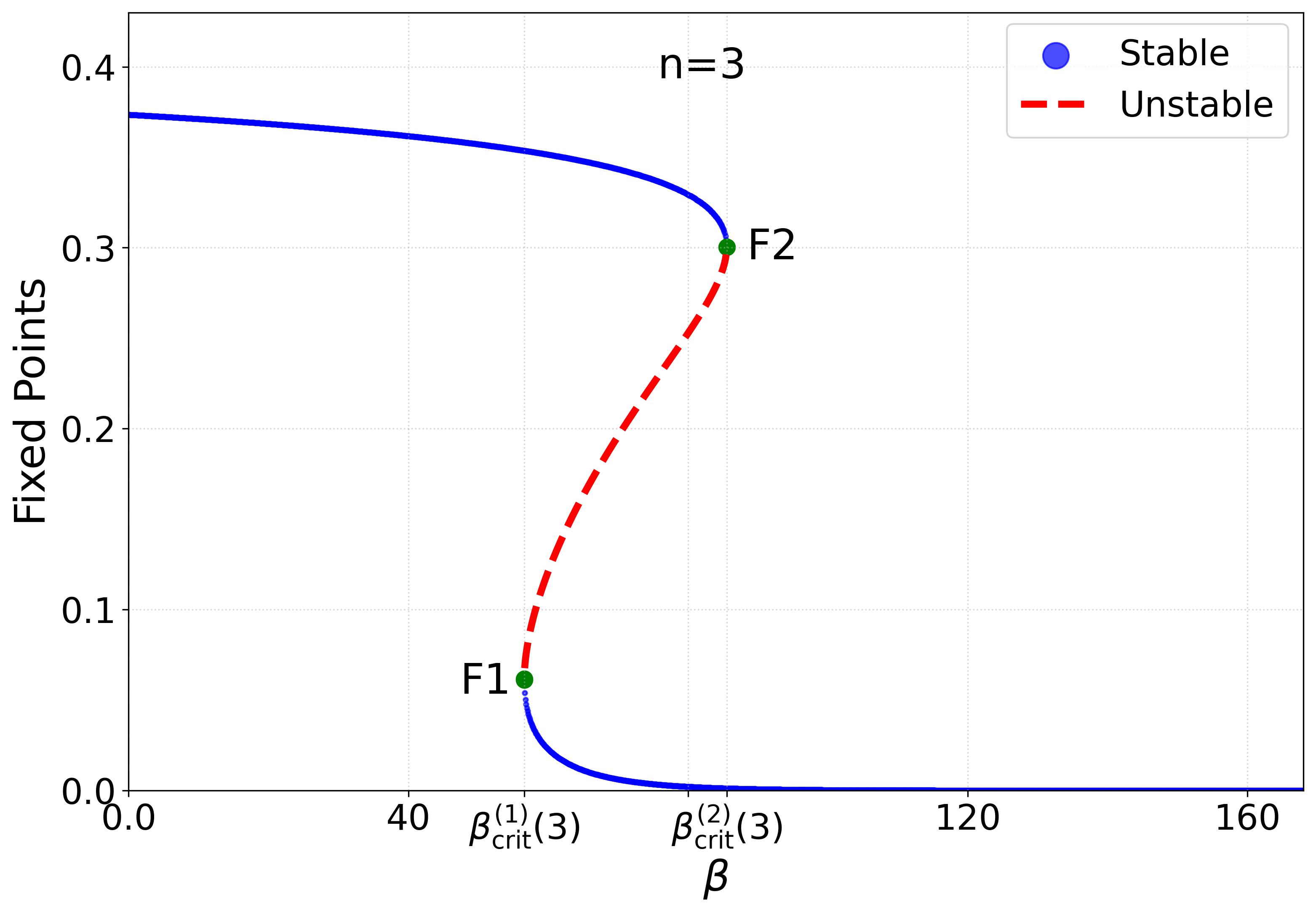}\label{fg.n=3}}
}
\subfloat[]{
{\includegraphics[width=0.45\textwidth]{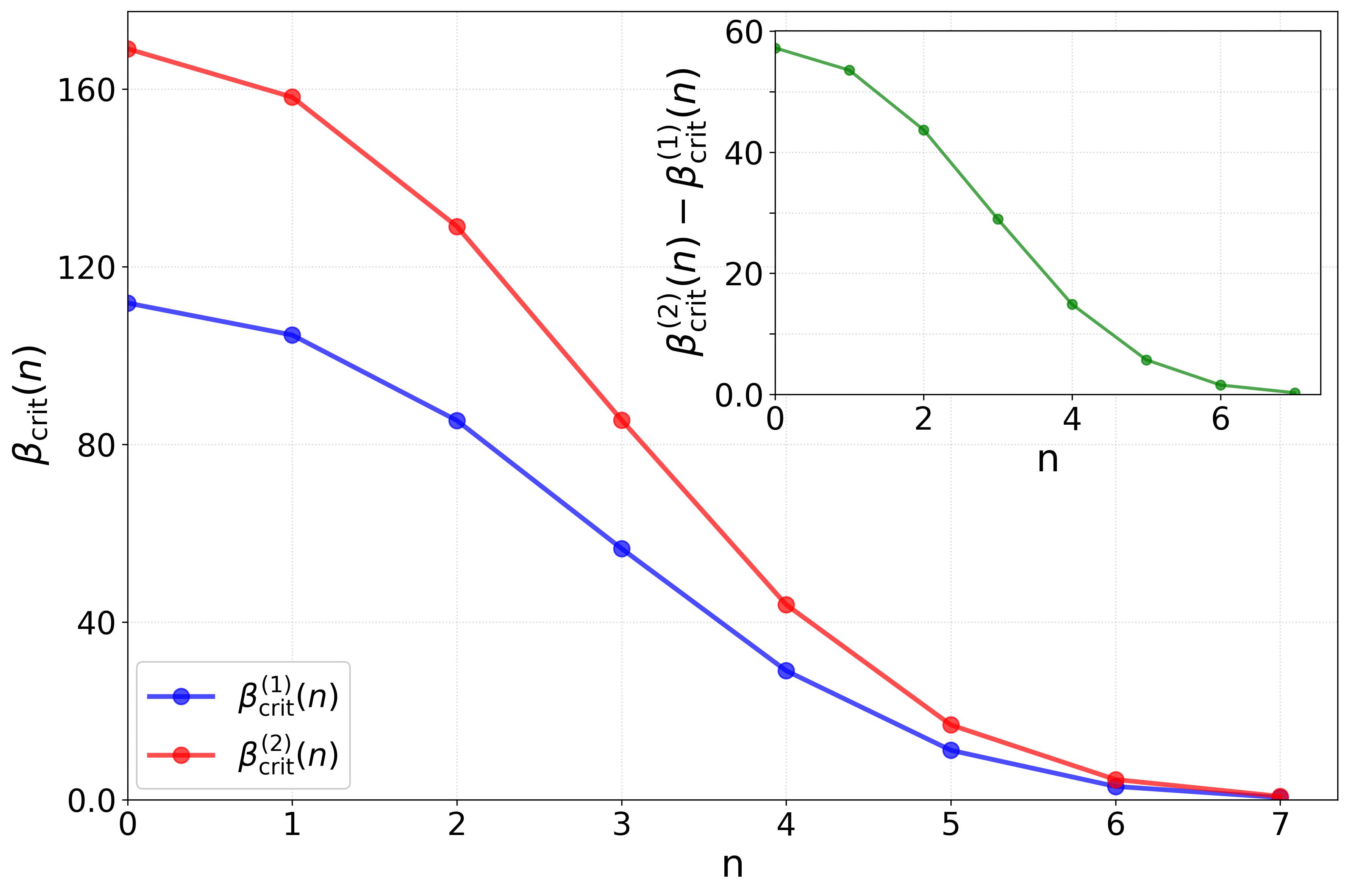}\label{fg.tippingwithn}}
}
\caption{\textbf{Bifurcation diagrams for different values of $n$.} Using parameters in Eq.~\eqref{eq:fix-params}, the covariance matrix $V_n$ is computed for the $\Binom$ distribution with $S=10$; F1 and F2 denote the tipping points. Fig.~(\ref{fg.n=1}) presents bifurcation diagram for $n=1$ where the tipping points occur at $\beta_{\text{crit}}^{(1)}{(1)}\approx 104.65$ and $\beta_{\text{crit}}^{(2)}{(2)}\approx 158.21$. Fig.~(\ref{fg.n=2}) shows the bifurcation diagram for $n=2$, with the tipping points shifted to $\beta_{\text{crit}}^{(1)}{(2)}\approx 85.40$ and $\beta_{\text{crit}}^{(2)}{(2)}\approx 129.11$; and to $\beta_{\text{crit}}^{(1)}{(3)}\approx 56.57$ and $\beta_{\text{crit}}^{(2)}{(3)}\approx 85.51$ for $n=3$ as shown in Fig.~(\ref{fg.n=3}). Overall, both tipping points shift to lower values as $n$ increases.  Fig.~(\ref{fg.tippingwithn}) presents critical values $\bone(n)$ and $\btwo(n)$ as a function of $n$ and reveals that both $\bone(n)$ and $\btwo(n)$ decrease monotonically as $n$ increases. The inset plots $\btwo(n)-\bone(n)$ as a function of $n$ which also decreases as we increase $n$.}
\label{fg.bifurcation}
\end{figure}

To analyze the dependence of tipping points on $n$, we focus on $\widetilde{\beta}(n)=\left(V_{n}^{-1}\right)_{00}\beta/2$ as a function of $n$. Figure~(\ref{fg.v_inv_COMPARISON}) plots $\left(V_{n}^{-1}\right)_{00}$ as a function of $n$ for two cases, namely, when environmental states are distributed according to distribution (1) $\unif$ (Eq.~\eqref{eq:unif-dist}) and (2) $\Binom$ (Eq.~\eqref{eq:binom-dist}). Note that the case where the environmental states are distributed according to distribution $\unif$, the uncertainty in the environmental states, as quantified by the Shannon entropy $\left(-\sum_{s=1}^Sq_s \log_2 q_s\right)$~\cite{Cover2005} of the distribution, is maximal and is equal to $\log_210\approx 3.32$ bits for $S=10$. While for the case where the environmental states are distributed according to distribution $\Binom$, the uncertainty in the environmental states is given by $\approx 2.69$ bits for $S=10$. Thus, in the case of distribution $\Binom$ the environmental states present a certain degree of predictability. In both the cases, Fig.~(\ref{fg.v_inv_COMPARISON}) shows that $\left(V_{n}^{-1}\right)_{00}$ is a monotonically increasing function of $n$, which for a fixed $\beta$ means that $\widetilde{\beta}(n)$ is also an increasing function of $n$. The following lemma proves this observation for any probability distribution.

\begin{lemma}[\cite{Brock2009}, Lemma 1]
\label{lem:increasing-bet-hedgers}
Let $V_{n}$ be the covariance matrix given by Eq.~\eqref{eq:correlation-matrix} for $n$ bet-hedgers. If one more bet-hedger is introduced in the system then the new covariance matrix $V_{n+1}$ satisfies
\begin{align}
    \left(V_{n+1}^{-1}\right)_{00} > \left(V_{n}^{-1}\right)_{00}.
\end{align}
\end{lemma}
\begin{proof}
The proof of the lemma is given in Ref.~\cite{Brock2009}. However, for the sake of completeness, we present the proof in  Appendix~\ref{append:proof of lemma 1}.
\end{proof}

Now let $\beta_{\text{crit}}{(n)}$ and $\beta_{\text{crit}}{(n+1)}$ be the values of parameter $\beta$ where the tipping points occur in the presence of $n$ and $n+1$ bet-hedgers, respectively. Then owing to the fact that the fitness dynamics depends on number $n$ of bet-hedgers only through $V_n$ and the parameter $\widetilde{\beta}(n)$ does not depend on it explicitly, we have at tipping points
\begin{align}
\label{eq:beta-tilde-equality}
    \widetilde{\beta}^*(n) = \widetilde{\beta}^*(n+1),
\end{align}
where $\widetilde{\beta}^*(n) =\frac{1}{2}\left(V_{n}^{-1}\right)_{00} \beta_{\text{crit}}{(n)} $ and $\widetilde{\beta}^*(n+1) =\frac{1}{2}\left(V_{n+1}^{-1}\right)_{00} \beta_{\text{crit}}{(n+1)}$. From Eq.~\eqref{eq:beta-tilde-equality} and Lemma~\ref{lem:increasing-bet-hedgers}, we see that
\begin{align}
\label{eq:effect-bet-hedge}
  \beta_{\text{crit}}{(n+1)} = \frac{\left(V_{n}^{-1}\right)_{00} }{\left(V_{n+1}^{-1}\right)_{00} } \beta_{\text{crit}}{(n)} < \beta_{\text{crit}}{(n)}.
\end{align}
Note that Eq.~\eqref{eq:effect-bet-hedge} implies $\bone(n+1)\leq \bone(n)$ and $\btwo(n+1)\leq \btwo(n)$, as shown in Fig.~(\ref{fg.bifurcation}). This means that as we increase the number $n$ of bet-hedgers, the tipping points occur at smaller values of $\beta$. In particular, in Figs.~(\ref{fg.n=1}),~(\ref{fg.n=2}), and~(\ref{fg.n=3}), we present the tipping points of the dynamics in the presence of $n=1$, $n=2$, and $n=3$ bet-hedgers, respectively and show that $\bone(3)\leq \bone(2)\leq \bone(1)$ and $\btwo(3)\leq \btwo(2)\leq \btwo(1)$. In Fig.~(\ref{fg.tippingwithn}), we plot $\bone(n)$ and $\btwo(n)$ as function of $n$ and show that both are monotonically decreasing functions of $n$. The inset in Fig.~(\ref{fg.tippingwithn}) shows further that $\btwo(n)-\bone(n)$ is also a decreasing function of $n$. We emphasize that the fact that tipping points appear earlier if we introduce one or more bet-hedgers into the system is independent of the total number $S$ of environment states and the number of bet-hedgers we start our analysis with. This is confirmed in cases with higher values of $S$ and $n$, as presented in the Appendix \ref{append:more-bif-fig}, where the bifurcation diagrams for $S=51$ and $n=27,~28,~29$ elicit similar qualitative features as of Figs.~(\ref{fg.n=1})-(\ref{fg.n=3}). 


\subsection{\label{result3}Critical slowing down near tipping point}
\begin{figure}
\begin{center}
\vspace*{0.5cm}
\subfloat[]{
{\includegraphics[width=0.45\textwidth]{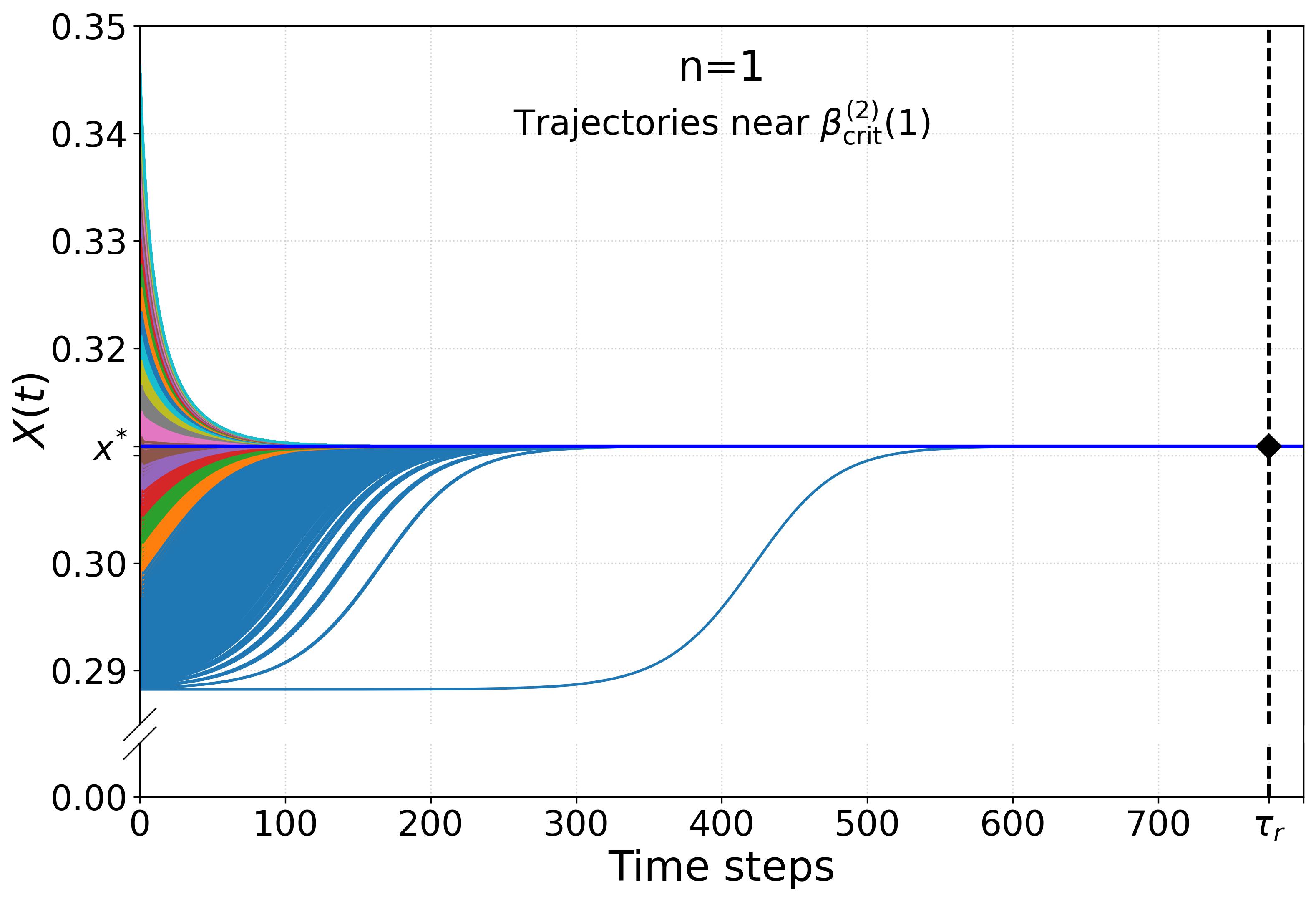}}
\label{n=1_trajectories} 
}
\subfloat[]{
{\includegraphics[width=0.45\textwidth]{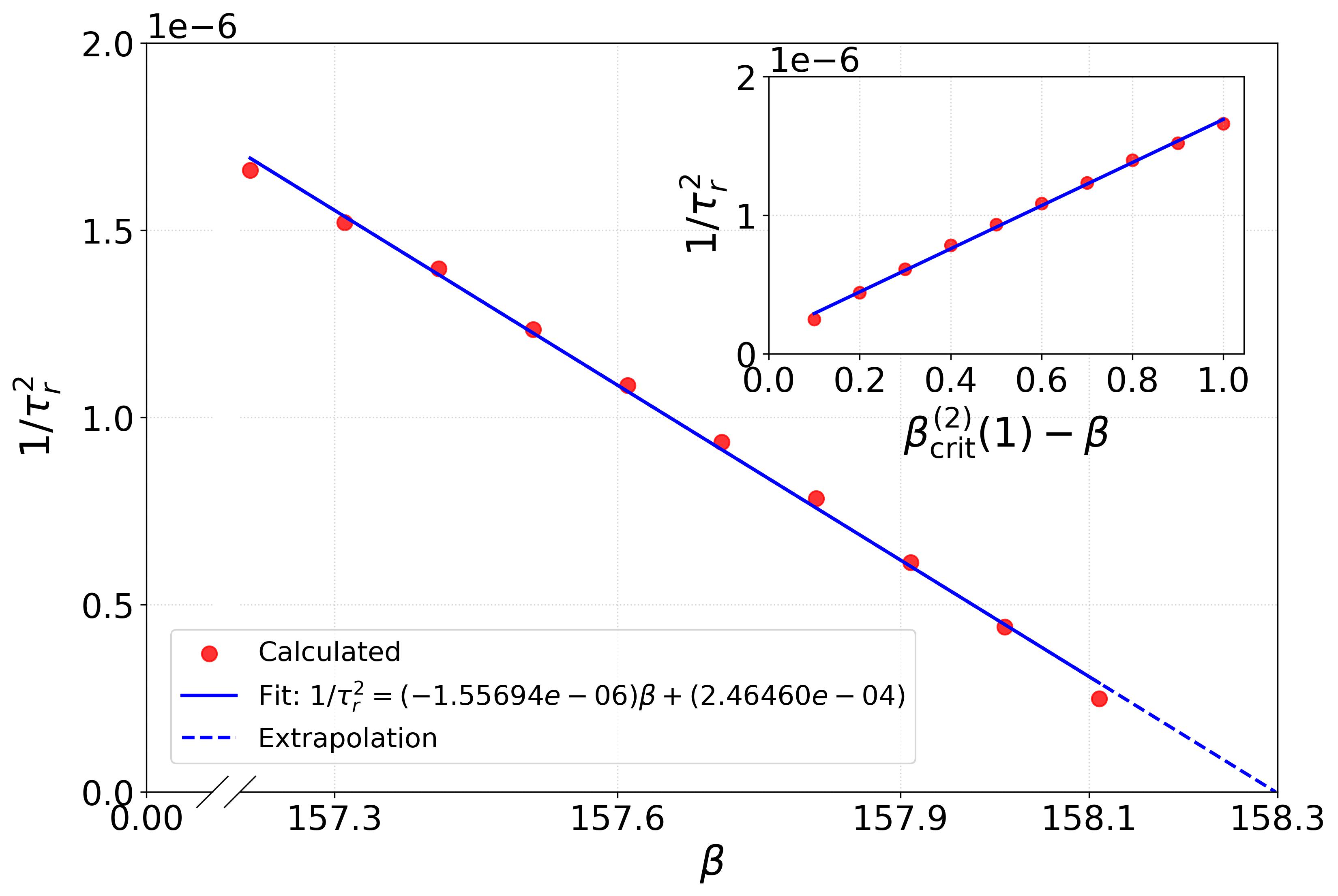}}
\label{n=1_fit} 
}
\\

\subfloat[]{
{\includegraphics[width=0.45\textwidth]{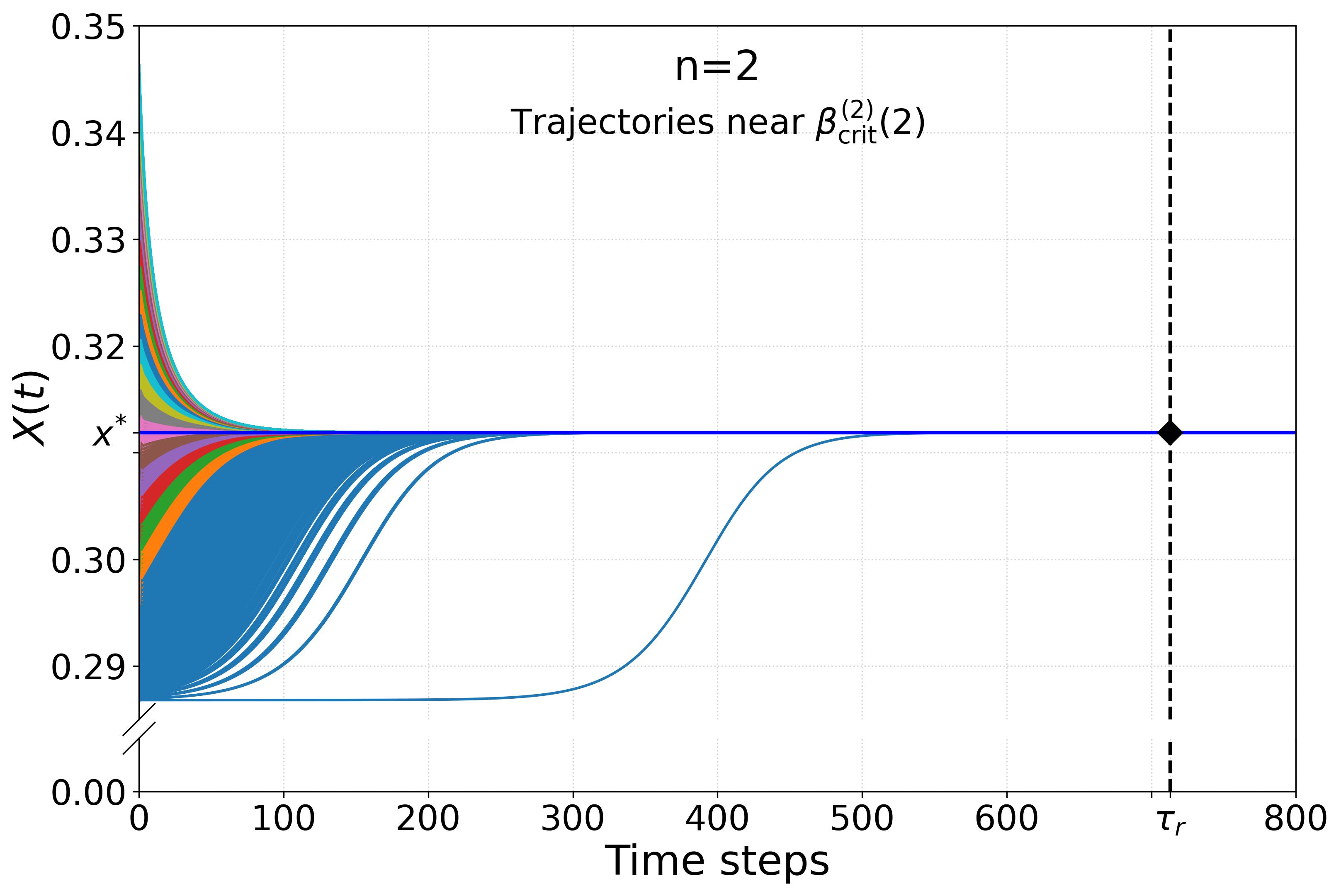}}
\label{n=2_trajectories} 
}
\subfloat[]{
{\includegraphics[width=0.45\textwidth]{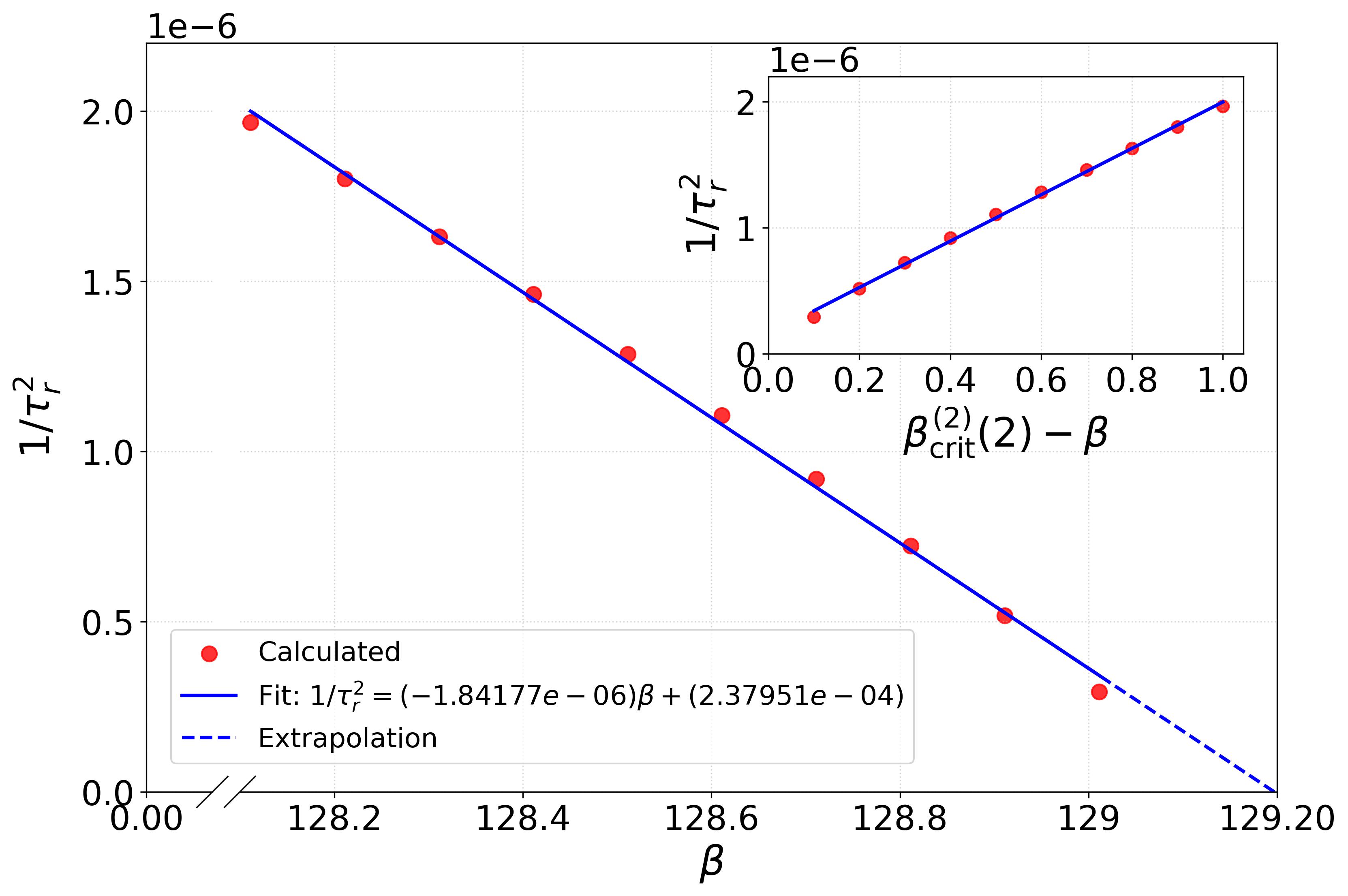}}
\label{n=2_fit}
}
\\
\subfloat[]{
{\includegraphics[width=0.45\textwidth]{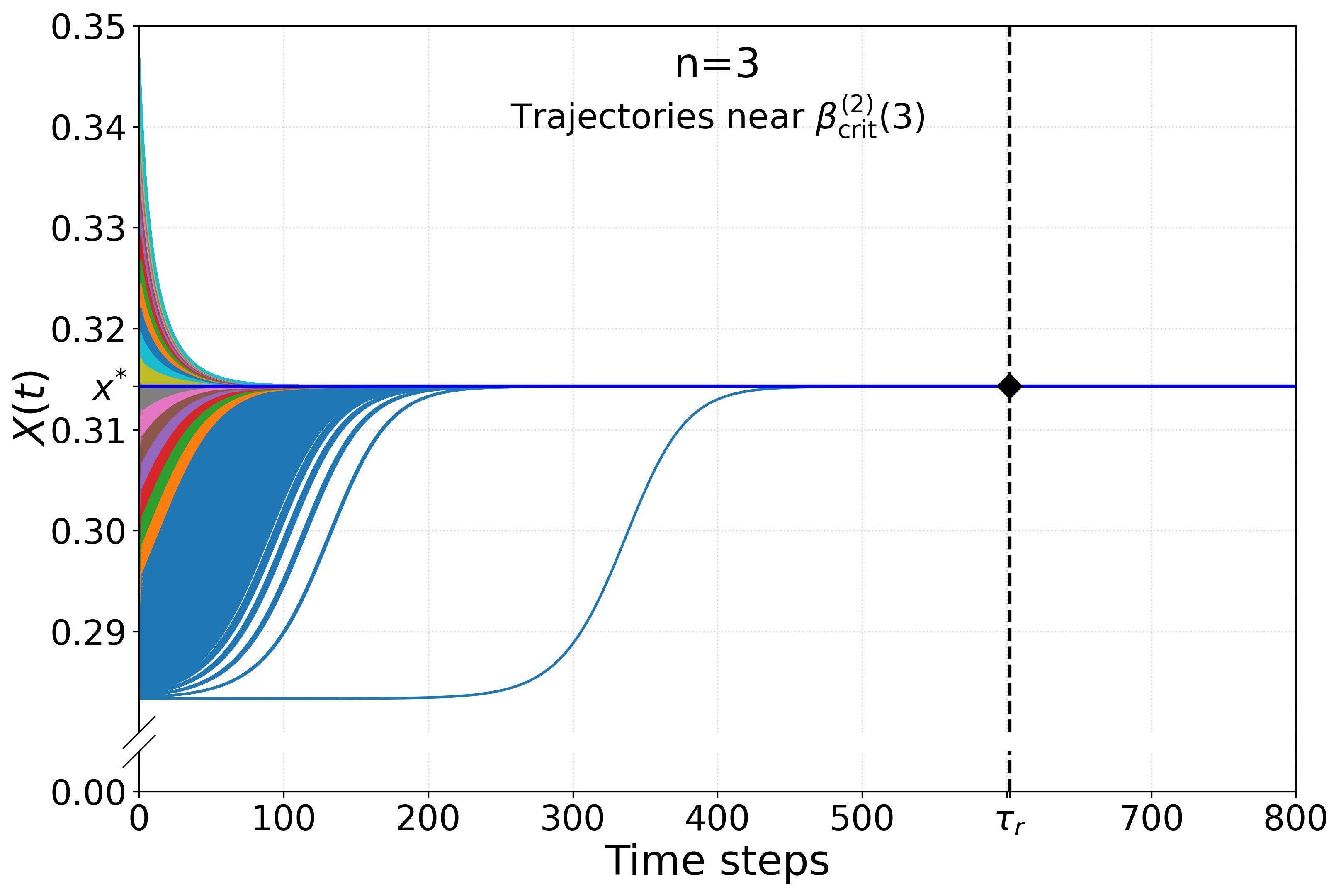}}
\label{n=3_trajectories} 
}
\subfloat[]{
{\includegraphics[width=0.45\textwidth]{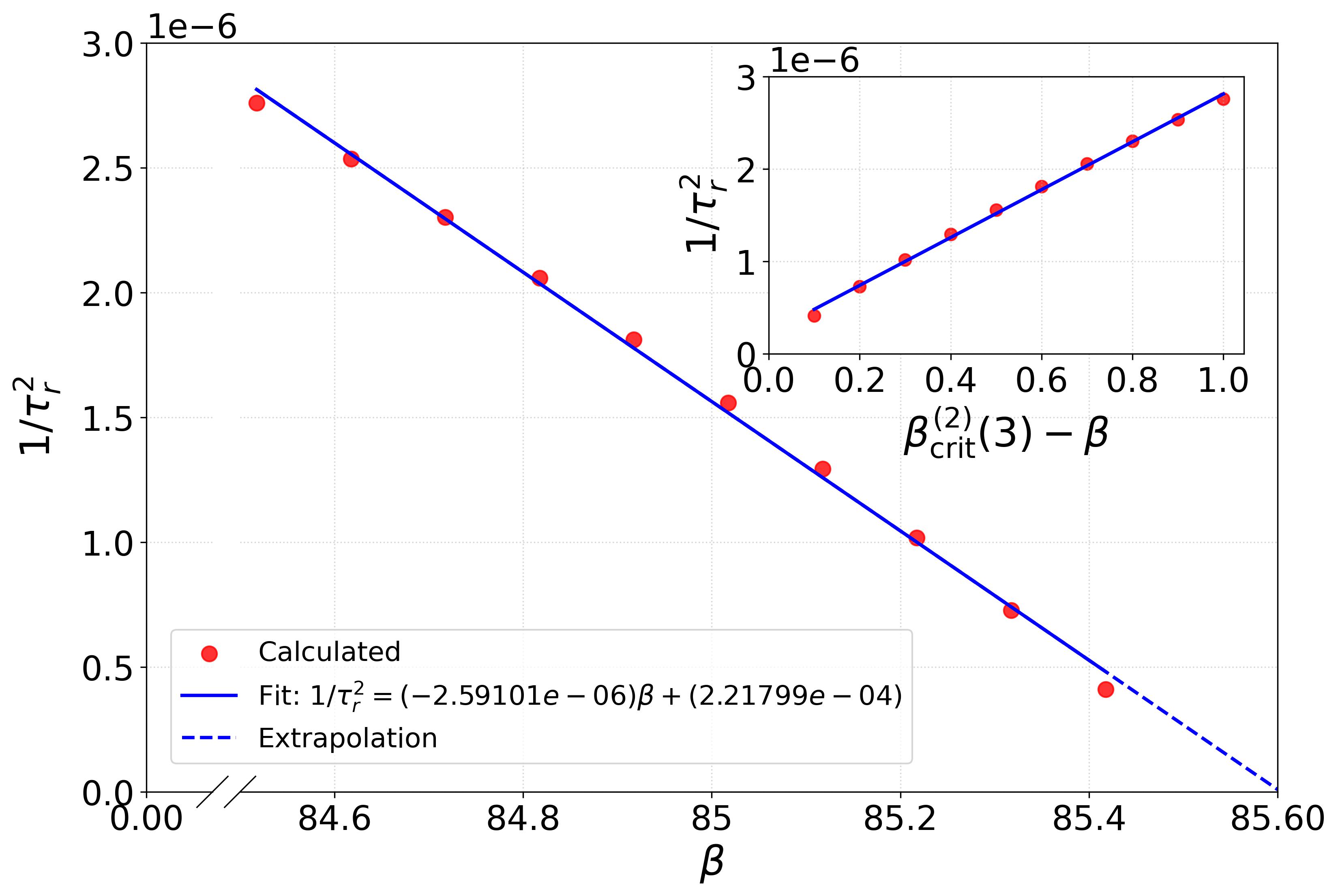}}
\label{n=3_fit}
}
\caption{\textbf{Estimation of the recovery time and critical slowing down near the tipping point.} Computations are carried out using parameters in Eq.~\eqref{eq:fix-params} and $S=10$, for $\Binom$ distribution of environmental states. Figs.~(\ref{n=1_trajectories}),~(\ref{n=2_trajectories}), and~(\ref{n=3_trajectories}) show the recovery times for various trajectories near fixed points for $n=1$, $n=2$  and $n=3$, respectively, for a fixed perturbation value $\btwo(n)-\beta = 1$ . The recovery time follows an inverse square scaling (Eq.~\eqref{eq:power-law-universal}) as a function of $\beta$ near the critical value $\btwo(n)$. Figs.~(\ref{n=1_fit}),~(\ref{n=2_fit}), and~(\ref{n=3_fit}) present the variation of $1/\tau_r^2$ as a function of $\beta$ for $n=1$, $n=2$, and $n=3$, respectively. Upon extrapolation, corresponding $\btwo(n)$ are obtained. The insets of Figs.~(\ref{n=1_fit}),~(\ref{n=2_fit}), and~(\ref{n=3_fit}) show the plots of $1/\tau_r^2$ as function of $\btwo(n)-\beta$ for $n=1$, $n=2$, and $n=3$, respectively, using which the proportionality constants for the scaling law was derived following Eq.~\eqref{eq:power-law-universal}.}
\label{fg.CSD}
\end{center}
\end{figure}

Delayed recovery from perturbations near the bifurcation point is one of the key signatures of an impending tipping point~\cite{vanNes2007, Scheffer2015rev}. Here, let $(\beta_{\text{crit}}(n), x_{\text{crit}}(n))$ denote the tipping points for the fitness dynamics, and $x_{\beta}^*$ be the stable fixed point at $\beta$. Now if one starts at initial condition $X(0)$ near the stable fixed point $x_{\beta}^*$ and a parameter value $\beta$, which is far away from the critical parameter $\beta_{\text{crit}}(n)$, then $X(t)$ converges exponentially to the fixed point $x_{\beta}^*$. However, if we start at initial condition $X(0)$ near the stable fixed point $x_{\beta}^*$ with a parameter $\beta$ close to the critical parameter $\beta_{\text{crit}}(n)$, then $X(t)$ does not converge as quickly to the fixed point $x_{\beta}^*$. In general, the recovery time, $\tau_r$, varies as inverse polynomial in the variable $|\beta-\beta_{\text{crit}}(n)|$, such that the system starting from some initial condition at the parameter value $\beta$, which is close to $\beta_{\text{crit}}(n)$, returns to the  stable fixed point $x_{\beta}^*$. The universal scaling laws for the recovery time~\cite{Wissel1984, Leonel2016} depend on the type of the tipping points~\cite{Kuehn2011, Strogatz2000, Fontich2022}, which in the case of fold bifurcation (as in our case), has an inverse square root scaling,
\begin{align}
\label{eq:power-law-universal}
    \tau_r\propto |\beta-\beta_{\text{crit}}(n)|^{-1/2}.
\end{align}
For $\beta$ very close to $\beta_{\text{crit}}(n)$ and $\beta <\beta_{\text{crit}}(n)$, the curve between $1/\tau_r^2$ and $\beta$ is a straight line and is given by $1/\tau_r^2 = u\left(-\beta + \beta_{\text{crit}}(n)\right)$, where $u$ is an unknown proportionality constant. Such a curve, if found experimentally, can provide the value of $u$ and  $\beta_{\text{crit}}(n)$ after extrapolation.

In the following, we compute numerically the scaling law for the recovery time $\tau_r$, for different values of $n$, around the tipping point $\btwo(n)$ in the fitness dynamics (Eq.~\eqref{eq:tw-het-eq}) with parameters given by Eq.~\eqref{eq:fix-params}. We further take $S=10$ and environmental states to be distributed according to distribution $\Binom$ (see Eq.~\eqref{eq:binom-dist}). Figs.~(\ref{n=1_trajectories}),~(\ref{n=2_trajectories}), and~(\ref{n=3_trajectories}) show various trajectories corresponding to various initial conditions at a fixed value of $n$ and $\beta=\btwo(n)-1$. We first obtain the recovery time for the trajectories near the tipping point with parameter value $\btwo(n)$ for the cases of $n=1$, $n=2$ and $n=3$. This is done numerically by computing the convergence to the stable fixed point corresponding to the parameter value $\btwo(n)-1\leq \beta <\btwo(n)$ for each $n=1,~2,~3$. We would like to note here that owing to the high non-linearity of the equation describing the fitness dynamics, the trajectories are quite sensitive to the initial point. We corrected our initial condition $X(0)$ for up to $8$ floating points for the computation of recovery time. We observe  numerically that the curve $1/\tau_r^2$ against the parameter $\beta$ is a straight line as seen in Figs.~(\ref{n=1_fit}),~(\ref{n=2_fit}), and~(\ref{n=3_fit}), thus confirming the scaling $\tau_r\propto \left|\beta-\btwo(n)\right|^{-1/2}$ to be true for the fitness dynamics. In fact, for $n=1$, as a best fit with root mean square error of order $\approx 10^{-8}$, we have $1/\tau_r^2 \approx \left(-1.55 \beta + 246.45 \right) \times 10^{-6}$  which yields $\btwo(1)\approx 158.3$ while the exact value is given by $\approx 158.2$. For $n=2$, as a best fit with root mean square error of order $\approx 10^{-8}$, we have $1/\tau_r^2 \approx \left(-1.84 \beta + 237.95 \right) \times 10^{-6}$ which yields $\btwo(2)\approx 129.2$ while the exact value is given by $\approx 129.1$. For $n=3$, as a best fit with root mean square error of order $\approx 10^{-8}$, we have $1/\tau_r^2 \approx \left(-2.59 \beta + 221.79 \right) \times 10^{-6}$ which yields $\btwo(3)\approx 85.6$ while the exact value is given by $\approx 85.5$. Further, the constant of proportionality in the scaling relation $ \tau_r\propto \left|\beta-\btwo(n)\right|^{-1/2}$ is obtained by plotting $1/\tau_r^2$ as a function of $\btwo(n)-\beta$ for $\beta \leq \btwo(n)$. The insets in Figs.~(\ref{n=1_fit}),~(\ref{n=2_fit}), and~(\ref{n=3_fit})
show that $1/\tau_r^2\approx 1.55\left(\btwo(1)-\beta\right)\times 10^{-6}$ for $n=1$, $1/\tau_r^2\approx 1.84\left(\btwo(2)-\beta\right)\times 10^{-6}$ for $n=2$ and $1/\tau_r^2\approx 2.59\left(\btwo(3)-\beta\right)\times 10^{-6}$. Summarizing the results, here we have shown that the power law scaling behaviour (Eq.~\eqref{eq:power-law-universal} is independent of the number of bet-hedgers but they affect the proportionality constant of the scaling behaviour.


\subsection{\label{result4} Threshold on the number of bet-hedgers}

Let us note that when the environmental states are sampled from the distribution $\Binom$ (see Eq.~\eqref{eq:binom-dist}), in the absence of bet-hedgers, $\left(V_0^{-1}\right)_{00}\approx 0.40$. Further, for a given set of parameter values stated in Eq.~\eqref{eq:fix-params} with $S=10$ and $n=0$, our system exhibits two tipping points at $\bone(0)\approx 111.85$ and $\btwo(0)\approx 169.08$ (see Sec.~\ref{result1} and Fig.~(\ref{fg.n=0})). If a $\beta$ is fixed such that $\beta<\bone(0)$ and the other parameter values are kept the same, for $n=0$, the system will not have any tipping point. So for such a fix $\beta$, can addition of bet-hedgers into system change the dynamical landscape? We analyze this in Figs.~(\ref{fg.beta=4}) and (\ref{fg.beta=8}) for $\beta=4$ and $\beta=8$, respectively, and identify one of the following three scenarios:

\begin{enumerate}[label=(\alph*)]
\item if $\bone(0) \left(V_n^{-1}\right)_{00}\leq \beta \left(V_n^{-1}\right)_{00} \leq \btwo(0) \left(V_n^{-1}\right)_{00}$, the system lies in the bistability region; 
\item if $\beta \left(V_n^{-1}\right)_{00} \leq \bone(0) \left(V_n^{-1}\right)_{00}$, the system follows a dynamics with unique high fitness stable fixed point; 
\item if $\beta \left(V_n^{-1}\right)_{00} \geq \btwo(0) \left(V_n^{-1}\right)_{00}$, the dynamics leads the system towards a unique low fitness stable fixed point. 
\end{enumerate}

For $\beta=4$ (Fig.~(\ref{fg.beta=4})), the system has a unique high fitness stable fixed state if $n\leq 5$, while for $n=6$ and $n=7$ the system respectively features bistability and a unique low fitness stable state. On the other hand, with $\beta=8$ (Fig.~(\ref{fg.beta=8})), the system has a unique high fitness stable fixed state for $n\leq 5$. If $n\geq 6$ then the system has a unique low fitness stable fixed state.

\begin{figure}[ht!]
\centering
\subfloat[]{
{\includegraphics[width=0.49\textwidth]{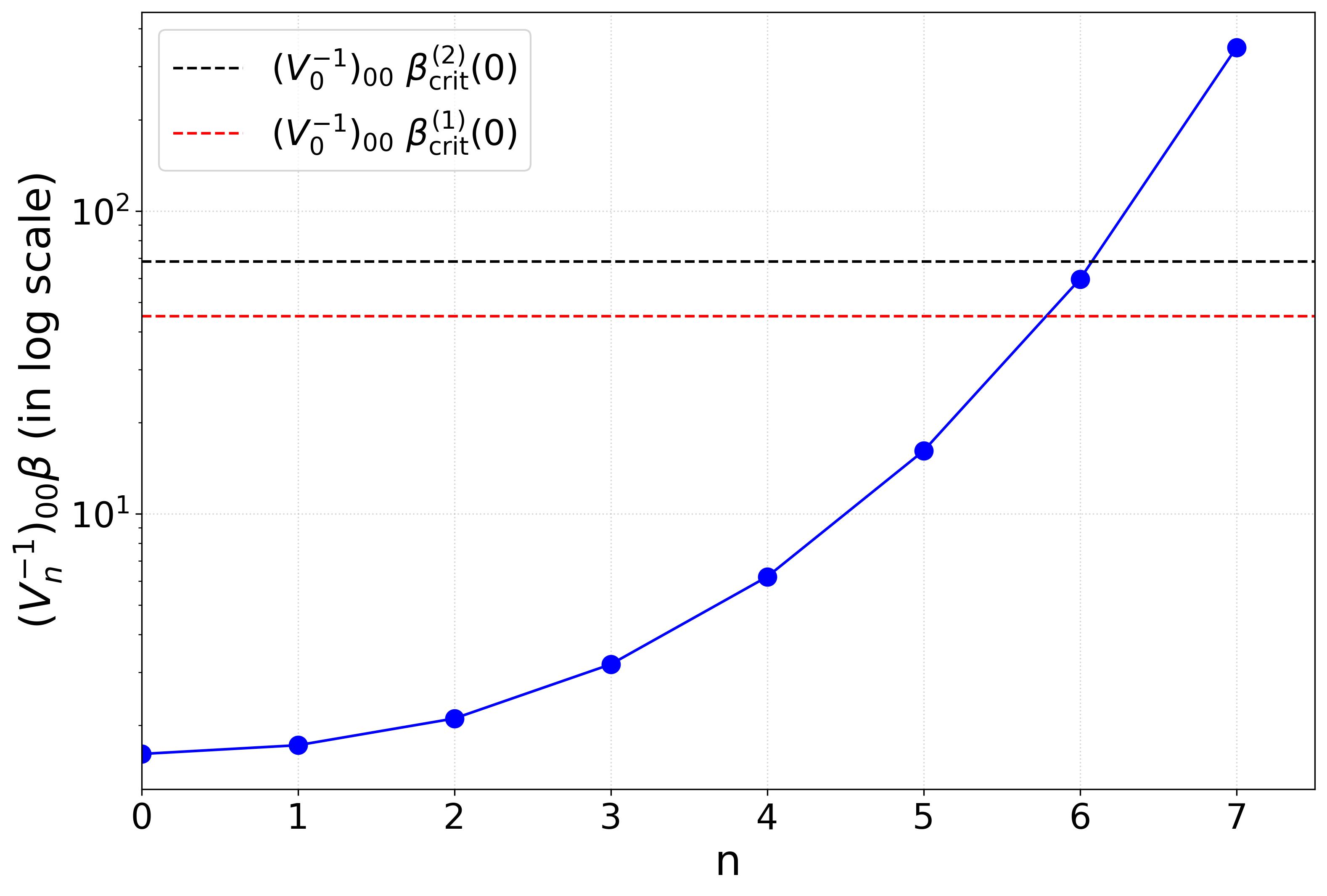}\label{fg.beta=4}}
}
\subfloat[]{
{\includegraphics[width=0.49\textwidth]{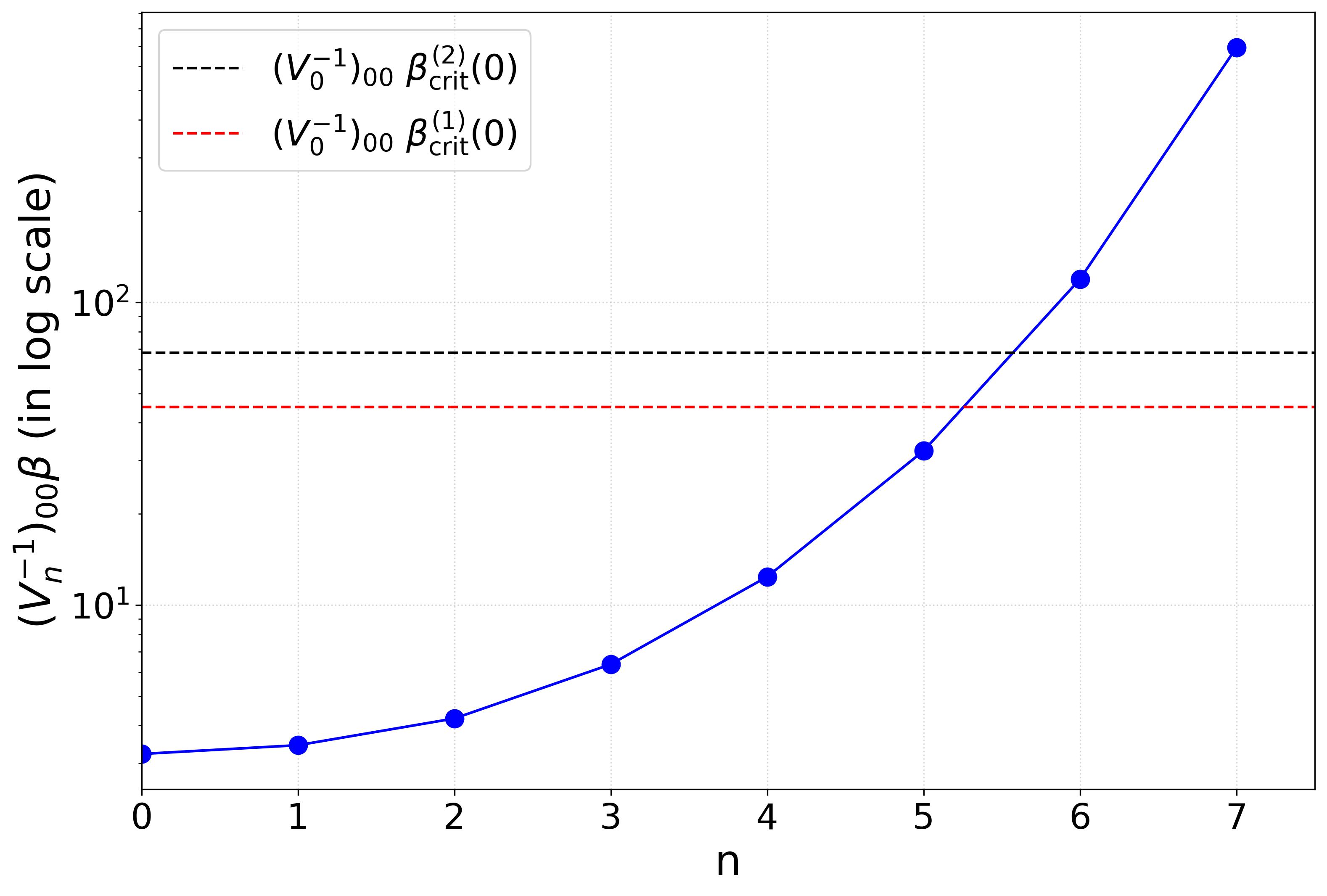}\label{fg.beta=8}}
}
\caption{\textbf{Estimating thresholds on the number of bet-hedgers.} Linear-log plot of $\beta \left(V_n^{-1}\right)_{00}$ as a function of $n$ for parameters considered in Eq.~\eqref{eq:fix-params}, with $S=10$ and environmental states following $\Binom$ distribution. Fig.~(\ref{fg.beta=4}) and Fig.~(\ref{fg.beta=8}) show the variation of $\beta \left(V_n^{-1}\right)_{00}$ (log scale) as a function of $n$, for $\beta=4$ and $\beta=8$, respectively. Fig.~(\ref{fg.beta=4}) shows that for $\beta=4$, $n\leq 5$, system has a unique high fitness stable fixed state, while for $n=6$, the system lies in the bistability region, and for $n=7$ the system has a unique low fitness stable state. For $\beta=8$ (Fig.~(\ref{fg.beta=8})), $n\leq 5$, results in a unique high fitness stable fixed state, whereas for $n\geq 6$, the system has a unique low fitness stable fixed state.}
\label{fg.no_bet-hedgers}
\end{figure}

Bet-hedgers trade-off their mean and variance in fitness in order to ‘hedge their evolutionary bets’ is well-established as a strategy to mitigate evolutionary risks. Our study uncovers the existence of a lower bound on fitness variance, based on the presence of bet-hedgers, that must be maintained to prevent the system from approaching a tipping point. However, if the number of bet-hedgers increases, the likelihood of tipping points, potentially triggering catastrophic shifts, in the system increases. This sets an upper bound on the number of bet-hedgers for a given set of parameters. In the absence of bet-hedgers $1/\left(V_0^{-1}\right)_{00}$ is exactly equal to $v$, that is the fitness variance for type $\mathsf{A}$ subpopulation (see Eq.~\eqref{eq:correlation-matrix}). Now in the presence of $n$ bet-hedgers the quantity $1/\left(V_n^{-1}\right)_{00}$ plays the role of effective fitness variance, denoted as $\efvar(n)$, for type $\mathsf{A}$ subpopulation. In the absence of any bet-hedger, let $\bcrit(0)$ denote the tipping point which can either be $\bone(0)$ or $\btwo(0)$. Then, noting that the parameter $\beta$ is related to the inverse of variance, by increasing $\beta$ or decreasing variance the system reaches to a tipping point. Now for a fixed value of $\beta$ in the presence of bet-hedgers, the effective variance decreases by increasing the number of bet-hedgers and the system may reach to a tipping point after certain number of bet-hedgers are introduced in the system. If there are $n$ bet-hedgers present in the system then the quantity $\beta /\efvar(n)$ can describe the tipping points via the relation $\frac{\beta}{\efvar(n)}=\frac{\bcrit(0)}{v}$. In particular, the quantity $\efvar^{\min}(\beta)$ defined as
\begin{align}
   \efvar^{\min}(\beta) :=\min_{n}\left\{\efvar(n): \frac{\beta}{\efvar(n)}<\frac{\bone(0)}{v}\right\}
\end{align}
gives a cut-off on the variance such that if the effective fitness variance is greater than $\efvar^{\min}(\beta)$ for a fixed $\beta$ then the dynamics of the system will not reach critical transition.

\section{\label{sec:discussion} Summary and conclusion}

Understanding and predicting a population's ability to respond and adapt to environmental fluctuations is a long-standing challenge that has eluded both biologists and physicists. Analysis of the fitness dynamics provide insights into population's stability over time, which could be leveraged for advancing predictive models. In this work we developed a complex systems framework based on the \textit{asset pricing} model to delineate the fitness landscape of an isogenic population exhibiting both discrete phenotypic plasticity as well as diversified bet-hedging strategy in response to environmental fluctuations. Analysis of the fitness dynamics of a phenotypically plastic populations revealed that, in absence of bet-hedgers, phenotypic switching between two discrete states may render the population vulnerable by pushing the system to an undesirable fitness state, particularly when heterogeneity is reduced. The loss of heterogeneity leads to decreased fitness variance, which increases the system's vulnerability against adverse environmental states. This outcome is to be expected, given that adaptive responses to environmental variations promote an increase in fitness. 
Population dynamics, on the other hand, do not directly integrate environmental variation, as the average over the environmental states is considered. Instead, the effects of variation manifest through adaptive responses in phenotypic traits, which in turn cause heterogeneity in fitness owing to a causal connection between phenotypic traits and fitness. Interestingly, as demonstrated by our results, addition of one or more bet-hedgers could impact the onset of tipping points, and thus affect the population-scale vulnerability. Although bet-hedging confers stability to the fitness landscape of a system exposed to environmental fluctuations, their numbers beyond a certain maximum may intensify the impact of phenotypic switching towards a particular strategy, potentially leading to more pronounced destabilizing consequences. While it is well established that evolutionary bet-hedging entails a trade-off between the mean and variance of fitness, wherein phenotypes with lower temporal fitness variance may be at a selective advantage under certain environmental conditions, our results uncover crucial factors--the number of bet-hedgers and its impact--typically neglected so far in fitness analyses. Our findings demonstrate that the reduction in fitness variance is advantageous only up to a critical threshold. Beyond this point, the system is driven towards a catastrophic shift. Furthermore, our dynamical model captures the phenomena of critical slowing down close to the tipping points, which has garnered considerable attention in theoretical and empirical research owing to its potential as an early warning indicator for catastrophic shifts. 

Beyond isogenic populations, our analyses open up the possibilities to study the fitness dynamics of natural and synthetic consortia comprising multiple species with underlying interaction mechanisms \cite{haas2022}, to ultimately track the emergence of system-scale tipping points. Predicting tipping points accurately can be challenging in such generic settings, due, in parts to the underlying non-linearities. Rigorous experimental data, together with a mechanistic understanding of the interactions and feedback between different components of the system may allow to advance this dynamical model with reliable predictive ability. This generic model can be more precisely tuned through validation and iterations against bespoke experimental datasets. 
Although we focused on a deterministic model to simplify the current study--akin to the mean-field theory of the condensed matter physics whereby phase transitions in higher dimensional stochastic models could be presented~\cite{Weiss1907, Kadanoff2009, Nishimori2010}--follow up studies could choose to work with the full statistics of fitness as a random variable. Further afar, the model can be adapted to accommodate the stochasticity in the dynamical system by including a time-dependent noise term in the fitness dynamics, offering valuable insights into ecologically-relevant fitness landscapes~\cite{Boettiger2020}. Such a stochastic model may elicit exotic outcomes when the system approaches the tipping point, for instance, a large perturbation in external conditions may cause the system to shift to an alternate equilibrium state even before the actual bifurcation event occurs. While we worked with a selection of parameters with specific values ($R, k_1, k_2$, and $k_3$), expanding the parameter range will offer key insights into the fitness dynamics with potentially richer dynamical properties of fitness.

\newpage


\section*{Acknowledgments}
 
 This work was funded by the University of Luxembourg and the Luxembourg National Research Fund's ATTRACT Investigator Grant (Grant no. A17/MS/ 11572821/MBRACE to AS) as well as the CORE Grant (Grant no. C19/MS/13719464/ TOPOFLUME/Sengupta to AS). SB thanks the International Institute of Information Technology Hyderabad for its kind hospitality during her visit. US acknowledges support from the Ministry of Electronics and Information Technology (MeitY), Government of India, under Grant No. 4(3)/2024-ITEA and thanks the International Institute of Information Technology Hyderabad for the Faculty Seed Grant. AS gratefully acknowledges the AUDACITY Grant (AUDACITY Grant no.: IAS-20/CAMEOS) from the Institute for Advanced Studies, University of Luxembourg for supporting this work.

\section*{Conflict of Interest}
The authors declare no conflict of interest.

\section*{Author contributions}
Conceptualization, planning, administration, and supervision: A.S. Methodology: A.S. and S.B. Investigation, data and statistical analysis, computations and modeling: S.B. and U.S. with inputs from A.S. Writing: S.B., U.S. and A.S.

\section*{Data and Code Availability Statement}
Relevant codes supporting this study are available from the corresponding
author upon reasonable request.


\newpage

\bibliography{ecology_V2}

\begin{thebibliography}{80}%
\makeatletter
\providecommand \@ifxundefined [1]{%
 \@ifx{#1\undefined}
}%
\providecommand \@ifnum [1]{%
 \ifnum #1\expandafter \@firstoftwo
 \else \expandafter \@secondoftwo
 \fi
}%
\providecommand \@ifx [1]{%
 \ifx #1\expandafter \@firstoftwo
 \else \expandafter \@secondoftwo
 \fi
}%
\providecommand \natexlab [1]{#1}%
\providecommand \enquote  [1]{``#1''}%
\providecommand \bibnamefont  [1]{#1}%
\providecommand \bibfnamefont [1]{#1}%
\providecommand \citenamefont [1]{#1}%
\providecommand \href@noop [0]{\@secondoftwo}%
\providecommand \href [0]{\begingroup \@sanitize@url \@href}%
\providecommand \@href[1]{\@@startlink{#1}\@@href}%
\providecommand \@@href[1]{\endgroup#1\@@endlink}%
\providecommand \@sanitize@url [0]{\catcode `\\12\catcode `\$12\catcode `\&12\catcode `\#12\catcode `\^12\catcode `\_12\catcode `\%12\relax}%
\providecommand \@@startlink[1]{}%
\providecommand \@@endlink[0]{}%
\providecommand \url  [0]{\begingroup\@sanitize@url \@url }%
\providecommand \@url [1]{\endgroup\@href {#1}{\urlprefix }}%
\providecommand \urlprefix  [0]{URL }%
\providecommand \Eprint [0]{\href }%
\providecommand \doibase [0]{http://dx.doi.org/}%
\providecommand \selectlanguage [0]{\@gobble}%
\providecommand \bibinfo  [0]{\@secondoftwo}%
\providecommand \bibfield  [0]{\@secondoftwo}%
\providecommand \translation [1]{[#1]}%
\providecommand \BibitemOpen [0]{}%
\providecommand \bibitemStop [0]{}%
\providecommand \bibitemNoStop [0]{.\EOS\space}%
\providecommand \EOS [0]{\spacefactor3000\relax}%
\providecommand \BibitemShut  [1]{\csname bibitem#1\endcsname}%
\let\auto@bib@innerbib\@empty
\bibitem [{\citenamefont {Scheffer}(2009)}]{SCHEFFERbook}%
  \BibitemOpen
  \bibfield  {author} {\bibinfo {author} {\bibfnamefont {M.}~\bibnamefont {Scheffer}},\ }\href {https://press.princeton.edu/books/paperback/9780691122045/critical-transitions-in-nature-and-society} {\emph {\bibinfo {title} {Critical Transitions in Nature and Society}}}\ (\bibinfo  {publisher} {Princeton University Press},\ \bibinfo {year} {2009})\BibitemShut {NoStop}%
\bibitem [{\citenamefont {Botero}\ \emph {et~al.}(2015)\citenamefont {Botero}, \citenamefont {Weissing}, \citenamefont {Wright},\ and\ \citenamefont {Rubenstein}}]{botero_evolutionary_2015}%
  \BibitemOpen
  \bibfield  {author} {\bibinfo {author} {\bibfnamefont {C.~A.}\ \bibnamefont {Botero}}, \bibinfo {author} {\bibfnamefont {F.~J.}\ \bibnamefont {Weissing}}, \bibinfo {author} {\bibfnamefont {J.}~\bibnamefont {Wright}}, \ and\ \bibinfo {author} {\bibfnamefont {D.~R.}\ \bibnamefont {Rubenstein}},\ }\bibfield  {title} {\enquote {\bibinfo {title} {Evolutionary tipping points in the capacity to adapt to environmental change},}\ }\href {\doibase 10.1073/pnas.1408589111} {\bibfield  {journal} {\bibinfo  {journal} {Proceedings of the National Academy of Sciences}\ }\textbf {\bibinfo {volume} {112}},\ \bibinfo {pages} {184--189} (\bibinfo {year} {2015})}\BibitemShut {NoStop}%
\bibitem [{\citenamefont {Dakos}\ \emph {et~al.}(2019)\citenamefont {Dakos}, \citenamefont {Matthews}, \citenamefont {Hendry}, \citenamefont {Levine}, \citenamefont {Loeuille}, \citenamefont {Norberg}, \citenamefont {Nosil}, \citenamefont {Scheffer},\ and\ \citenamefont {De~Meester}}]{dakos_ecosystem_2019}%
  \BibitemOpen
  \bibfield  {author} {\bibinfo {author} {\bibfnamefont {V.}~\bibnamefont {Dakos}}, \bibinfo {author} {\bibfnamefont {B.}~\bibnamefont {Matthews}}, \bibinfo {author} {\bibfnamefont {A.~P.}\ \bibnamefont {Hendry}}, \bibinfo {author} {\bibfnamefont {J.}~\bibnamefont {Levine}}, \bibinfo {author} {\bibfnamefont {N.}~\bibnamefont {Loeuille}}, \bibinfo {author} {\bibfnamefont {J.}~\bibnamefont {Norberg}}, \bibinfo {author} {\bibfnamefont {P.}~\bibnamefont {Nosil}}, \bibinfo {author} {\bibfnamefont {M.}~\bibnamefont {Scheffer}}, \ and\ \bibinfo {author} {\bibfnamefont {L.}~\bibnamefont {De~Meester}},\ }\bibfield  {title} {\enquote {\bibinfo {title} {Ecosystem tipping points in an evolving world},}\ }\href {\doibase 10.1038/s41559-019-0797-2} {\bibfield  {journal} {\bibinfo  {journal} {Nature Ecology \& Evolution}\ }\textbf {\bibinfo {volume} {3}},\ \bibinfo {pages} {355--362} (\bibinfo {year} {2019})}\BibitemShut {NoStop}%
\bibitem [{\citenamefont {Kussell}\ and\ \citenamefont {Leibler}(2005)}]{Kussell2005b}%
  \BibitemOpen
  \bibfield  {author} {\bibinfo {author} {\bibfnamefont {E.}~\bibnamefont {Kussell}}\ and\ \bibinfo {author} {\bibfnamefont {S.}~\bibnamefont {Leibler}},\ }\bibfield  {title} {\enquote {\bibinfo {title} {Phenotypic diversity, population growth, and information in fluctuating environments},}\ }\href {\doibase 10.1126/science.1114383} {\bibfield  {journal} {\bibinfo  {journal} {Science}\ }\textbf {\bibinfo {volume} {309}},\ \bibinfo {pages} {2075--2078} (\bibinfo {year} {2005})}\BibitemShut {NoStop}%
\bibitem [{\citenamefont {Ackermann}(2015)}]{ackermann2015}%
  \BibitemOpen
  \bibfield  {author} {\bibinfo {author} {\bibfnamefont {M.}~\bibnamefont {Ackermann}},\ }\bibfield  {title} {\enquote {\bibinfo {title} {A functional perspective on phenotypic heterogeneity in microorganisms},}\ }\href {\doibase 10.1038/nrmicro3491} {\bibfield  {journal} {\bibinfo  {journal} {Nature Reviews Microbiology}\ }\textbf {\bibinfo {volume} {13}},\ \bibinfo {pages} {497--508} (\bibinfo {year} {2015})}\BibitemShut {NoStop}%
\bibitem [{\citenamefont {Sengupta}\ \emph {et~al.}(2017)\citenamefont {Sengupta}, \citenamefont {Carrara},\ and\ \citenamefont {Stocker}}]{sengupta2017}%
  \BibitemOpen
  \bibfield  {author} {\bibinfo {author} {\bibfnamefont {A.}~\bibnamefont {Sengupta}}, \bibinfo {author} {\bibfnamefont {F.}~\bibnamefont {Carrara}}, \ and\ \bibinfo {author} {\bibfnamefont {R.}~\bibnamefont {Stocker}},\ }\bibfield  {title} {\enquote {\bibinfo {title} {Phytoplankton can actively diversify their migration strategy in response to turbulent cues},}\ }\href {\doibase 10.1038/nature21415} {\bibfield  {journal} {\bibinfo  {journal} {Nature}\ }\textbf {\bibinfo {volume} {543}},\ \bibinfo {pages} {555--558} (\bibinfo {year} {2017})}\BibitemShut {NoStop}%
\bibitem [{\citenamefont {Morawska}\ \emph {et~al.}(2022)\citenamefont {Morawska}, \citenamefont {Hernandez-Valdes},\ and\ \citenamefont {Kuipers}}]{Morawska2022}%
  \BibitemOpen
  \bibfield  {author} {\bibinfo {author} {\bibfnamefont {L.~P.}\ \bibnamefont {Morawska}}, \bibinfo {author} {\bibfnamefont {J.~A.}\ \bibnamefont {Hernandez-Valdes}}, \ and\ \bibinfo {author} {\bibfnamefont {O.~P.}\ \bibnamefont {Kuipers}},\ }\bibfield  {title} {\enquote {\bibinfo {title} {Diversity of bet-hedging strategies in microbial communities—recent cases and insights},}\ }\href {\doibase https://doi.org/10.1002/wsbm.1544} {\bibfield  {journal} {\bibinfo  {journal} {WIREs Mechanisms of Disease}\ }\textbf {\bibinfo {volume} {14}},\ \bibinfo {pages} {e1544} (\bibinfo {year} {2022})}\BibitemShut {NoStop}%
\bibitem [{\citenamefont {Hoffmann}\ and\ \citenamefont {Sgr{\`o}}(2011)}]{Hoffmann2011}%
  \BibitemOpen
  \bibfield  {author} {\bibinfo {author} {\bibfnamefont {A.~A.}\ \bibnamefont {Hoffmann}}\ and\ \bibinfo {author} {\bibfnamefont {C.~M.}\ \bibnamefont {Sgr{\`o}}},\ }\bibfield  {title} {\enquote {\bibinfo {title} {Climate change and evolutionary adaptation},}\ }\href {\doibase 10.1038/nature09670} {\bibfield  {journal} {\bibinfo  {journal} {Nature}\ }\textbf {\bibinfo {volume} {470}},\ \bibinfo {pages} {479--485} (\bibinfo {year} {2011})}\BibitemShut {NoStop}%
\bibitem [{\citenamefont {Thomas}\ \emph {et~al.}(2012)\citenamefont {Thomas}, \citenamefont {Kremer}, \citenamefont {Klausmeier},\ and\ \citenamefont {Litchman}}]{thomas2012}%
  \BibitemOpen
  \bibfield  {author} {\bibinfo {author} {\bibfnamefont {M.K.}\ \bibnamefont {Thomas}}, \bibinfo {author} {\bibfnamefont {C.T.}\ \bibnamefont {Kremer}}, \bibinfo {author} {\bibfnamefont {C.A.}\ \bibnamefont {Klausmeier}}, \ and\ \bibinfo {author} {\bibfnamefont {E.}~\bibnamefont {Litchman}},\ }\bibfield  {title} {\enquote {\bibinfo {title} {A global pattern of thermal adaptation in marine phytoplankton},}\ }\href {\doibase 10.1126/science.1224836} {\bibfield  {journal} {\bibinfo  {journal} {Science}\ }\textbf {\bibinfo {volume} {338}},\ \bibinfo {pages} {1085--1088} (\bibinfo {year} {2012})}\BibitemShut {NoStop}%
\bibitem [{\citenamefont {Moritz}\ and\ \citenamefont {Agudo}(2013)}]{Moritz2013}%
  \BibitemOpen
  \bibfield  {author} {\bibinfo {author} {\bibfnamefont {C.}~\bibnamefont {Moritz}}\ and\ \bibinfo {author} {\bibfnamefont {R.}~\bibnamefont {Agudo}},\ }\bibfield  {title} {\enquote {\bibinfo {title} {The future of species under climate change: Resilience or decline?}}\ }\href {\doibase 10.1126/science.1237190} {\bibfield  {journal} {\bibinfo  {journal} {Science}\ }\textbf {\bibinfo {volume} {341}},\ \bibinfo {pages} {504--508} (\bibinfo {year} {2013})}\BibitemShut {NoStop}%
\bibitem [{\citenamefont {Carrara}\ \emph {et~al.}(2021)\citenamefont {Carrara}, \citenamefont {Sengupta}, \citenamefont {Behrendt}, \citenamefont {Vardi},\ and\ \citenamefont {Stocker}}]{carrara2021}%
  \BibitemOpen
  \bibfield  {author} {\bibinfo {author} {\bibfnamefont {F.}~\bibnamefont {Carrara}}, \bibinfo {author} {\bibfnamefont {A.}~\bibnamefont {Sengupta}}, \bibinfo {author} {\bibfnamefont {L.}~\bibnamefont {Behrendt}}, \bibinfo {author} {\bibfnamefont {A.}~\bibnamefont {Vardi}}, \ and\ \bibinfo {author} {\bibfnamefont {R.}~\bibnamefont {Stocker}},\ }\bibfield  {title} {\enquote {\bibinfo {title} {Bistability in oxidative stress response determines the migration behavior of phytoplankton in turbulence},}\ }\href {\doibase 10.1073/pnas.2005944118} {\bibfield  {journal} {\bibinfo  {journal} {Proceedings of the National Academy of Sciences}\ }\textbf {\bibinfo {volume} {118}},\ \bibinfo {pages} {e2005944118} (\bibinfo {year} {2021})}\BibitemShut {NoStop}%
\bibitem [{\citenamefont {Sengupta}\ \emph {et~al.}(2022)\citenamefont {Sengupta}, \citenamefont {Dhar}, \citenamefont {Danza}, \citenamefont {Ghoshal}, \citenamefont {Müller},\ and\ \citenamefont {Kakavand}}]{sengupta2022}%
  \BibitemOpen
  \bibfield  {author} {\bibinfo {author} {\bibfnamefont {A.}~\bibnamefont {Sengupta}}, \bibinfo {author} {\bibfnamefont {J.}~\bibnamefont {Dhar}}, \bibinfo {author} {\bibfnamefont {F.}~\bibnamefont {Danza}}, \bibinfo {author} {\bibfnamefont {A.}~\bibnamefont {Ghoshal}}, \bibinfo {author} {\bibfnamefont {S.}~\bibnamefont {Müller}}, \ and\ \bibinfo {author} {\bibfnamefont {N.}~\bibnamefont {Kakavand}},\ }\bibfield  {title} {\enquote {\bibinfo {title} {Active reconfiguration of cytoplasmic lipid droplets governs migration of nutrient-limited phytoplankton},}\ }\href {\doibase 10.1126/sciadv.abn6005} {\bibfield  {journal} {\bibinfo  {journal} {Science Advances}\ }\textbf {\bibinfo {volume} {8}},\ \bibinfo {pages} {eabn6005} (\bibinfo {year} {2022})}\BibitemShut {NoStop}%
\bibitem [{\citenamefont {Barrett}\ and\ \citenamefont {Schluter}(2008)}]{barrett_adaptation_2008}%
  \BibitemOpen
  \bibfield  {author} {\bibinfo {author} {\bibfnamefont {R.~D.~H.}\ \bibnamefont {Barrett}}\ and\ \bibinfo {author} {\bibfnamefont {D.}~\bibnamefont {Schluter}},\ }\bibfield  {title} {\enquote {\bibinfo {title} {Adaptation from standing genetic variation},}\ }\href {\doibase 10.1016/j.tree.2007.09.008} {\bibfield  {journal} {\bibinfo  {journal} {Trends in Ecology \& Evolution}\ }\textbf {\bibinfo {volume} {23}},\ \bibinfo {pages} {38--44} (\bibinfo {year} {2008})}\BibitemShut {NoStop}%
\bibitem [{\citenamefont {Simons}(2011)}]{simons_modes_2011}%
  \BibitemOpen
  \bibfield  {author} {\bibinfo {author} {\bibfnamefont {A.~M.}\ \bibnamefont {Simons}},\ }\bibfield  {title} {\enquote {\bibinfo {title} {Modes of response to environmental change and the elusive empirical evidence for bet hedging},}\ }\href {\doibase 10.1098/rspb.2011.0176} {\bibfield  {journal} {\bibinfo  {journal} {Proceedings of the Royal Society B: Biological Sciences}\ }\textbf {\bibinfo {volume} {278}},\ \bibinfo {pages} {1601--1609} (\bibinfo {year} {2011})}\BibitemShut {NoStop}%
\bibitem [{\citenamefont {De~Meester}(1996)}]{de1996evolutionary}%
  \BibitemOpen
  \bibfield  {author} {\bibinfo {author} {\bibfnamefont {L.}~\bibnamefont {De~Meester}},\ }\bibfield  {title} {\enquote {\bibinfo {title} {Evolutionary potential and local genetic differentiation in a phenotypically plastic trait of a cyclical parthenogen, daphnia magna},}\ }\href {\doibase 10.1111/j.1558-5646.1996.tb02369.x} {\bibfield  {journal} {\bibinfo  {journal} {Evolution}\ }\textbf {\bibinfo {volume} {50}},\ \bibinfo {pages} {1293--1298} (\bibinfo {year} {1996})}\BibitemShut {NoStop}%
\bibitem [{\citenamefont {West-Eberhard}(2003)}]{Eberhard2003}%
  \BibitemOpen
  \bibfield  {author} {\bibinfo {author} {\bibfnamefont {M.~J.}\ \bibnamefont {West-Eberhard}},\ }\href {\doibase 10.1093/oso/9780195122343.001.0001} {\emph {\bibinfo {title} {{Developmental Plasticity and Evolution}}}}\ (\bibinfo  {publisher} {Oxford University Press},\ \bibinfo {year} {2003})\BibitemShut {NoStop}%
\bibitem [{\citenamefont {Ratikainen}\ and\ \citenamefont {Kokko}(2019)}]{Ratikainen2019}%
  \BibitemOpen
  \bibfield  {author} {\bibinfo {author} {\bibfnamefont {Irja~I.}\ \bibnamefont {Ratikainen}}\ and\ \bibinfo {author} {\bibfnamefont {Hanna}\ \bibnamefont {Kokko}},\ }\bibfield  {title} {\enquote {\bibinfo {title} {The coevolution of lifespan and reversible plasticity},}\ }\href {\doibase 10.1038/s41467-019-08502-9} {\bibfield  {journal} {\bibinfo  {journal} {Nature Communications}\ }\textbf {\bibinfo {volume} {10}},\ \bibinfo {pages} {538} (\bibinfo {year} {2019})}\BibitemShut {NoStop}%
\bibitem [{\citenamefont {Philippi}\ and\ \citenamefont {Seger}(1989)}]{philippi_hedging_1989}%
  \BibitemOpen
  \bibfield  {author} {\bibinfo {author} {\bibfnamefont {T.}~\bibnamefont {Philippi}}\ and\ \bibinfo {author} {\bibfnamefont {J.}~\bibnamefont {Seger}},\ }\bibfield  {title} {\enquote {\bibinfo {title} {Hedging one's evolutionary bets, revisited},}\ }\href {\doibase https://doi.org/10.1016/0169-5347(89)90138-9} {\bibfield  {journal} {\bibinfo  {journal} {Trends in Ecology \& Evolution}\ }\textbf {\bibinfo {volume} {4}},\ \bibinfo {pages} {41--44} (\bibinfo {year} {1989})}\BibitemShut {NoStop}%
\bibitem [{\citenamefont {Ripa}\ \emph {et~al.}(2010)\citenamefont {Ripa}, \citenamefont {Olofsson},\ and\ \citenamefont {Jonzén}}]{ripa2010}%
  \BibitemOpen
  \bibfield  {author} {\bibinfo {author} {\bibfnamefont {J.}~\bibnamefont {Ripa}}, \bibinfo {author} {\bibfnamefont {H.}~\bibnamefont {Olofsson}}, \ and\ \bibinfo {author} {\bibfnamefont {N.}~\bibnamefont {Jonzén}},\ }\bibfield  {title} {\enquote {\bibinfo {title} {What is bet-hedging, really?}}\ }\href {\doibase 10.1098/rspb.2009.2023} {\bibfield  {journal} {\bibinfo  {journal} {Proceedings of the Royal Society B: Biological Sciences}\ }\textbf {\bibinfo {volume} {277}},\ \bibinfo {pages} {1153--1154} (\bibinfo {year} {2010})}\BibitemShut {NoStop}%
\bibitem [{\citenamefont {Simons}\ and\ \citenamefont {Johnston}(2003)}]{simons_suboptimal_2003}%
  \BibitemOpen
  \bibfield  {author} {\bibinfo {author} {\bibfnamefont {A.~M.}\ \bibnamefont {Simons}}\ and\ \bibinfo {author} {\bibfnamefont {M.~O.}\ \bibnamefont {Johnston}},\ }\bibfield  {title} {\enquote {\bibinfo {title} {Suboptimal timing of reproduction in {Lobelia} inflata may be a conservative bet‐hedging strategy},}\ }\href {\doibase 10.1046/j.1420-9101.2003.00530.x} {\bibfield  {journal} {\bibinfo  {journal} {Journal of Evolutionary Biology}\ }\textbf {\bibinfo {volume} {16}},\ \bibinfo {pages} {233--243} (\bibinfo {year} {2003})}\BibitemShut {NoStop}%
\bibitem [{\citenamefont {Acar}\ \emph {et~al.}(2008)\citenamefont {Acar}, \citenamefont {Mettetal},\ and\ \citenamefont {van Oudenaarden}}]{Acar2008}%
  \BibitemOpen
  \bibfield  {author} {\bibinfo {author} {\bibfnamefont {M.}~\bibnamefont {Acar}}, \bibinfo {author} {\bibfnamefont {J.~T.}\ \bibnamefont {Mettetal}}, \ and\ \bibinfo {author} {\bibfnamefont {A.}~\bibnamefont {van Oudenaarden}},\ }\bibfield  {title} {\enquote {\bibinfo {title} {Stochastic switching as a survival strategy in fluctuating environments},}\ }\href {\doibase 10.1038/ng.110} {\bibfield  {journal} {\bibinfo  {journal} {Nature Genetics}\ }\textbf {\bibinfo {volume} {40}},\ \bibinfo {pages} {471--475} (\bibinfo {year} {2008})}\BibitemShut {NoStop}%
\bibitem [{\citenamefont {Veening}\ \emph {et~al.}(2008)\citenamefont {Veening}, \citenamefont {Stewart}, \citenamefont {Berngruber}, \citenamefont {Taddei}, \citenamefont {Kuipers},\ and\ \citenamefont {Hamoen}}]{Veening2008}%
  \BibitemOpen
  \bibfield  {author} {\bibinfo {author} {\bibfnamefont {J.-W.}\ \bibnamefont {Veening}}, \bibinfo {author} {\bibfnamefont {E.~J.}\ \bibnamefont {Stewart}}, \bibinfo {author} {\bibfnamefont {T.~W.}\ \bibnamefont {Berngruber}}, \bibinfo {author} {\bibfnamefont {F.}~\bibnamefont {Taddei}}, \bibinfo {author} {\bibfnamefont {O.~P.}\ \bibnamefont {Kuipers}}, \ and\ \bibinfo {author} {\bibfnamefont {L.~W.}\ \bibnamefont {Hamoen}},\ }\bibfield  {title} {\enquote {\bibinfo {title} {Bet-hedging and epigenetic inheritance in bacterial cell development},}\ }\href {\doibase 10.1073/pnas.0700463105} {\bibfield  {journal} {\bibinfo  {journal} {Proceedings of the National Academy of Sciences}\ }\textbf {\bibinfo {volume} {105}},\ \bibinfo {pages} {4393--4398} (\bibinfo {year} {2008})}\BibitemShut {NoStop}%
\bibitem [{\citenamefont {Beaumont}\ \emph {et~al.}(2009)\citenamefont {Beaumont}, \citenamefont {Gallie}, \citenamefont {Kost}, \citenamefont {Ferguson},\ and\ \citenamefont {Rainey}}]{Beaumont2009}%
  \BibitemOpen
  \bibfield  {author} {\bibinfo {author} {\bibfnamefont {H.~J.~E.}\ \bibnamefont {Beaumont}}, \bibinfo {author} {\bibfnamefont {J.}~\bibnamefont {Gallie}}, \bibinfo {author} {\bibfnamefont {C.}~\bibnamefont {Kost}}, \bibinfo {author} {\bibfnamefont {G.~C.}\ \bibnamefont {Ferguson}}, \ and\ \bibinfo {author} {\bibfnamefont {P.~B.}\ \bibnamefont {Rainey}},\ }\bibfield  {title} {\enquote {\bibinfo {title} {Experimental evolution of bet hedging},}\ }\href {\doibase 10.1038/nature08504} {\bibfield  {journal} {\bibinfo  {journal} {Nature}\ }\textbf {\bibinfo {volume} {462}},\ \bibinfo {pages} {90--93} (\bibinfo {year} {2009})}\BibitemShut {NoStop}%
\bibitem [{\citenamefont {Grimbergen}\ \emph {et~al.}(2015)\citenamefont {Grimbergen}, \citenamefont {Siebring}, \citenamefont {Solopova},\ and\ \citenamefont {Kuipers}}]{Grimbergen2015}%
  \BibitemOpen
  \bibfield  {author} {\bibinfo {author} {\bibfnamefont {A.~J.}\ \bibnamefont {Grimbergen}}, \bibinfo {author} {\bibfnamefont {J.}~\bibnamefont {Siebring}}, \bibinfo {author} {\bibfnamefont {A.}~\bibnamefont {Solopova}}, \ and\ \bibinfo {author} {\bibfnamefont {O.~P.}\ \bibnamefont {Kuipers}},\ }\bibfield  {title} {\enquote {\bibinfo {title} {Microbial bet-hedging: the power of being different},}\ }\href {\doibase https://doi.org/10.1016/j.mib.2015.04.008} {\bibfield  {journal} {\bibinfo  {journal} {Current Opinion in Microbiology}\ }\textbf {\bibinfo {volume} {25}},\ \bibinfo {pages} {67--72} (\bibinfo {year} {2015})}\BibitemShut {NoStop}%
\bibitem [{\citenamefont {Carey}\ and\ \citenamefont {Goulian}(2019)}]{Carey2019}%
  \BibitemOpen
  \bibfield  {author} {\bibinfo {author} {\bibfnamefont {J.~N.}\ \bibnamefont {Carey}}\ and\ \bibinfo {author} {\bibfnamefont {M.}~\bibnamefont {Goulian}},\ }\bibfield  {title} {\enquote {\bibinfo {title} {A bacterial signaling system regulates noise to enable bet hedging},}\ }\href {\doibase 10.1007/s00294-018-0856-2} {\bibfield  {journal} {\bibinfo  {journal} {Current Genetics}\ }\textbf {\bibinfo {volume} {65}},\ \bibinfo {pages} {65--70} (\bibinfo {year} {2019})}\BibitemShut {NoStop}%
\bibitem [{\citenamefont {Villa~Martín}\ \emph {et~al.}(2019)\citenamefont {Villa~Martín}, \citenamefont {Muñoz},\ and\ \citenamefont {Pigolotti}}]{villa2019}%
  \BibitemOpen
  \bibfield  {author} {\bibinfo {author} {\bibfnamefont {P.}~\bibnamefont {Villa~Martín}}, \bibinfo {author} {\bibfnamefont {M.A.}\ \bibnamefont {Muñoz}}, \ and\ \bibinfo {author} {\bibfnamefont {S.}~\bibnamefont {Pigolotti}},\ }\bibfield  {title} {\enquote {\bibinfo {title} {Bet-hedging strategies in expanding populations},}\ }\href {\doibase 10.1371/journal.pcbi.1006529} {\bibfield  {journal} {\bibinfo  {journal} {PLoS Computational Biology}\ }\textbf {\bibinfo {volume} {15}},\ \bibinfo {pages} {e1006529} (\bibinfo {year} {2019})}\BibitemShut {NoStop}%
\bibitem [{\citenamefont {Solopova}\ \emph {et~al.}(2014)\citenamefont {Solopova}, \citenamefont {van Gestel}, \citenamefont {Weissing}, \citenamefont {Bachmann}, \citenamefont {Teusink}, \citenamefont {Kok},\ and\ \citenamefont {Kuipers}}]{solopova_bet-hedging_2014}%
  \BibitemOpen
  \bibfield  {author} {\bibinfo {author} {\bibfnamefont {A.}~\bibnamefont {Solopova}}, \bibinfo {author} {\bibfnamefont {J.}~\bibnamefont {van Gestel}}, \bibinfo {author} {\bibfnamefont {F.~J.}\ \bibnamefont {Weissing}}, \bibinfo {author} {\bibfnamefont {H.}~\bibnamefont {Bachmann}}, \bibinfo {author} {\bibfnamefont {B.}~\bibnamefont {Teusink}}, \bibinfo {author} {\bibfnamefont {J.}~\bibnamefont {Kok}}, \ and\ \bibinfo {author} {\bibfnamefont {O.~P.}\ \bibnamefont {Kuipers}},\ }\bibfield  {title} {\enquote {\bibinfo {title} {Bet-hedging during bacterial diauxic shift},}\ }\href {\doibase 10.1073/pnas.1320063111} {\bibfield  {journal} {\bibinfo  {journal} {Proceedings of the National Academy of Sciences}\ }\textbf {\bibinfo {volume} {111}},\ \bibinfo {pages} {7427--7432} (\bibinfo {year} {2014})}\BibitemShut {NoStop}%
\bibitem [{\citenamefont {Kotte}\ \emph {et~al.}(2014)\citenamefont {Kotte}, \citenamefont {Volkmer}, \citenamefont {Radzikowski},\ and\ \citenamefont {Heinemann}}]{kotte_phenotypic_2014}%
  \BibitemOpen
  \bibfield  {author} {\bibinfo {author} {\bibfnamefont {O.}~\bibnamefont {Kotte}}, \bibinfo {author} {\bibfnamefont {B.}~\bibnamefont {Volkmer}}, \bibinfo {author} {\bibfnamefont {J.~L.}\ \bibnamefont {Radzikowski}}, \ and\ \bibinfo {author} {\bibfnamefont {M.}~\bibnamefont {Heinemann}},\ }\bibfield  {title} {\enquote {\bibinfo {title} {Phenotypic bistability in {Escherichia} coli's central carbon metabolism},}\ }\href {\doibase 10.15252/msb.20135022} {\bibfield  {journal} {\bibinfo  {journal} {Molecular Systems Biology}\ }\textbf {\bibinfo {volume} {10}},\ \bibinfo {pages} {736} (\bibinfo {year} {2014})}\BibitemShut {NoStop}%
\bibitem [{\citenamefont {Setlow}(2007)}]{setlow2007will}%
  \BibitemOpen
  \bibfield  {author} {\bibinfo {author} {\bibfnamefont {P.}~\bibnamefont {Setlow}},\ }\bibfield  {title} {\enquote {\bibinfo {title} {I will survive: Dna protection in bacterial spores},}\ }\href {\doibase https://doi.org/10.1016/j.tim.2007.02.004} {\bibfield  {journal} {\bibinfo  {journal} {Trends in Microbiology}\ }\textbf {\bibinfo {volume} {15}},\ \bibinfo {pages} {172--180} (\bibinfo {year} {2007})}\BibitemShut {NoStop}%
\bibitem [{\citenamefont {Jin}\ and\ \citenamefont {Sengupta}(2024)}]{jin2024}%
  \BibitemOpen
  \bibfield  {author} {\bibinfo {author} {\bibfnamefont {C.}~\bibnamefont {Jin}}\ and\ \bibinfo {author} {\bibfnamefont {A.}~\bibnamefont {Sengupta}},\ }\bibfield  {title} {\enquote {\bibinfo {title} {Microbes in porous environments: from active interactions to emergent feedback},}\ }\href {\doibase 10.1007/s12551-024-01185-7} {\bibfield  {journal} {\bibinfo  {journal} {Biophysical Reviews}\ }\textbf {\bibinfo {volume} {16}},\ \bibinfo {pages} {1--16} (\bibinfo {year} {2024})}\BibitemShut {NoStop}%
\bibitem [{\citenamefont {Hubert}\ \emph {et~al.}(2024)\citenamefont {Hubert}, \citenamefont {Tabuteau}, \citenamefont {Farasin}, \citenamefont {Loncar}, \citenamefont {Dufresne}, \citenamefont {Méheust},\ and\ \citenamefont {Le~Borgne}}]{hubert2024}%
  \BibitemOpen
  \bibfield  {author} {\bibinfo {author} {\bibfnamefont {A.}~\bibnamefont {Hubert}}, \bibinfo {author} {\bibfnamefont {H.}~\bibnamefont {Tabuteau}}, \bibinfo {author} {\bibfnamefont {J.}~\bibnamefont {Farasin}}, \bibinfo {author} {\bibfnamefont {A.}~\bibnamefont {Loncar}}, \bibinfo {author} {\bibfnamefont {A.}~\bibnamefont {Dufresne}}, \bibinfo {author} {\bibfnamefont {Y.}~\bibnamefont {Méheust}}, \ and\ \bibinfo {author} {\bibfnamefont {T.}~\bibnamefont {Le~Borgne}},\ }\bibfield  {title} {\enquote {\bibinfo {title} {Fluid flow drives phenotypic heterogeneity in bacterial growth and adhesion on surfaces},}\ }\href {\doibase 10.1038/s41467-024-49997-1} {\bibfield  {journal} {\bibinfo  {journal} {Nature Communications}\ }\textbf {\bibinfo {volume} {15}},\ \bibinfo {pages} {6161} (\bibinfo {year} {2024})}\BibitemShut {NoStop}%
\bibitem [{\citenamefont {Gasperotti}\ \emph {et~al.}(2020)\citenamefont {Gasperotti}, \citenamefont {Brameyer}, \citenamefont {Fabiani},\ and\ \citenamefont {Jung}}]{Gasperotti2020}%
  \BibitemOpen
  \bibfield  {author} {\bibinfo {author} {\bibfnamefont {A.}~\bibnamefont {Gasperotti}}, \bibinfo {author} {\bibfnamefont {S.}~\bibnamefont {Brameyer}}, \bibinfo {author} {\bibfnamefont {F.}~\bibnamefont {Fabiani}}, \ and\ \bibinfo {author} {\bibfnamefont {K.}~\bibnamefont {Jung}},\ }\bibfield  {title} {\enquote {\bibinfo {title} {Phenotypic heterogeneity of microbial populations under nutrient limitation},}\ }\href {\doibase https://doi.org/10.1016/j.copbio.2019.09.016} {\bibfield  {journal} {\bibinfo  {journal} {Current Opinion in Biotechnology}\ }\textbf {\bibinfo {volume} {62}},\ \bibinfo {pages} {160--167} (\bibinfo {year} {2020})}\BibitemShut {NoStop}%
\bibitem [{\citenamefont {Di~Nezio}\ \emph {et~al.}(2023)\citenamefont {Di~Nezio}, \citenamefont {Ong}, \citenamefont {Riedel}, \citenamefont {Ghoshal}, \citenamefont {Dhar}, \citenamefont {Roman},\ and\ \citenamefont {Sengupta}}]{dinezio2023}%
  \BibitemOpen
  \bibfield  {author} {\bibinfo {author} {\bibfnamefont {F.}~\bibnamefont {Di~Nezio}}, \bibinfo {author} {\bibfnamefont {I.L.H.}\ \bibnamefont {Ong}}, \bibinfo {author} {\bibfnamefont {R.}~\bibnamefont {Riedel}}, \bibinfo {author} {\bibfnamefont {A.}~\bibnamefont {Ghoshal}}, \bibinfo {author} {\bibfnamefont {J.}~\bibnamefont {Dhar}}, \bibinfo {author} {\bibfnamefont {N.}~\bibnamefont {Roman}, \bibfnamefont {S.and~Storelli}}, \ and\ \bibinfo {author} {\bibfnamefont {A.}~\bibnamefont {Sengupta}},\ }\bibfield  {title} {\enquote {\bibinfo {title} {Synergistic phenotypic shifts during domestication promote plankton-to-biofilm transition in purple sulfur bacterium chromatium okenii},}\ }\href {\doibase 10.1101/2023.10.20.563228} {\bibfield  {journal} {\bibinfo  {journal} {bioRxiv}\ ,\ \bibinfo {pages} {2023--10}} (\bibinfo {year} {2023})}\BibitemShut {NoStop}%
\bibitem [{\citenamefont {de~Jong}\ \emph {et~al.}(2011)\citenamefont {de~Jong}, \citenamefont {Haccou},\ and\ \citenamefont {Kuipers}}]{Jong2011}%
  \BibitemOpen
  \bibfield  {author} {\bibinfo {author} {\bibfnamefont {I.~G.}\ \bibnamefont {de~Jong}}, \bibinfo {author} {\bibfnamefont {P.}~\bibnamefont {Haccou}}, \ and\ \bibinfo {author} {\bibfnamefont {O.~P.}\ \bibnamefont {Kuipers}},\ }\bibfield  {title} {\enquote {\bibinfo {title} {Bet hedging or not? a guide to proper classification of microbial survival strategies},}\ }\href {\doibase https://doi.org/10.1002/bies.201000127} {\bibfield  {journal} {\bibinfo  {journal} {BioEssays}\ }\textbf {\bibinfo {volume} {33}},\ \bibinfo {pages} {215--223} (\bibinfo {year} {2011})}\BibitemShut {NoStop}%
\bibitem [{\citenamefont {Xue}\ \emph {et~al.}(2019)\citenamefont {Xue}, \citenamefont {Sartori},\ and\ \citenamefont {Leibler}}]{Xue2019}%
  \BibitemOpen
  \bibfield  {author} {\bibinfo {author} {\bibfnamefont {B.~K.}\ \bibnamefont {Xue}}, \bibinfo {author} {\bibfnamefont {P.}~\bibnamefont {Sartori}}, \ and\ \bibinfo {author} {\bibfnamefont {S.}~\bibnamefont {Leibler}},\ }\bibfield  {title} {\enquote {\bibinfo {title} {Environment-to-phenotype mapping and adaptation strategies in varying environments},}\ }\href {\doibase 10.1073/pnas.1903232116} {\bibfield  {journal} {\bibinfo  {journal} {Proceedings of the National Academy of Sciences}\ }\textbf {\bibinfo {volume} {116}},\ \bibinfo {pages} {13847--13855} (\bibinfo {year} {2019})}\BibitemShut {NoStop}%
\bibitem [{\citenamefont {Wong}\ and\ \citenamefont {Ackerly}(2005)}]{Wong2005}%
  \BibitemOpen
  \bibfield  {author} {\bibinfo {author} {\bibfnamefont {T.~G.}\ \bibnamefont {Wong}}\ and\ \bibinfo {author} {\bibfnamefont {D.~D.}\ \bibnamefont {Ackerly}},\ }\bibfield  {title} {\enquote {\bibinfo {title} {Optimal reproductive allocation in annuals and an informational constraint on plasticity},}\ }\href {\doibase https://doi.org/10.1111/j.1469-8137.2005.01375.x} {\bibfield  {journal} {\bibinfo  {journal} {New Phytologist}\ }\textbf {\bibinfo {volume} {166}},\ \bibinfo {pages} {159--172} (\bibinfo {year} {2005})}\BibitemShut {NoStop}%
\bibitem [{\citenamefont {Relyea}(2002)}]{Relyea2002}%
  \BibitemOpen
  \bibfield  {author} {\bibinfo {author} {\bibfnamefont {R.~A.}\ \bibnamefont {Relyea}},\ }\bibfield  {title} {\enquote {\bibinfo {title} {Costs of phenotypic plasticity.}}\ }\href {\doibase 10.1086/338540} {\bibfield  {journal} {\bibinfo  {journal} {The American Naturalist}\ }\textbf {\bibinfo {volume} {159}},\ \bibinfo {pages} {272--282} (\bibinfo {year} {2002})}\BibitemShut {NoStop}%
\bibitem [{\citenamefont {Gilbert}\ and\ \citenamefont {McPeek}(2013)}]{Gilbert2013}%
  \BibitemOpen
  \bibfield  {author} {\bibinfo {author} {\bibfnamefont {John~J.}\ \bibnamefont {Gilbert}}\ and\ \bibinfo {author} {\bibfnamefont {Mark~A.}\ \bibnamefont {McPeek}},\ }\bibfield  {title} {\enquote {\bibinfo {title} {Maternal age and spine development in a rotifer: ecological implications and evolution},}\ }\href {\doibase https://doi.org/10.1890/13-0768.1} {\bibfield  {journal} {\bibinfo  {journal} {Ecology}\ }\textbf {\bibinfo {volume} {94}},\ \bibinfo {pages} {2166--2172} (\bibinfo {year} {2013})}\BibitemShut {NoStop}%
\bibitem [{\citenamefont {Furness}\ \emph {et~al.}(2015)\citenamefont {Furness}, \citenamefont {Lee},\ and\ \citenamefont {Reznick}}]{Furness2015}%
  \BibitemOpen
  \bibfield  {author} {\bibinfo {author} {\bibfnamefont {A.~I.}\ \bibnamefont {Furness}}, \bibinfo {author} {\bibfnamefont {K.}~\bibnamefont {Lee}}, \ and\ \bibinfo {author} {\bibfnamefont {D.~N.}\ \bibnamefont {Reznick}},\ }\bibfield  {title} {\enquote {\bibinfo {title} {{Adaptation in a variable environment: Phenotypic plasticity and bet-hedging during egg diapause and hatching in an annual killifish}},}\ }\href {\doibase 10.1111/evo.12669} {\bibfield  {journal} {\bibinfo  {journal} {Evolution}\ }\textbf {\bibinfo {volume} {69}},\ \bibinfo {pages} {1461--1475} (\bibinfo {year} {2015})}\BibitemShut {NoStop}%
\bibitem [{\citenamefont {Grantham}\ \emph {et~al.}(2016)\citenamefont {Grantham}, \citenamefont {Antonio}, \citenamefont {O'Neil}, \citenamefont {Zhan},\ and\ \citenamefont {Brisson}}]{Grantham2016}%
  \BibitemOpen
  \bibfield  {author} {\bibinfo {author} {\bibfnamefont {M.~E.}\ \bibnamefont {Grantham}}, \bibinfo {author} {\bibfnamefont {C.~J.}\ \bibnamefont {Antonio}}, \bibinfo {author} {\bibfnamefont {B.~R.}\ \bibnamefont {O'Neil}}, \bibinfo {author} {\bibfnamefont {Y.~X.}\ \bibnamefont {Zhan}}, \ and\ \bibinfo {author} {\bibfnamefont {J.~A.}\ \bibnamefont {Brisson}},\ }\bibfield  {title} {\enquote {\bibinfo {title} {A case for a joint strategy of diversified bet hedging and plasticity in the pea aphid wing polyphenism},}\ }\href {\doibase 10.1098/rsbl.2016.0654} {\bibfield  {journal} {\bibinfo  {journal} {Biology Letters}\ }\textbf {\bibinfo {volume} {12}},\ \bibinfo {pages} {20160654} (\bibinfo {year} {2016})}\BibitemShut {NoStop}%
\bibitem [{\citenamefont {Donaldson-Matasci}\ \emph {et~al.}(2013)\citenamefont {Donaldson-Matasci}, \citenamefont {Bergstrom},\ and\ \citenamefont {Lachmann}}]{Donaldson2013}%
  \BibitemOpen
  \bibfield  {author} {\bibinfo {author} {\bibfnamefont {M.~C.}\ \bibnamefont {Donaldson-Matasci}}, \bibinfo {author} {\bibfnamefont {C.~T.}\ \bibnamefont {Bergstrom}}, \ and\ \bibinfo {author} {\bibfnamefont {M.}~\bibnamefont {Lachmann}},\ }\bibfield  {title} {\enquote {\bibinfo {title} {When unreliable cues are good enough.}}\ }\href {\doibase 10.1086/671161} {\bibfield  {journal} {\bibinfo  {journal} {The American Naturalist}\ }\textbf {\bibinfo {volume} {182}},\ \bibinfo {pages} {313--327} (\bibinfo {year} {2013})},\ \bibinfo {note} {pMID: 23933723}\BibitemShut {NoStop}%
\bibitem [{\citenamefont {Reed}\ \emph {et~al.}(2010)\citenamefont {Reed}, \citenamefont {Waples}, \citenamefont {Schindler}, \citenamefont {Hard},\ and\ \citenamefont {Kinnison}}]{Reed2010}%
  \BibitemOpen
  \bibfield  {author} {\bibinfo {author} {\bibfnamefont {T.~E.}\ \bibnamefont {Reed}}, \bibinfo {author} {\bibfnamefont {R.~S.}\ \bibnamefont {Waples}}, \bibinfo {author} {\bibfnamefont {D.~E.}\ \bibnamefont {Schindler}}, \bibinfo {author} {\bibfnamefont {J.~J.}\ \bibnamefont {Hard}}, \ and\ \bibinfo {author} {\bibfnamefont {M.~T.}\ \bibnamefont {Kinnison}},\ }\bibfield  {title} {\enquote {\bibinfo {title} {Phenotypic plasticity and population viability: the importance of environmental predictability},}\ }\href {\doibase 10.1098/rspb.2010.0771} {\bibfield  {journal} {\bibinfo  {journal} {Proceedings of the Royal Society B: Biological Sciences}\ }\textbf {\bibinfo {volume} {277}},\ \bibinfo {pages} {3391--3400} (\bibinfo {year} {2010})}\BibitemShut {NoStop}%
\bibitem [{\citenamefont {Draghi}(2023)}]{Draghi2023}%
  \BibitemOpen
  \bibfield  {author} {\bibinfo {author} {\bibfnamefont {J.~A.}\ \bibnamefont {Draghi}},\ }\bibfield  {title} {\enquote {\bibinfo {title} {{Bet‐hedging via dispersal aids the evolution of plastic responses to unreliable cues}},}\ }\href {\doibase 10.1111/jeb.14182} {\bibfield  {journal} {\bibinfo  {journal} {Journal of Evolutionary Biology}\ }\textbf {\bibinfo {volume} {36}},\ \bibinfo {pages} {893--905} (\bibinfo {year} {2023})}\BibitemShut {NoStop}%
\bibitem [{\citenamefont {Joschinski}\ and\ \citenamefont {Bonte}(2020)}]{Joschinski2020}%
  \BibitemOpen
  \bibfield  {author} {\bibinfo {author} {\bibfnamefont {J.}~\bibnamefont {Joschinski}}\ and\ \bibinfo {author} {\bibfnamefont {D.}~\bibnamefont {Bonte}},\ }\bibfield  {title} {\enquote {\bibinfo {title} {Transgenerational plasticity and bet-hedging: A framework for reaction norm evolution},}\ }\href {\doibase 10.3389/fevo.2020.517183} {\bibfield  {journal} {\bibinfo  {journal} {Frontiers in Ecology and Evolution}\ }\textbf {\bibinfo {volume} {8}} (\bibinfo {year} {2020}),\ 10.3389/fevo.2020.517183}\BibitemShut {NoStop}%
\bibitem [{\citenamefont {Starrfelt}\ and\ \citenamefont {Kokko}(2012)}]{starrfelt_bet-hedgingtriple_nodate}%
  \BibitemOpen
  \bibfield  {author} {\bibinfo {author} {\bibfnamefont {J.}~\bibnamefont {Starrfelt}}\ and\ \bibinfo {author} {\bibfnamefont {H.}~\bibnamefont {Kokko}},\ }\bibfield  {title} {\enquote {\bibinfo {title} {Bet-hedging—a triple trade-off between means, variances and correlations},}\ }\href {\doibase https://doi.org/10.1111/j.1469-185X.2012.00225.x} {\bibfield  {journal} {\bibinfo  {journal} {Biological Reviews}\ }\textbf {\bibinfo {volume} {87}},\ \bibinfo {pages} {742--755} (\bibinfo {year} {2012})}\BibitemShut {NoStop}%
\bibitem [{\citenamefont {Haaland}\ \emph {et~al.}(2019)\citenamefont {Haaland}, \citenamefont {Wright}, \citenamefont {Tufto},\ and\ \citenamefont {Ratikainen}}]{haaland2019short}%
  \BibitemOpen
  \bibfield  {author} {\bibinfo {author} {\bibfnamefont {T.~R.}\ \bibnamefont {Haaland}}, \bibinfo {author} {\bibfnamefont {J.}~\bibnamefont {Wright}}, \bibinfo {author} {\bibfnamefont {J.}~\bibnamefont {Tufto}}, \ and\ \bibinfo {author} {\bibfnamefont {I.~I.}\ \bibnamefont {Ratikainen}},\ }\bibfield  {title} {\enquote {\bibinfo {title} {Short-term insurance versus long-term bet-hedging strategies as adaptations to variable environments},}\ }\href {\doibase https://doi.org/10.1111/evo.13659} {\bibfield  {journal} {\bibinfo  {journal} {Evolution}\ }\textbf {\bibinfo {volume} {73}},\ \bibinfo {pages} {145--157} (\bibinfo {year} {2019})}\BibitemShut {NoStop}%
\bibitem [{\citenamefont {Libby}\ and\ \citenamefont {Ratcliff}(2019)}]{Libby2019}%
  \BibitemOpen
  \bibfield  {author} {\bibinfo {author} {\bibfnamefont {E.}~\bibnamefont {Libby}}\ and\ \bibinfo {author} {\bibfnamefont {W.~C.}\ \bibnamefont {Ratcliff}},\ }\bibfield  {title} {\enquote {\bibinfo {title} {Shortsighted evolution constrains the efficacy of long-term bet hedging},}\ }\href {\doibase 10.1086/701786} {\bibfield  {journal} {\bibinfo  {journal} {The American Naturalist}\ }\textbf {\bibinfo {volume} {193}},\ \bibinfo {pages} {409--423} (\bibinfo {year} {2019})}\BibitemShut {NoStop}%
\bibitem [{\citenamefont {J.~Seger}(1987)}]{Seger1987}%
  \BibitemOpen
  \bibfield  {author} {\bibinfo {author} {\bibfnamefont {H.~J.~Brockman}\ \bibnamefont {J.~Seger}},\ }\bibfield  {title} {\enquote {\bibinfo {title} {What is bet-hedging?}}\ }\href {https://www.cabidigitallibrary.org/doi/full/10.5555/19880164099} {\bibfield  {journal} {\bibinfo  {journal} {Oxford Surveys in Evolutionary Biology}\ }\textbf {\bibinfo {volume} {4}},\ \bibinfo {pages} {182--211} (\bibinfo {year} {1987})}\BibitemShut {NoStop}%
\bibitem [{\citenamefont {Gillespie}(1974)}]{Gillespie1974}%
  \BibitemOpen
  \bibfield  {author} {\bibinfo {author} {\bibfnamefont {J.~H.}\ \bibnamefont {Gillespie}},\ }\bibfield  {title} {\enquote {\bibinfo {title} {Natural selection for within-generation variance in offspring number},}\ }\href {\doibase 10.1093/genetics/76.3.601} {\bibfield  {journal} {\bibinfo  {journal} {Genetics}\ }\textbf {\bibinfo {volume} {76}},\ \bibinfo {pages} {601--606} (\bibinfo {year} {1974})}\BibitemShut {NoStop}%
\bibitem [{\citenamefont {Kussell}\ \emph {et~al.}(2005)\citenamefont {Kussell}, \citenamefont {Kishony}, \citenamefont {Balaban},\ and\ \citenamefont {Leibler}}]{Kussell2005}%
  \BibitemOpen
  \bibfield  {author} {\bibinfo {author} {\bibfnamefont {E.}~\bibnamefont {Kussell}}, \bibinfo {author} {\bibfnamefont {R.}~\bibnamefont {Kishony}}, \bibinfo {author} {\bibfnamefont {N.~Q.}\ \bibnamefont {Balaban}}, \ and\ \bibinfo {author} {\bibfnamefont {S.}~\bibnamefont {Leibler}},\ }\bibfield  {title} {\enquote {\bibinfo {title} {{Bacterial Persistence: A Model of Survival in Changing Environments}},}\ }\href {\doibase 10.1534/genetics.104.035352} {\bibfield  {journal} {\bibinfo  {journal} {Genetics}\ }\textbf {\bibinfo {volume} {169}},\ \bibinfo {pages} {1807--1814} (\bibinfo {year} {2005})}\BibitemShut {NoStop}%
\bibitem [{\citenamefont {Brock}\ \emph {et~al.}(2009)\citenamefont {Brock}, \citenamefont {Hommes},\ and\ \citenamefont {Wagener}}]{Brock2009}%
  \BibitemOpen
  \bibfield  {author} {\bibinfo {author} {\bibfnamefont {W.~A.}\ \bibnamefont {Brock}}, \bibinfo {author} {\bibfnamefont {C.~H.}\ \bibnamefont {Hommes}}, \ and\ \bibinfo {author} {\bibfnamefont {F.~O.~O.}\ \bibnamefont {Wagener}},\ }\bibfield  {title} {\enquote {\bibinfo {title} {More hedging instruments may destabilize markets},}\ }\href {\doibase https://doi.org/10.1016/j.jedc.2009.05.004} {\bibfield  {journal} {\bibinfo  {journal} {Journal of Economic Dynamics and Control}\ }\textbf {\bibinfo {volume} {33}},\ \bibinfo {pages} {1912--1928} (\bibinfo {year} {2009})}\BibitemShut {NoStop}%
\bibitem [{\citenamefont {Brock}\ and\ \citenamefont {Hommes}(1998)}]{Brock1998}%
  \BibitemOpen
  \bibfield  {author} {\bibinfo {author} {\bibfnamefont {W.~A.}\ \bibnamefont {Brock}}\ and\ \bibinfo {author} {\bibfnamefont {C.~H.}\ \bibnamefont {Hommes}},\ }\bibfield  {title} {\enquote {\bibinfo {title} {Heterogeneous beliefs and routes to chaos in a simple asset pricing model},}\ }\href {\doibase https://doi.org/10.1016/S0165-1889(98)00011-6} {\bibfield  {journal} {\bibinfo  {journal} {Journal of Economic Dynamics and Control}\ }\textbf {\bibinfo {volume} {22}},\ \bibinfo {pages} {1235--1274} (\bibinfo {year} {1998})}\BibitemShut {NoStop}%
\bibitem [{\citenamefont {Leslie}(1945)}]{Leslie1945}%
  \BibitemOpen
  \bibfield  {author} {\bibinfo {author} {\bibfnamefont {P.~H.}\ \bibnamefont {Leslie}},\ }\bibfield  {title} {\enquote {\bibinfo {title} {{On the Use of Matrices in Certain Populations Mathematics}},}\ }\href {\doibase 10.1093/biomet/33.3.183} {\bibfield  {journal} {\bibinfo  {journal} {Biometrika}\ }\textbf {\bibinfo {volume} {33}},\ \bibinfo {pages} {183--212} (\bibinfo {year} {1945})}\BibitemShut {NoStop}%
\bibitem [{\citenamefont {Caswell}(2001)}]{Caswell2001}%
  \BibitemOpen
  \bibfield  {author} {\bibinfo {author} {\bibfnamefont {H.}~\bibnamefont {Caswell}},\ }\href {https://global.oup.com/academic/product/matrix-population-models-9780878931217?cc=in&lang=en&#} {\emph {\bibinfo {title} {Matrix population models : construction, analysis, and interpretation}}},\ \bibinfo {edition} {2nd}\ ed.\ (\bibinfo  {publisher} {Sinauer Associates},\ \bibinfo {address} {Sunderland, Massachusetts},\ \bibinfo {year} {2001})\BibitemShut {NoStop}%
\bibitem [{\citenamefont {Dempster}(1955)}]{Dempster1955}%
  \BibitemOpen
  \bibfield  {author} {\bibinfo {author} {\bibfnamefont {E.~R.}\ \bibnamefont {Dempster}},\ }\bibfield  {title} {\enquote {\bibinfo {title} {Maintenance of genetic heterogeneity},}\ }\href {\doibase 10.1101/SQB.1955.020.01.005} {\bibfield  {journal} {\bibinfo  {journal} {Cold Spring Harbor Symposia on Quantitative Biology}\ }\textbf {\bibinfo {volume} {20}},\ \bibinfo {pages} {25--32} (\bibinfo {year} {1955})}\BibitemShut {NoStop}%
\bibitem [{\citenamefont {Haldane}\ and\ \citenamefont {Jayakar}(1963)}]{Haldane1963}%
  \BibitemOpen
  \bibfield  {author} {\bibinfo {author} {\bibfnamefont {J.~B.~S.}\ \bibnamefont {Haldane}}\ and\ \bibinfo {author} {\bibfnamefont {S.~D.}\ \bibnamefont {Jayakar}},\ }\bibfield  {title} {\enquote {\bibinfo {title} {Polymorphism due to selection of varying direction},}\ }\href {\doibase 10.1007/BF02986143} {\bibfield  {journal} {\bibinfo  {journal} {Journal of Genetics}\ }\textbf {\bibinfo {volume} {58}},\ \bibinfo {pages} {237--242} (\bibinfo {year} {1963})}\BibitemShut {NoStop}%
\bibitem [{\citenamefont {Gillespie}(1973)}]{Gillespie1973}%
  \BibitemOpen
  \bibfield  {author} {\bibinfo {author} {\bibfnamefont {J.~H.}\ \bibnamefont {Gillespie}},\ }\bibfield  {title} {\enquote {\bibinfo {title} {Natural selection with varying selection coefficients --a haploid model},}\ }\href {\doibase 10.1017/S001667230001329X} {\bibfield  {journal} {\bibinfo  {journal} {Genetical Research}\ }\textbf {\bibinfo {volume} {21}},\ \bibinfo {pages} {115--120} (\bibinfo {year} {1973})}\BibitemShut {NoStop}%
\bibitem [{\citenamefont {Orr}(2007)}]{Orr2007}%
  \BibitemOpen
  \bibfield  {author} {\bibinfo {author} {\bibfnamefont {H.~A.}\ \bibnamefont {Orr}},\ }\bibfield  {title} {\enquote {\bibinfo {title} {{Absolute Fitness, Relative Fitness, and Utility}},}\ }\href {\doibase 10.1111/j.1558-5646.2007.00237.x} {\bibfield  {journal} {\bibinfo  {journal} {Evolution}\ }\textbf {\bibinfo {volume} {61}},\ \bibinfo {pages} {2997--3000} (\bibinfo {year} {2007})}\BibitemShut {NoStop}%
\bibitem [{\citenamefont {Anderson}\ \emph {et~al.}(1988)\citenamefont {Anderson}, \citenamefont {Palma},\ and\ \citenamefont {Thisse}}]{Anderson1988}%
  \BibitemOpen
  \bibfield  {author} {\bibinfo {author} {\bibfnamefont {S.~P.}\ \bibnamefont {Anderson}}, \bibinfo {author} {\bibfnamefont {A.~De}\ \bibnamefont {Palma}}, \ and\ \bibinfo {author} {\bibfnamefont {J.-F.}\ \bibnamefont {Thisse}},\ }\bibfield  {title} {\enquote {\bibinfo {title} {A representative consumer theory of the logit model},}\ }\href {http://www.jstor.org/stable/2526791} {\bibfield  {journal} {\bibinfo  {journal} {International Economic Review}\ }\textbf {\bibinfo {volume} {29}},\ \bibinfo {pages} {461--466} (\bibinfo {year} {1988})}\BibitemShut {NoStop}%
\bibitem [{\citenamefont {Jaynes}(1957{\natexlab{a}})}]{Jaynes1957a}%
  \BibitemOpen
  \bibfield  {author} {\bibinfo {author} {\bibfnamefont {E.~T.}\ \bibnamefont {Jaynes}},\ }\bibfield  {title} {\enquote {\bibinfo {title} {Information theory and statistical mechanics},}\ }\href {\doibase 10.1103/PhysRev.106.620} {\bibfield  {journal} {\bibinfo  {journal} {Phys. Rev.}\ }\textbf {\bibinfo {volume} {106}},\ \bibinfo {pages} {620--630} (\bibinfo {year} {1957}{\natexlab{a}})}\BibitemShut {NoStop}%
\bibitem [{\citenamefont {Jaynes}(1957{\natexlab{b}})}]{Jaynes1957b}%
  \BibitemOpen
  \bibfield  {author} {\bibinfo {author} {\bibfnamefont {E.~T.}\ \bibnamefont {Jaynes}},\ }\bibfield  {title} {\enquote {\bibinfo {title} {Information theory and statistical mechanics. ii},}\ }\href {\doibase 10.1103/PhysRev.108.171} {\bibfield  {journal} {\bibinfo  {journal} {Phys. Rev.}\ }\textbf {\bibinfo {volume} {108}},\ \bibinfo {pages} {171--190} (\bibinfo {year} {1957}{\natexlab{b}})}\BibitemShut {NoStop}%
\bibitem [{\citenamefont {A.~Golan}(1996)}]{Golan1996}%
  \BibitemOpen
  \bibfield  {author} {\bibinfo {author} {\bibfnamefont {D.~Miller}\ \bibnamefont {A.~Golan}, \bibfnamefont {G.~G.~Judge}},\ }\href {https://www.wiley.com/en-ca/Maximum+Entropy+Econometrics%3A+Robust+Estimation+with+Limited+Data-p-9780471953111} {\emph {\bibinfo {title} {Maximum Entropy Econometrics: Robust Estimation with Limited Data}}}\ (\bibinfo  {publisher} {John Wiley \& Sons Inc},\ \bibinfo {year} {1996})\BibitemShut {NoStop}%
\bibitem [{\citenamefont {Shipley}\ \emph {et~al.}(2006)\citenamefont {Shipley}, \citenamefont {Vile},\ and\ \citenamefont {Garnier}}]{Shipley2006}%
  \BibitemOpen
  \bibfield  {author} {\bibinfo {author} {\bibfnamefont {B.}~\bibnamefont {Shipley}}, \bibinfo {author} {\bibfnamefont {D.}~\bibnamefont {Vile}}, \ and\ \bibinfo {author} {\bibfnamefont {É.}\ \bibnamefont {Garnier}},\ }\bibfield  {title} {\enquote {\bibinfo {title} {From plant traits to plant communities: A statistical mechanistic approach to biodiversity},}\ }\href {\doibase 10.1126/science.1131344} {\bibfield  {journal} {\bibinfo  {journal} {Science}\ }\textbf {\bibinfo {volume} {314}},\ \bibinfo {pages} {812--814} (\bibinfo {year} {2006})}\BibitemShut {NoStop}%
\bibitem [{\citenamefont {Harte}(2011)}]{Harte2011}%
  \BibitemOpen
  \bibfield  {author} {\bibinfo {author} {\bibfnamefont {J.}~\bibnamefont {Harte}},\ }\href {\doibase 10.1093/acprof:oso/9780199593415.001.0001} {\emph {\bibinfo {title} {{Maximum Entropy and Ecology: A Theory of Abundance, Distribution, and Energetics}}}}\ (\bibinfo  {publisher} {Oxford University Press},\ \bibinfo {year} {2011})\BibitemShut {NoStop}%
\bibitem [{\citenamefont {Yu}\ \emph {et~al.}(2019)\citenamefont {Yu}, \citenamefont {Si}, \citenamefont {Hu},\ and\ \citenamefont {Zhang}}]{Yu2019}%
  \BibitemOpen
  \bibfield  {author} {\bibinfo {author} {\bibfnamefont {Y.}~\bibnamefont {Yu}}, \bibinfo {author} {\bibfnamefont {X.}~\bibnamefont {Si}}, \bibinfo {author} {\bibfnamefont {C.}~\bibnamefont {Hu}}, \ and\ \bibinfo {author} {\bibfnamefont {J.}~\bibnamefont {Zhang}},\ }\bibfield  {title} {\enquote {\bibinfo {title} {A review of recurrent neural networks: Lstm cells and network architectures},}\ }\href {\doibase 10.1162/neco_a_01199} {\bibfield  {journal} {\bibinfo  {journal} {Neural Computation}\ }\textbf {\bibinfo {volume} {31}},\ \bibinfo {pages} {1235--1270} (\bibinfo {year} {2019})}\BibitemShut {NoStop}%
\bibitem [{\citenamefont {Alzubaidi}\ \emph {et~al.}(2021)\citenamefont {Alzubaidi}, \citenamefont {Zhang}, \citenamefont {Humaidi}, \citenamefont {Al-Dujaili}, \citenamefont {Duan}, \citenamefont {Al-Shamma}, \citenamefont {Santamar{\'\i}a}, \citenamefont {Fadhel}, \citenamefont {Al-Amidie},\ and\ \citenamefont {Farhan}}]{Alzubaidi2021}%
  \BibitemOpen
  \bibfield  {author} {\bibinfo {author} {\bibfnamefont {L.}~\bibnamefont {Alzubaidi}}, \bibinfo {author} {\bibfnamefont {J.}~\bibnamefont {Zhang}}, \bibinfo {author} {\bibfnamefont {A.~J.}\ \bibnamefont {Humaidi}}, \bibinfo {author} {\bibfnamefont {A.}~\bibnamefont {Al-Dujaili}}, \bibinfo {author} {\bibfnamefont {Y.}~\bibnamefont {Duan}}, \bibinfo {author} {\bibfnamefont {O.}~\bibnamefont {Al-Shamma}}, \bibinfo {author} {\bibfnamefont {J.}~\bibnamefont {Santamar{\'\i}a}}, \bibinfo {author} {\bibfnamefont {M.~A.}\ \bibnamefont {Fadhel}}, \bibinfo {author} {\bibfnamefont {M.}~\bibnamefont {Al-Amidie}}, \ and\ \bibinfo {author} {\bibfnamefont {L.}~\bibnamefont {Farhan}},\ }\bibfield  {title} {\enquote {\bibinfo {title} {Review of deep learning: concepts, cnn architectures, challenges, applications, future directions},}\ }\href {\doibase 10.1186/s40537-021-00444-8} {\bibfield  {journal} {\bibinfo  {journal} {Journal of Big Data}\ }\textbf {\bibinfo {volume} {8}},\ \bibinfo {pages} {53} (\bibinfo {year}
  {2021})}\BibitemShut {NoStop}%
\bibitem [{\citenamefont {Elaydi}(2005)}]{Elaydi2005}%
  \BibitemOpen
  \bibfield  {author} {\bibinfo {author} {\bibfnamefont {S.}~\bibnamefont {Elaydi}},\ }\href {\doibase https://doi.org/10.1007/0-387-27602-5} {\emph {\bibinfo {title} {An Introduction to Difference Equations}}}\ (\bibinfo  {publisher} {Springer New York, NY},\ \bibinfo {year} {2005})\BibitemShut {NoStop}%
\bibitem [{\citenamefont {Cover}\ and\ \citenamefont {Thomas}(2005)}]{Cover2005}%
  \BibitemOpen
  \bibfield  {author} {\bibinfo {author} {\bibfnamefont {T.~M.}\ \bibnamefont {Cover}}\ and\ \bibinfo {author} {\bibfnamefont {J.~A.}\ \bibnamefont {Thomas}},\ }\href {\doibase 10.1002/047174882X} {\emph {\bibinfo {title} {Elements of Information Theory}}}\ (\bibinfo  {publisher} {John Wiley \& Sons, Ltd},\ \bibinfo {year} {2005})\BibitemShut {NoStop}%
\bibitem [{\citenamefont {van Nes}\ and\ \citenamefont {Scheffer}(2007)}]{vanNes2007}%
  \BibitemOpen
  \bibfield  {author} {\bibinfo {author} {\bibfnamefont {E.~H.}\ \bibnamefont {van Nes}}\ and\ \bibinfo {author} {\bibfnamefont {M.}~\bibnamefont {Scheffer}},\ }\bibfield  {title} {\enquote {\bibinfo {title} {Slow recovery from perturbations as a generic indicator of a nearby catastrophic shift.}}\ }\href {\doibase 10.1086/516845} {\bibfield  {journal} {\bibinfo  {journal} {The American Naturalist}\ }\textbf {\bibinfo {volume} {169}},\ \bibinfo {pages} {738--747} (\bibinfo {year} {2007})}\BibitemShut {NoStop}%
\bibitem [{\citenamefont {Scheffer}\ \emph {et~al.}(2015)\citenamefont {Scheffer}, \citenamefont {Carpenter}, \citenamefont {Dakos},\ and\ \citenamefont {van Nes}}]{Scheffer2015rev}%
  \BibitemOpen
  \bibfield  {author} {\bibinfo {author} {\bibfnamefont {M.}~\bibnamefont {Scheffer}}, \bibinfo {author} {\bibfnamefont {S.~R.}\ \bibnamefont {Carpenter}}, \bibinfo {author} {\bibfnamefont {V.}~\bibnamefont {Dakos}}, \ and\ \bibinfo {author} {\bibfnamefont {E.~H.}\ \bibnamefont {van Nes}},\ }\bibfield  {title} {\enquote {\bibinfo {title} {Generic indicators of ecological resilience: Inferring the chance of a critical transition},}\ }\href {\doibase https://doi.org/10.1146/annurev-ecolsys-112414-054242} {\bibfield  {journal} {\bibinfo  {journal} {Annual Review of Ecology, Evolution, and Systematics}\ }\textbf {\bibinfo {volume} {46}},\ \bibinfo {pages} {145--167} (\bibinfo {year} {2015})}\BibitemShut {NoStop}%
\bibitem [{\citenamefont {Wissel}(1984)}]{Wissel1984}%
  \BibitemOpen
  \bibfield  {author} {\bibinfo {author} {\bibfnamefont {C.}~\bibnamefont {Wissel}},\ }\bibfield  {title} {\enquote {\bibinfo {title} {A universal law of the characteristic return time near thresholds},}\ }\href {\doibase 10.1007/BF00384470} {\bibfield  {journal} {\bibinfo  {journal} {Oecologia}\ }\textbf {\bibinfo {volume} {65}},\ \bibinfo {pages} {101--107} (\bibinfo {year} {1984})}\BibitemShut {NoStop}%
\bibitem [{\citenamefont {Leonel}(2016)}]{Leonel2016}%
  \BibitemOpen
  \bibfield  {author} {\bibinfo {author} {\bibfnamefont {E.~D.}\ \bibnamefont {Leonel}},\ }\bibfield  {title} {\enquote {\bibinfo {title} {Defining universality classes for three different local bifurcations},}\ }\href {\doibase https://doi.org/10.1016/j.cnsns.2016.04.008} {\bibfield  {journal} {\bibinfo  {journal} {Communications in Nonlinear Science and Numerical Simulation}\ }\textbf {\bibinfo {volume} {39}},\ \bibinfo {pages} {520--528} (\bibinfo {year} {2016})}\BibitemShut {NoStop}%
\bibitem [{\citenamefont {Kuehn}(2011)}]{Kuehn2011}%
  \BibitemOpen
  \bibfield  {author} {\bibinfo {author} {\bibfnamefont {C.}~\bibnamefont {Kuehn}},\ }\bibfield  {title} {\enquote {\bibinfo {title} {A mathematical framework for critical transitions: Bifurcations, fast–slow systems and stochastic dynamics},}\ }\href {\doibase https://doi.org/10.1016/j.physd.2011.02.012} {\bibfield  {journal} {\bibinfo  {journal} {Physica D: Nonlinear Phenomena}\ }\textbf {\bibinfo {volume} {240}},\ \bibinfo {pages} {1020--1035} (\bibinfo {year} {2011})}\BibitemShut {NoStop}%
\bibitem [{\citenamefont {Strogatz}(2000)}]{Strogatz2000}%
  \BibitemOpen
  \bibfield  {author} {\bibinfo {author} {\bibfnamefont {S.~H.}\ \bibnamefont {Strogatz}},\ }\href {\doibase https://doi.org/10.1201/9780429398490} {\emph {\bibinfo {title} {Nonlinear Dynamics and Chaos}}}\ (\bibinfo  {publisher} {Chapman and Hall/CRC, Boca Raton},\ \bibinfo {year} {2000})\BibitemShut {NoStop}%
\bibitem [{\citenamefont {Fontich}\ \emph {et~al.}(2022)\citenamefont {Fontich}, \citenamefont {Guillamon}, \citenamefont {Lázaro}, \citenamefont {Alarcón}, \citenamefont {Vidiella},\ and\ \citenamefont {Sardanyés}}]{Fontich2022}%
  \BibitemOpen
  \bibfield  {author} {\bibinfo {author} {\bibfnamefont {E.}~\bibnamefont {Fontich}}, \bibinfo {author} {\bibfnamefont {A.}~\bibnamefont {Guillamon}}, \bibinfo {author} {\bibfnamefont {J.~T.}\ \bibnamefont {Lázaro}}, \bibinfo {author} {\bibfnamefont {T.}~\bibnamefont {Alarcón}}, \bibinfo {author} {\bibfnamefont {B.}~\bibnamefont {Vidiella}}, \ and\ \bibinfo {author} {\bibfnamefont {J.}~\bibnamefont {Sardanyés}},\ }\bibfield  {title} {\enquote {\bibinfo {title} {Critical slowing down close to a global bifurcation of a curve of quasi-neutral equilibria},}\ }\href {\doibase https://doi.org/10.1016/j.cnsns.2021.106032} {\bibfield  {journal} {\bibinfo  {journal} {Communications in Nonlinear Science and Numerical Simulation}\ }\textbf {\bibinfo {volume} {104}},\ \bibinfo {pages} {106032} (\bibinfo {year} {2022})}\BibitemShut {NoStop}%
\bibitem [{\citenamefont {Haas}\ \emph {et~al.}(2022)\citenamefont {Haas}, \citenamefont {Gutierrez}, \citenamefont {Oliveira},\ and\ \citenamefont {Goldstein}}]{haas2022}%
  \BibitemOpen
  \bibfield  {author} {\bibinfo {author} {\bibfnamefont {P.A.}\ \bibnamefont {Haas}}, \bibinfo {author} {\bibfnamefont {M.A.}\ \bibnamefont {Gutierrez}}, \bibinfo {author} {\bibfnamefont {N.M.}\ \bibnamefont {Oliveira}}, \ and\ \bibinfo {author} {\bibfnamefont {R.E.}\ \bibnamefont {Goldstein}},\ }\bibfield  {title} {\enquote {\bibinfo {title} {Stabilization of microbial communities by responsive phenotypic switching},}\ }\href {\doibase 10.1103/PhysRevResearch.4.033224} {\bibfield  {journal} {\bibinfo  {journal} {Physical Review Research}\ }\textbf {\bibinfo {volume} {4}},\ \bibinfo {pages} {033224} (\bibinfo {year} {2022})}\BibitemShut {NoStop}%
\bibitem [{\citenamefont {Weiss}(1907)}]{Weiss1907}%
  \BibitemOpen
  \bibfield  {author} {\bibinfo {author} {\bibfnamefont {P.}~\bibnamefont {Weiss}},\ }\bibfield  {title} {\enquote {\bibinfo {title} {{The hypothesis of the molecular field and the ferromagnetic property}},}\ }\href {\doibase 10.1051/jphystap:019070060066100} {\bibfield  {journal} {\bibinfo  {journal} {{J. Phys. Theor. Appl.}}\ }\textbf {\bibinfo {volume} {6}},\ \bibinfo {pages} {661--690} (\bibinfo {year} {1907})}\BibitemShut {NoStop}%
\bibitem [{\citenamefont {Kadanoff}(2009)}]{Kadanoff2009}%
  \BibitemOpen
  \bibfield  {author} {\bibinfo {author} {\bibfnamefont {L.~P.}\ \bibnamefont {Kadanoff}},\ }\bibfield  {title} {\enquote {\bibinfo {title} {More is the same; phase transitions and mean field theories},}\ }\href {\doibase 10.1007/s10955-009-9814-1} {\bibfield  {journal} {\bibinfo  {journal} {Journal of Statistical Physics}\ }\textbf {\bibinfo {volume} {137}},\ \bibinfo {pages} {777--797} (\bibinfo {year} {2009})}\BibitemShut {NoStop}%
\bibitem [{\citenamefont {Nishimori}\ and\ \citenamefont {Ortiz}(2010)}]{Nishimori2010}%
  \BibitemOpen
  \bibfield  {author} {\bibinfo {author} {\bibfnamefont {H.}~\bibnamefont {Nishimori}}\ and\ \bibinfo {author} {\bibfnamefont {G.}~\bibnamefont {Ortiz}},\ }\href {\doibase 10.1093/acprof:oso/9780199577224.001.0001} {\emph {\bibinfo {title} {{Elements of Phase Transitions and Critical Phenomena}}}}\ (\bibinfo  {publisher} {Oxford University Press},\ \bibinfo {year} {2010})\BibitemShut {NoStop}%
\bibitem [{\citenamefont {Boettiger}\ and\ \citenamefont {Batt}(2020)}]{Boettiger2020}%
  \BibitemOpen
  \bibfield  {author} {\bibinfo {author} {\bibfnamefont {C.}~\bibnamefont {Boettiger}}\ and\ \bibinfo {author} {\bibfnamefont {R.}~\bibnamefont {Batt}},\ }\bibfield  {title} {\enquote {\bibinfo {title} {Bifurcation or state tipping: assessing transition type in a model trophic cascade},}\ }\href {\doibase 10.1007/s00285-019-01358-z} {\bibfield  {journal} {\bibinfo  {journal} {Journal of Mathematical Biology}\ }\textbf {\bibinfo {volume} {80}},\ \bibinfo {pages} {143--155} (\bibinfo {year} {2020})}\BibitemShut {NoStop}%
\end{thebibliography}%
\newpage

 \appendix


\section{list of symbols}
\label{glossary}

Table~\ref{tblr} presents key symbols used in this work and their description.

\NewTblrTheme{fancy}{
  \SetTblrStyle{firsthead}{font=\bfseries}
  \SetTblrStyle{firstfoot}{fg=blue2}
  \SetTblrStyle{middlefoot}{fg=blue2}
}
\begin{longtblr}[
  theme = fancy,
  caption = {Glossary of symbols used},
  label = {tblr},
]{
  colspec={|X[c] X[3]|}, width = \linewidth, hlines,
  rowhead = 1, rowfoot = 0,
  row{1} = {mycolour1!15,c}, row{2-Z} = {mycolour1!4},
}
 {\bf Symbol} & {\bf Definition}\\
 \hline
 $\mathsf{A}$& Subpopulation that exhibits phenotypic plasticity as an adaptive response to environmental
variations. \\
 \hline
 $\mathsf{B}$   & Bet-hedging subpopulation that is preadapted corresponding to various environmental states.    \\
 \hline
 $\mathsf{C}$   & Subpopulation that is indifferent to environmental variations.    \\
 \hline
 $\mathsf{S}$& The set consisting of all possible environmental states, which is equal to $\{1,\cdots,S\}.$ \\
 \hline
 $q_s$   & Probability that the environment is in state $s\in\mathsf{S}$.    \\
 \hline
 $\bm{q}$   & Probability vector, which is equal to $\left(q_1,\cdots,q_S\right)^T$.    \\
 \hline
 $R$& Growth rate, which is independent of the environmental states, of subpopulation $\mathsf{C}$. \\
 \hline
 $\alpha_s$   & Time independent stoschastic contribution to the fitness of type $\mathsf{A}$ subpopulation when the environment is in state $s$.     \\
 \hline
 $\overbar{\alpha}$   & Mean of random variable $\alpha$, which is given by $\sum_{s=1}^S\alpha_sq_s$.    \\
 \hline
 $P$   & Total population of types $\mathsf{A}$ and $\mathsf{B}$.   \\
 \hline
 $\fitBold{}{t}$  & Fitness of an individual from type $\mathsf{A}$ subpopulation at time $t$.    \\
 \hline
 $\fitB{s}{t+1}$& Fitness of an individual from type $\mathsf{A}$ subpopulation at time $t+1$ when the environment is in state $s$.   \\
 \hline
 $n$& Total number of bet-hedgers.   \\
 \hline
 $\fitCs{i}{t+1}$   & Fitness of an individual from $i^{\text{th}}$ type bet-hedging strategy inside type $\mathsf{B}$ subpopulation at time $t+1$ when the environment is in state $s$.     \\
 \hline
 $\mathsf{H}(t+1)$   & Total mean-variance fitness of types $\mathsf{A}$ and $\mathsf{B}$ subpopulations at time $t+1$. This is an approximation for the geometric mean fitness.    \\
 \hline
 $V_n$& $(n+1)\times (n+1)$ fitness correlation matrix for type $\mathsf{A}$ subpopulation and type $\mathsf{B}$ subpopulation consisting of $n$ bet-hedgers. \\
 \hline
 $\left(V_{n}^{-1}\right)_{00}$ & The first element (00 matrix component) of the inverse of $V_n$. \\
 \hline
 $\vec{Z}(t)$   & $(n+1)\times 1$ population vector at time $t$ comprising total population of type $\mathsf{A}$ subpopulation and type $\mathsf{B}$ subpopulation that is consisting of $n$ bet-hedgers.    \\
 \hline
 $\vec{Z}_{\text{opt}}(t)$  & Optimal population vector at time $t$ that maximizes the mean-variance fitness $\mathsf{H}(t+1)$ at time $t+1$.     \\
 \hline
 $A_*$& Intrinsic fitness of an individual inside subpopulation $\mathsf{A}$ when the whole population $P$ is of type $\mathsf{A}$ and the mean-variance fitness is maximum. \\
 \hline
 $\bm{B}_{*}$   & $n\times 1$ vector of intrinsic fitnesses of individuals following bet-hedging strategies inside type $\mathsf{B}$ subpopulation when the whole population $P$ is of type $\mathsf{A}$ and the mean-variance fitness is maximum.    \\
 \hline
 $X(t)$   & Deviation from intrinsic fitness of a type $\mathsf{A}$ individual at time $t$, which is given by $X(t)=A^{(0)}(t)-A_{*}$.    \\
 \hline
 $\bm{Y}(t)$& Vector of the deviations from intrinsic fitnesses of individuals following bet-hedging strategies inside type $\mathsf{B}$ subpopulation at time $t$, which is given by $\bm{Y}(t)=\bm{B}(t)-\bm{B}_{*}$. \\
 \hline
$X_{\mu}(t+1)$  & Deviation from intrinsic fitness of a type $\mathsf{A}$ individual with phenotype $\mu$ at time $t+1$   \\
 \hline
 $h_{\mu}^{(t)}$   & Function that dictates the dependence of fitness deviation $X_{\mu}(t+1)$ of an individual with phenotype $\mu$ at time $t+1$ on past fitness deviations of the same individual.  \\
 \hline
 $\vec{Z}_{\mu}(t)$  & Optimal population vector consisting of $\mu^{\text{th}}$ phenotypic population and bet-hedgers.    \\
 \hline
 $n_{\mu}(t)$& Fraction of total population $P$ at time $t$ following $\mu^{\text{th}}$ phenotypic plasticity such that $\sum_{\mu}n_{\mu}(t)=1$. \\
 \hline
 $f_{\mu}(t-1)$   & An effective measure of fitness at time $(t-1)$ for $\mu^{\text{th}}$ phenotype that decides the fraction $n_{\mu}(t)$ of population following  $\mu^{\text{th}}$ phenotypic plasticity.  \\
 \hline
 $\beta$  & A strength parameter that decides how strongly the fraction  $n_{\mu}(t)$ depends on $f_{\mu}(t-1)$ via the relation $n_{\mu}(t)= \exp\left[\beta f_{\mu}(t-1)\right]/\left(\sum_{\mu}\exp\left[\beta f_{\mu}(t-1)\right]\right)$.   \\
 \hline
 $\widetilde{\beta}(n)$   & Scaled strength parameter defined as $\widetilde{\beta} = \frac{1}{2}\left(V_{n}^{-1}\right)_{00}\beta$    \\
 \hline
 $\beta_{\text{crit}} (n)$  & Critical value of the strength parameter $\beta$, in the presence of $n$ bet-hedgers, where the stability properties of the fixed points of the fitness dynamics change. At times where $n$ dependence is clear, we omit explicit $n$ dependence and write just $\beta_{\text{crit}}$.    \\
 \hline
 $\bone(n)$   & Critical value of the strength parameter $\beta$ where a stable fixed point becomes an unstable fixed point of the fitness dynamics when approached from higher than $\bone$ value of $\beta$ in the presence of $n$ bet-hedgers.    \\
 \hline
 $\btwo(n)$& Critical value of the strength parameter $\beta$ where a stable fixed point becomes an unstable fixed point of the fitness dynamics when approached from lower than $\btwo$ value of $\beta$ in the presence of $n$ bet-hedgers.  \\
 \hline
 $x^*$   & A fixed point of the fitness dynamics.     \\
 \hline
 $x_{\text{crit}}(n)$   & Critical value of the deviation from intrinsic fitness of a type $\mathsf{A}$ individual obtained at $\beta_{\text{crit}}(n)$ in the presence of $n$ bet-hedgers.    \\
 \hline
 $\tau_r$   & Recovery time taken to reach the fixed point of the fitness dynamics when dynamics starts near the said fixed point.    \\
 \hline
\end{longtblr}

\newpage

\section{Fitness dynamics for large number of environmental states}
\label{append:more-bif-fig}
\begin{figure}[hbt!]
\centering
\subfloat[]{
{\includegraphics[width=0.45\textwidth]{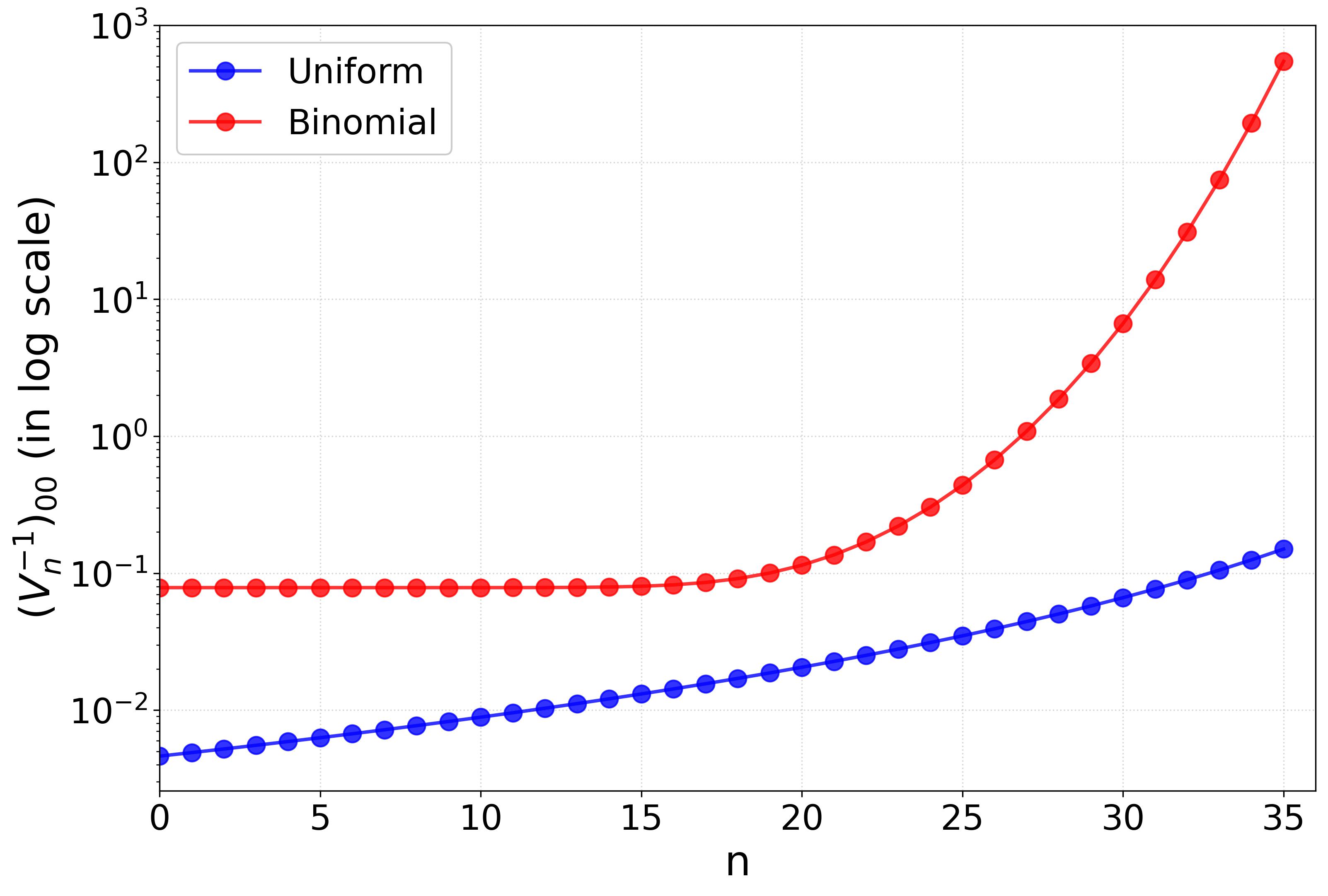}\label{fg.n_vs_v_inv_large}}
}\\
\subfloat[]{
{\includegraphics[width=0.45\textwidth]
{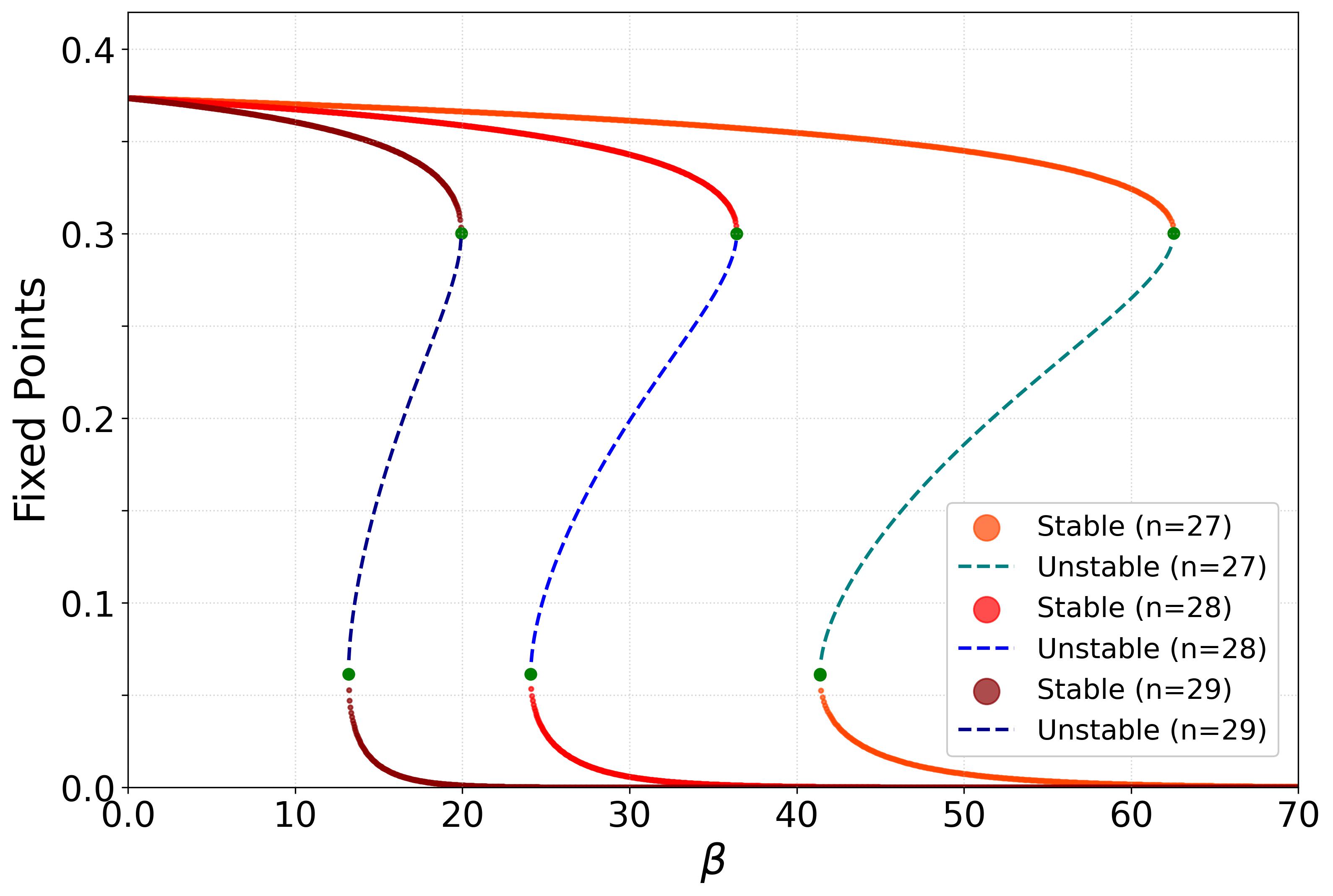}\label{fg.n_comparison}}
}
\subfloat[]{
{\includegraphics[width=0.45\textwidth]
{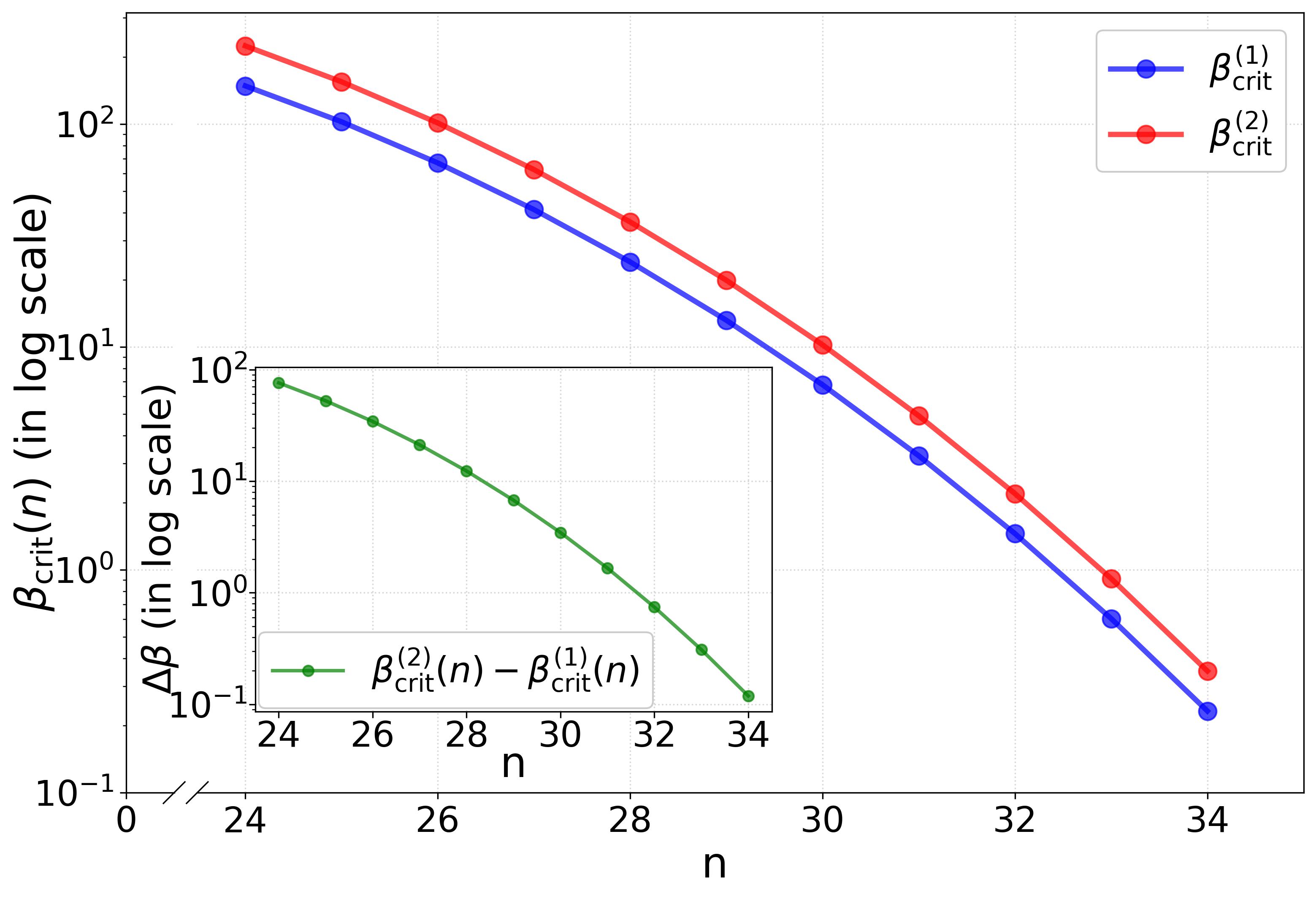}\label{fg.n_vs_beta_large}}
}

\caption{The figure illustrates the dynamical behavior of our system when $S$ is large, here $S=51$. The parameters are fixed as in Eq.~\eqref{eq:fix-params}. Fig.~(\ref{fg.n_vs_v_inv_large}) shows variation of $\left(V_{n}^{-1}\right)_{00}$ (in log scale) as a function of $n$ for $S=51$. We observe that $\left(V_{n}^{-1}\right)_{00}$ is a monotonically increasing function of $n$, where blue curve corresponds to the case where all the environmental states are distributed according to distribution $\unif$ given by Eq.~\eqref{eq:unif-dist} and the red curve corresponds to the case where the environmental states are distributed according to distribution $\Binom$ given by Eq.~\eqref{eq:binom-dist}. Fig.~(\ref{fg.n_comparison}) shows three bifurcation diagrams for $n=27$, $n=28$, and $n=29$. Each bifurcation diagram shows two tipping points; F1 denoted by $\left(\bone(n),x^{(1)}_{\text{crit}}(n)\right)$ and F2 denoted by $\left(\btwo(n),x^{(2)}_{\text{crit}}(n)\right)$. Two tipping points F1 and F2 occur, respectively at $(41.37, 0.061)$ and $(62.55, 0.300)$ for $n=27$, at $(24.08, 0.061)$ and $(36.40, 0.300)$ for $n=28$, and at $(13.18, 0.061)$ and $(19.92, 0.300)$ for $n=29$. For both tipping points $\beta_{\text{crit}}{(29)}<\beta_{\text{crit}}{(28)}<\beta_{\text{crit}}{(27)}$. The critical values $\bone(n)$ and $\btwo(n)$ are plotted (in log scale) as a function of $n$ in Fig.~(\ref{fg.n_vs_beta_large}), which demonstrates that $\bone(n)$ and $\btwo(n)$ decline monotonically as we increase $n$. The plot of $\btwo(n)-\bone(n)$, defined as $\delta\beta$, in log scale, as a function of $n$ is shown in the inset, and it also decreases as $n$ increases. Here, all calculations are performed with a $10^{-8}$ tolerance. Fig.~(\ref{fg.extra}) illustrates that the dynamical behavior of the system exhibits consistent trends even for larger values of $S$. This demonstrates the general applicability and robustness of the model across a broader parameter range.}
\label{fg.extra}
\end{figure}

\newpage
\section{Geometric mean and arithmetic mean }
 \label{append:geo-ari}
 Let $X$ be a discrete random variable which takes values $x_i$ with probability $p_i$ for $i=1,\cdots,N$. The arithmetic mean $\mu$ and the geometric mean $g$ of $X$ are defined as
 \begin{align}
     &\mu = \mathbb{E}[X]=\sum_{i=1}^N x_i p_i \\
     &g= \exp\left[\mathbb{E}[\ln X]\right] = \prod_{i=1}^N x_i^{p_i}.
 \end{align}
Now define $Y=\frac{X-\mu}{\mu}$. Note that $\mathbb{E}[Y]=0$ and $\mathbb{E}\left[Y^2\right]=\text{var}/\mu^2$, where $\text{var}$ is the variance of $X$. Then
\begin{align*}
     g&= \exp\left[\ln \mu + \mathbb{E}[\ln (1+Y)]\right]\\
     &= \mu \exp\left[\mathbb{E}\left[Y-\frac{Y^2}{2}+O(Y^3)\right]\right]\\
     &= \mu \exp\left[- \frac{\text{var}}{2\mu^2} + \mathbb{E}\left[O(Y^3)\right]\right]\\
     &\approx \mu \left[1- \frac{\text{var}}{2\mu^2} \right]\\
     &=\mu - \frac{\text{var}}{2\mu}
 \end{align*}
Thus if the central moments $\mathbb{E}\left[\left(X-\mu\right)^n\right]$ of $X$ are small for $n=3,4,\cdots$ then geometric mean $g$ can be approximated as $g\approx \mu - \frac{\text{var}}{2\mu}$.
\\
\\

 \section{Mean-variance fitness}
 \label{append:var-matrix}
 We have total fitness as
 \begin{align}
    G_s(t+1)=\innerp*{\begin{pmatrix}
        -R\fitBold{}{t} + \fitB{s}{t+1}\\
        -R \bm{B}(t) + \bm{\mathcal{B}}_s(t+1)
    \end{pmatrix},~\vec{Z}(t)
    },
\end{align}
The conservative mean-variance fitness $\mathsf{H}(t+1)$ of the population then is defined as follows.
 \begin{align}
     \mathsf{H}(t+1) = \mathbb{E}\left[G(t+1)\right] - \frac{1}{2}\mathbb{E}\left[\left(G(t+1)-\mathbb{E}\left[G(t+1)\right]\right)^2\right],
 \end{align}
Then
 \begin{align*}
    \mathbb{E}\left[G(t+1)\right] &= \sum_{s=1}^n q_s G_s(t+1)\\
    &=\innerp*{\begin{pmatrix}
        -R\fitBold{}{t} + \mathbb{E}\left[\fitB{}{t+1}\right]\\
        -R \bm{B}(t) + \bm{q}
    \end{pmatrix},~\vec{Z}(t)
    },
 \end{align*}
where $\bm{q}=\left(q_1,\cdots,q_n\right)^T$. Let $\mathcal{A}^{(0)}(t+1)$ be a random variable that takes values $\{\fitB{s}{t+1}\}$ with probability $\{q_s\}$ and $\vec{\mathcal{B}}(t+1)$ be a vector valued random variable that takes values $\vec{\mathcal{B}}_s(t+1)$ with probability $\{q_s\}$.
 \begin{align*}
    &\mathbb{E}\left[\left(G(t+1)-\mathbb{E}\left[G(t+1)\right]\right)^2\right]\\
    &=\sum_{s=1}^n q_s \left(G_s(t+1)-\mathbb{E}\left[G(t+1)\right]\right)^2\\
    &=\sum_{s=1}^nq_s\innerp*{\begin{pmatrix}
       \fitB{s}{t+1} - \mathbb{E}\left[\fitB{}{t+1}\right]\\
        \bm{\mathcal{B}}_s(t+1)- \bm{q}
    \end{pmatrix},~\vec{Z}(t)
    }^2\\
    &=\vec{Z}^T(t)\begin{pmatrix}
       v & \bm{u}^T\\
        \bm{u} & \Gamma
    \end{pmatrix}~\vec{Z}(t),
 \end{align*}
 where $v=\mathbb{E}\left[\left(\fitB{}{t+1} - \mathbb{E}\left[\fitB {}{t+1}\right]\right)^2\right]$, $\bm{u}=\mathbb{E}\left[\left(\fitB{}{t+1} - \mathbb{E}\left[\fitB {}{t+1}\right]\right)\left(\bm{\mathcal{B}}(t+1)- \bm{q}\right)\right]$ and $\Gamma=\mathbb{E}\left[\left(\bm{\mathcal{B}}(t+1)- \bm{q}\right)\left(\bm{\mathcal{B}}(t+1)- \bm{q}\right)^T\right]$. More specifically, we have
 \begin{align}
     v&=\sum_{s=1}^S q_s \left(\alpha_s-\mathbb{E}[\alpha]\right)^2\nonumber\\
     &=\sum_{s=1}^Sq_s \left(\alpha_s-\overbar{\alpha}\right)^2,
 \end{align}
 where $\overbar{\alpha}=\mathbb{E}[\alpha]$. We have
 \begin{align}
     u_i=(\bm{u})_i&=\sum_{s=1}^S q_s \left(\alpha_s-\overbar{\alpha}\right)\left(\fitC{i}{t+1}- q_i\right)\nonumber\\
     &=\sum_{s=1}^S q_s \left(\alpha_s-\overbar{\alpha}\right)\left(\delta_{is}- q_i\right)\nonumber\\ 
     &=q_i\left(\alpha_i-\overbar{\alpha}\right)-q_i\sum_{s=1}^S q_s \left(\alpha_s-\overbar{\alpha}\right)\nonumber\\ 
      &=q_i\left(\alpha_i-\overbar{\alpha}\right).
 \end{align}
 Finally,
 \begin{align}
     \Gamma_{ij}&=\sum_{s=1}^S q_s \left(\fitC{i}{t+1}- q_i\right)\left(\fitC{j}{t+1}- q_j\right)\nonumber\\
     &=\sum_{s=1}^Sq_s \left(\delta_{is}- q_i\right)\left(\delta_{js}- q_j\right)\nonumber\\
     &=q_i \left(\delta_{ij}- q_j\right)-\sum_{s=1}^Sq_s  q_i\left(\delta_{js}- q_j\right)\nonumber\\ 
     &=q_i \left(\delta_{ij}- q_j\right).
 \end{align}
Combining the above equations, we get
\begin{align}
     \mathsf{H}(t+1) = \innerp*{\begin{pmatrix}
        -R\fitBold{}{t} + \mathbb{E}\left[\fitB{}{t+1}\right]\\
        -R \bm{B}(t) + \bm{q}
    \end{pmatrix},~\vec{Z}(t)
    }-\frac{1}{2}\innerp*{\vec{Z}(t), V_{n}\vec{Z}(t)},
 \end{align}
 The optimal population vector $\vec{Z}_{\text{opt}}$ that maximizes $\mathsf{H}(t+1)$ is obtained as solution of the following equations.
 \begin{align}
    &\frac{\partial \mathsf{H}(t+1)}{\partial z^{(0)}(t)}\Bigg|_{\vec{Z}_{\text{opt}}(t)}=-R\fitBold{}{t} + \mathbb{E}\left[\fitB{}{t+1}\right]-\left(V_{n} \vec{Z}_{\text{opt}}(t)\right)_0=0\\
    &\frac{\partial \mathsf{H}(t+1)}{\partial z^{(i)}(t)}\Bigg|_{\vec{Z}_{\text{opt}}(t)}=-R B^{(i)}(t) + q_i - \left(V_{n}\vec{Z}_{\text{opt}}(t)\right)_i=0~~~\forall i\in\{1,\cdots,n\}.
 \end{align}
Thus,
\begin{align}
    \vec{Z}_{\text{opt}}(t)=V_{n}^{-1}\begin{pmatrix}
        -R\fitBold{}{t} + \mathbb{E}\left[\fitB{}{t+1}\right]\vspace{0.2cm}\\
        -R \bm{B}(t) + \bm{q}.
    \end{pmatrix}
 \end{align}
\\
\\


\section{Proof of Lemma~\ref{lem:increasing-bet-hedgers}}
\label{append:proof of lemma 1}
We first prove following theorem.
\begin{theorem}
\label{thm:increasing-n}
    Let $V_{n}$ be an $(n+1)\times (n+1)$ real positive definite matrix, $\vec{\psi}$ be an $(n+1)\times 1$ real vector, $\alpha$ be a real number, and $V_{n+1}$ be the  $(n+2)\times (n+2)$ real positive definite matrix defined as
    \begin{align}
        V_{n+1}:=\begin{pmatrix}
            V_{n} & \vec{\psi}\\
            \vec{\psi}^T & \alpha
        \end{pmatrix}.
    \end{align}
    Then,
    \begin{align}
    \label{eq:00relation}
        \innerp*{w', V_{n+1}^{-1}w'}:=\innerp*{\vec{w},V_{n}^{-1}\vec{w}}+\frac{\left(w_{n+2}-\innerp*{\vec{\psi}, V_{n}^{-1}\vec{w}}\right)^2}{\alpha-\innerp*{\vec{\psi},V_{n}^{-1}\vec{\psi}}}
    \end{align}
    where $w'=\begin{pmatrix}
        \vec{w}\\ w_{n+2}
    \end{pmatrix}$ with $\vec{w}$ being an $(n+1)\times 1$ vector and $w_{n+2}$ being a real number. Further,
    \begin{align}
    \label{eq:det-relation}
        \mathrm{det}\left(V_{n+1}\right)=\left(\alpha-\innerp*{\vec{\psi},V_{n}^{-1}\vec{\psi}}\right)\mathrm{det}\left(V_{n}\right).
    \end{align}
\end{theorem}

\begin{proof}
    Define $\left(\vec{\xi}^T,\xi_{n+2}\right)^T=\xi':=\left(V_{n+1}\right)^{-1}w'$. In other words
    \begin{align*}
        w'&=V_{n+1}\xi'\\
        &=\begin{pmatrix}
            V_{n} & \vec{\psi}\\
            \vec{\psi}^T & \alpha
        \end{pmatrix}\begin{pmatrix}
            \vec{\xi}\\
            \xi_{n+2}
        \end{pmatrix}\\
        &=\begin{pmatrix}
            V_{n}\vec{\xi} + \xi_{n+2} \vec{\psi}\vspace{0.2cm}\\
            \vec{\psi}^T \vec{\xi}+ \xi_{n+2} \alpha
        \end{pmatrix}.
    \end{align*}
    This implies that
    \begin{align}
        &\vec{w}= V_{n}\vec{\xi} + \xi_{n+2} \vec{\psi};\label{eq:B1}\\
        &w_{n+2} = \vec{\psi}^T \vec{\xi}+ \xi_{n+2} \alpha.\label{eq:B2}
    \end{align}
   From Eqs.~\eqref{eq:B1} and \eqref{eq:B2}, we get
    \begin{align}
    \label{eq:vec-comp}
        &\vec{\xi} =V^{-1}_{n}\vec{w}-\left[\frac{w_{n+2}-\innerp*{\vec{\psi},V_{n}^{-1}\vec{w}}}{\alpha-\innerp*{\vec{\psi},V_{n}^{-1}\vec{\psi}}}\right] V^{-1}_{n}\vec{\psi};\\
        &\xi_{n+2}=\frac{w_{n+2}-\innerp*{\vec{\psi},V_{n}^{-1}\vec{w}}}{\alpha-\innerp*{\vec{\psi},V_{n}^{-1}\vec{\psi}}}
    \end{align}
Now, we have
\begin{align*}
        \innerp*{w',V_{n+1}^{-1}w'} &= \innerp*{w',\xi'}\\
        &= \innerp*{\vec{w},\vec{\xi}}+w_{n+2}\xi_{n+2}\\
        &= \innerp*{\vec{w},V^{-1}_{n}\vec{w}}-\left[\frac{w_{n+2}-\innerp*{\vec{\psi},V_{n}^{-1}\vec{w}}}{\alpha-\innerp*{\vec{\psi},V_{n}^{-1}\vec{\psi}}}\right]\innerp*{\vec{w},V^{-1}_{n}\vec{\psi}} +w_{n+2} \left[\frac{w_{n+2}-\innerp*{\vec{\psi},V_{n}^{-1}\vec{w}}}{\alpha-\innerp*{\vec{\psi},V_{n}^{-1}\vec{\psi}}}\right]\\
        &= \innerp*{\vec{w},V^{-1}_{n}\vec{w}}+\frac{\left(w_{n+2}-\innerp*{\vec{\psi},V_{n}^{-1}\vec{w}}\right)^2}{\alpha-\innerp*{\vec{\psi},V_{n}^{-1}\vec{\psi}}}\innerp*{\vec{w},V^{-1}_{n+1}\vec{\psi}}.
    \end{align*}
This completes the proof of the first part of the theorem.

Now, since $V_{n}$ is a positive definite matrix, it can be diagonalized by an orthogonal matrix. Let $O$ be such an orthogonal  $(n+1)\times (n+1)$ matrix ($OO^T=\mathbb{I}$), i.e., $OV_{n}O^T=D_{n+1}:=\mathrm{diag}\left(d_0,\cdots, d_{n+1}\right)$. Further, let $O'=\begin{pmatrix}O & 0\\ 0 & 1\end{pmatrix}$ be an $(n+2)\times (n+2)$ orthogonal matrix. Then,
\begin{align}
    O'V_{n+1}O'^T&=\begin{pmatrix}
        O V_{n} O^T & O\vec{\psi} \vspace{0.2cm}\\
        \vec{\psi}^T O^T & \alpha
    \end{pmatrix}\nonumber\\
    &=\begin{pmatrix}
        D_{n+1} & \vec{\phi}\vspace{0.2cm}\\
        \vec{\phi}^T & \alpha
    \end{pmatrix},
\end{align}
where $O\vec{\psi}=\vec{\phi}:=\left(\phi_0,\cdots,\phi_{n}\right)^T$. Then
\begin{align*}
    \mathrm{det}\left(V_{n+1}\right) &=\mathrm{det}\left(O'V_{n+1}O'^T\right)\\
    &=\mathrm{det}\begin{pmatrix}
        D_{n+1} & \vec{\phi} \vspace{0.2cm}\\
        \vec{\phi}^T & \alpha
    \end{pmatrix}\\
    &=\mathrm{det}\begin{pmatrix}
       \alpha & \vec{\phi}^T \vspace{0.2cm}\\
       \vec{\phi}  & D_{n+1} 
    \end{pmatrix}\\
    &=\alpha ~\mathrm{det}\left(D_{n+1}\right) - \phi_0~\mathrm{det}\begin{pmatrix}
        \phi_0 & 0 & 0 &\cdots &0\\
        \phi_1& 0  & d_1&\cdots &0\\
        \vdots&\vdots & \vdots & \ddots & \vdots\\
        \phi_n & 0 & 0 & \cdots &d_n
    \end{pmatrix}+\cdots +(-1)^{n+1}\phi_n~\mathrm{det}\begin{pmatrix}
        \phi_0 & d_0 & 0 &\cdots &0\\
        \phi_1& 0  & d_1&\cdots &0\\
        \vdots&\vdots & \vdots & \ddots & \vdots\\
        \phi_n & 0 & 0 & \cdots &0
    \end{pmatrix}\\
    &=\alpha \mathrm{det}\left(D_{n+1}\right) - \frac{\phi_0^2}{d_0}\mathrm{det}\left(D_{n+1}\right)-\cdots- \frac{\phi_n^2}{d_n}\mathrm{det}\left(D_{n+1}\right)\\
     &=\left(\alpha - \sum_{i=0}^n\frac{\phi_i^2}{d_i}\right)\mathrm{det}\left(D_{n+1}\right)\\
     &=\left(\alpha - \vec{\phi}^T D_{n+1}^{-1}\vec{\phi}\right)\mathrm{det}\left(D_{n+1}\right)=\left(\alpha - \innerp*{\vec{\psi}, V_{n}^{-1}\vec{\psi}}\right)\mathrm{det}\left(V_{n}\right).
\end{align*}
This completes the second part of the theorem.
\end{proof}

The proof of Lemma~\ref{lem:increasing-bet-hedgers} directly follows from Theorem~\ref{thm:increasing-n} by noting the following facts:
\begin{enumerate}
    \item Take $w'=(1,0,\cdots,0,0)^T$. This implies $\vec{w}=(1,0,\cdots,0)^T$ and $w_{n+2}=0$. Then from Eq.~\eqref{eq:00relation}, we have
    \begin{align}
        \left(V^{-1}_{n+1}\right)_{00} = \left(V^{-1}_{n}\right)_{00}+ \frac{\innerp*{\vec{\psi}, V_{n}^{-1}\vec{w}}^2}{\alpha-\innerp*{\vec{\psi},V_{n}^{-1}\vec{\psi}}}.
    \end{align}
    \item Since $\mathrm{det}\left(V_{n}\right)>0$ and $\mathrm{det}\left(V_{n+1}\right)>0$, from Eq.~\eqref{eq:det-relation} we have
    \begin{align}
       \left(\alpha-\innerp*{\vec{\psi},V_{n}^{-1}\vec{\psi}}\right)>0.
    \end{align}
\end{enumerate}
Combining above two points we conclude that 
\begin{align}
        \left(V^{-1}_{n+1}\right)_{00} > \left(V^{-1}_{n}\right)_{00}.
    \end{align}

\end{document}